\documentclass[twoside]{article}[11pt]
\usepackage[colorlinks,linkcolor=violet,citecolor=violet]{hyperref}
\usepackage{amssymb,amsmath,theorem,bbm}
\usepackage{url}
\usepackage{leftidx,MnSymbol}
\usepackage{fancyhdr,datetime}
\usepackage{fullpage}
\usepackage{graphicx,subfigure} 
\usepackage{tikz}
\usepackage{relsize}
\usepackage{upgreek}
\usepackage[switch]{lineno}

\numberwithin{equation}{section}
\numberwithin{figure}{section}

\title{
Sparse Interpolation in Terms of Multivariate Chebyshev Polynomials  
}

\author{Evelyne Hubert
         \thanks{INRIA M\'editerran\'ee, 06902 Sophia Antipolis, France.
           \texttt{evelyne.hubert@inria.fr}} \and
Michael F.~Singer \thanks{North Carolina State University, Department of Mathematics, Box 8205, Raleigh, NC 27695-8205,  \texttt{singer@ncsu.edu}. The second author was partially supported by a grant from the Simons Foundation (\#349357, Michael Singer). }
         }
\date{}

\definecolor{asparagus}{rgb}{0.53, 0.66, 0.42}

\definecolor{britishracinggreen}{rgb}{0.0, 0.26, 0.15}
\definecolor{deeplilac}{rgb}{0.6, 0.33, 0.73}

\newcommand{\N}{\mathbb{N}}
\newcommand\Z{\mathbb{Z}}
\newcommand\C{\mathbb{C}}
\newcommand\Q{\mathbb{Q}}
\newcommand\R{\mathbb{R}}
\newcommand{\K}{\mathbb{K}}
\newcommand{\Ks}{\ensuremath{{\K}^*}}

\newcommand{\bKs}{\ensuremath{\bar{\K}^*}}

\newcommand{\KX}{\ensuremath{\K[X]}} 
\newcommand{\Kx}{\ensuremath{\K[x^{\pm}]}} 
\newcommand{\IKx}[1][]{\ensuremath{\K[x^{\pm}]^{\gva}_{#1}}} 
\newcommand{\SIKx}{\ensuremath{\K[x^{\pm}]^{\gva}_{\chi}}} 
\newcommand{\SIKxd}{\ensuremath{\left(\K[x^{\pm}]^{\gva}_{\chi}\right)^*}} 
\newcommand{\SKx}{\ensuremath{\mathcal{R}}} 

\newcommand{\Skep}{\ensuremath{\mathcal{S}}}
\newcommand{\Skel}{\ensuremath{\mathcal{S}^*}}

\newcommand{\hcross}[2][n]{\ensuremath{\mathcal{C}_{#2}^{#1}}}
\newcommand{\wcross}[2][\gva]{\ensuremath{\mathfrak{X}_{#2}^{#1}}}
\newcommand{\wcroxx}[2][\gva]{\ensuremath{\check{\mathfrak{X}}_{#2}^{#1}}}

\newcommand{\I}[1][]{\ensuremath{I_{#1}}}
\newcommand{\lspan}[1]{\ensuremath{ \left\langle #1\right\rangle}}
\DeclareMathOperator{\rank}{rank}
\DeclareMathOperator{\diag}{diag}
\newcommand{\tr}[2][]{\ensuremath{{#2}_{#1}^{\mathsmaller{\mathsf{T}}}}}
\newcommand{\ti}[2][]{\ensuremath{{#2}_{#1}^{\mathsmaller{-\mathsf{T}}}}}

\newcommand{\GL}[1][n]{\ensuremath{\mathrm{GL}_{#1}}}
\newcommand{\GLZ}[1][n]{\ensuremath{\mathrm{GL}_{#1}(\Z)}}

\newcommand{\gva}{\ensuremath{\mathcal{W}}}

\newcommand{\reynold}[1][\chi]{\ensuremath{\mathfrak{p}_{#1}}}

\newcommand{\WeylG}{\ensuremath{\mathcal{W}}}
\newcommand{\WeylA}[1][2]{\ensuremath{\mathcal{A}_{#1}}}
\newcommand{\WeylB}[1][2]{\ensuremath{\mathcal{B}_{#1}}}
\newcommand{\PC}{\ensuremath{\Lambda\!\!\Lambda}}

\newcommand{\ssl}[1][n]{\ensuremath{\mathrm{sl}_{#1}}}
\newcommand{\ggl}[1][n]{\ensuremath{\mathrm{gl}_{#1}}}
\newcommand\frakg{\mathfrak{g}}
\newcommand\frakh{\mathfrak{h}}
\newcommand\ad{{\rm ad}}

\newcommand{\isp}[2]{\ensuremath{\langle#1,#2\rangle}}

\newcommand{\mB}{\mathrm{B}}
\newcommand{\mR}{\mathrm{R}}

\newcommand{\cha}[1]{\ensuremath{{\Xi}_{#1}}}
\newcommand{\orb}[1]{\ensuremath{{\Theta}_{#1}}}
\newcommand{\aorb}[1]{\ensuremath{{\Upsilon}_{#1}}}

\newcommand{\corb}[2]{\ensuremath{{\Psi}^{#1}_{#2}}}

\newcommand{\ba}{\ensuremath{\Upupsilon}}

\newcommand{\Tche}[1]{\ensuremath{T_{#1}}}

\newcommand{\lifo}{\ensuremath{\Omega}}

\newcommand{\eval}[1][]{\ensuremath{\mathbbm{e}_{#1}}} 
\newcommand{\Hanc}{\ensuremath{\widehat{\mathcal{H}}}}
\newcommand{\SHanc}{\ensuremath{\widehat{\mathcal{H}}}}
\newcommand{\IHanc}{\ensuremath{\widehat{\mathcal{H}}^{\gva}}}

\newcommand{\SHank}{\ensuremath{\mathcal{H}}}

\newcommand{\kerk}[1][]{\ensuremath{I_{\lifo_{#1}}}} 
\newcommand{\Skerk}[1][]{\ensuremath{I_{\lifo_{#1}}}}
\newcommand{\Ikerk}[1][]{\ensuremath{I_{\lifo_{#1}}^{\gva}}}

\newcommand{\Multi}{\ensuremath{\mathcal{M}}}

\theorembodyfont{\slshape}
\newtheorem{lemma}{Lemma}[section]
\newtheorem{example}[lemma]{Example}

\newtheorem{definition}[lemma]{Definition}
\newtheorem{proposition}[lemma]{Proposition}

\newtheorem{theorem}[lemma]{Theorem}
\newtheorem{corollary}[lemma]{Corollary}

\newenvironment{proof}{\textsc{proof:} }{\ensuremath{\blacksquare}}
\def \qed{\hspace*{2mm} \hfill $\Box $\bigskip}
\newcommand{\ds}{\displaystyle}

\newcommand{\Ref}[1]{(\ref{#1})}
\newcommand{\secr}[1]{Section~\ref{#1}}
\newcommand{\prpr}[1]{Proposition~\ref{#1}}
\newcommand{\colr}[1]{Corollary~\ref{#1}}
\newcommand{\thmr}[1]{Theorem~\ref{#1}}
\newcommand{\lemr}[1]{Lemma~\ref{#1}}
\newcommand{\exmr}[1]{Example~\ref{#1}}
\newcommand{\dfnr}[1]{Definition~\ref{#1}}

\newcommand{\algr}[1]{Algorithm~\ref{#1}}

\newtheorem{algoy}[lemma]{Algorithm}{\scshape}{\slshape}
\newcommand{\Algf}{\bfseries \sffamily }

\newcommand{\italgf}{\slshape  }

\newenvironment{italg}{\begin{list}{}{\italgf\setlength{\itemsep}{0pt}\
\setlength{\topsep}{0pt}\setlength{\leftmargin}{10pt}\
\setlength{\itemindent}{0pt}}}{\end{list}}

\newcommand{\In}[1]{\hspace*{-\parindent}{\textsc{Input:} } {\rmfamily #1}}
\newcommand{\Out}[1]{\hspace*{-\parindent}{\textsc{Output:} } {\rmfamily #1} }

\begin{document}

\maketitle



\begin{abstract} {Sparse interpolation} refers to the exact recovery of a function
as a short linear combination of basis functions from a limited number of evaluations.
For multivariate functions, the case of the monomial basis  is well studied, 
as is now the basis of exponential functions. 
 Beyond the multivariate Chebyshev polynomial obtained as tensor products of univariate Chebyshev polynomials, the theory of root systems allows to define a variety of generalized multivariate  Chebyshev polynomials that have connections to topics such as Fourier analysis and representations of Lie algebras.     
We present a deterministic algorithm to recover a function that is the linear combination of at most $r$ such polynomials from the knowledge of $r$ and an explicitly bounded number of evaluations of this function. 
\\[5pt]
\\[3pt] 
 \end{abstract}

\newpage

\tableofcontents

\newpage

%

\section{Introduction} 

The goal of \emph{sparse interpolation} is the exact recovery of a function
as a short linear combination of elements in a specific set of functions, 
usually of infinite cardinality, 
from a limited number of evaluations, or other functional values.
The function to recover is sometimes refered to as a \emph{blackbox}: it can be evaluated, but its expression is unknown.
We consider the case of a 
multivariate function $f(x_1, \ldots , x_n)$ 
that is a sum of  {generalized Chebyshev polynomials}
and present an algorithm to retrieve the summands.
We assume we know the number of summands, or an upper bound for this number,
and the values of the function at a  finite set of well chosen points.

Beside their strong impact in analysis, Chebyshev polynomials 
arise in the representation theory of simple Lie algebras.  
In particular, the Chebyshev polynomials of the first kind may be identified 
with orbit sums of weights of the Lie algebra $\mathrm{sl}_2$ 
and the  Chebyshev polynomials of the 
second kind may be identified with characters of this Lie algebra. 
{Both types of polynomials are invariant under the action of the  symmetric group  $\{1,-1\}$, 
the associated associated Weyl group, on the exponents of the monomials. }
In presentations of the theory of Lie algebras (c.f., \cite[Ch.5,\S3]{Bourbaki_4_5_6}), 
this identification is often discussed in the context of the associated root systems 
and we will take this approach. In particular, we define the \emph{generalized Chebyshev polynomials}  
associated to a  root system, as similarly done in
\cite{Hoffman88,LyUv2013,Moody87,Munthe-Kaas13}.
Several authors have already exploited 
the connection between Chebyshev polynomials 
and  the theory of  Lie algebras or root systems 
(e.g., \cite{Dieudonne_79}, \cite{Nesterenko_Patera}, \cite{Vilenkin68}) 
and successfully used this in the context of quadrature problems 
\cite{Li10,Moody14,Moody11,Munthe-Kaas13} or differential equations \cite{Ryland11}.

A forebear of our algorithm is Prony's method to retrieve a univariate function as a linear combination of exponential functions from its values at equally spaced points \cite{Riche95}. 
The method was further developed in a numerical  context \cite{Pereyra10}.
In exact computation, mostly over finite fields,  some of the algorithms for 
the sparse  interpolation of multivariate polynomial functions in terms of monomials 
bear similarities to Prony's method
and have connections with linear codes \cite{Ben-Or88, Arnold16}.
General frameworks for sparse interpolation 
were proposed in terms of sums of 
characters of Abelian groups and  sums of eigenfunctions of 
linear operators \cite{Dress91,Grigoriev91}.
The algorithm in \cite{Lakshman95} for the recovery of 
a linear combination of univariate Chebyshev polynomials does not 
fit in these frameworks though.
Yet, as observed in \cite{Arnold15}, a simple change of variables turns  
Chebyshev polynomials into  Laurent polynomials with a simple symmetry
in the exponents. This symmetry is most naturally explained in the context of root systems and Weyl groups and leads to a multivariate generalization. 

Previous algorithms \cite{Arnold15,Giesbrecht04,Kaltofen18,Lakshman95,Potts14} for sparse interpolation in terms of Chebyshev polynomials  of one variable depend 
heavily on the  relations for the products,  an identification property,
and the commutation of composition.
We show in this paper how analogous results hold for generalized Chebyshev polynomials of several variables and stem from the underlying 
root system. 
As already known,  
expressing the multiplication of generalized Chebyshev polynomials  in terms of other generalized Chebyshev polynomials is  presided over by the Weyl group.
As a first original result we show how 
to select $n$ points in $\Q^n$ so that each $n$-variable generalized Chebyshev polynomial  is determined by its values at these $n$ points 
(\lemr{falafel}, \thmr{kumquat}).
A second original observation permits to generalize the commutation 
property in that we identify points where commutation 
is available (\prpr{basic}).

To provide a full algorithm, we  revisit sparse interpolation in an intrinsically multivariate approach that allows one  to preserve and exploit symmetry. 
For the interpolation of sparse sums of Laurent monomials the algorithm 
 presented (\secr{pulpo}) has strong ties with a multivariate Prony method \cite{Kunis16,Mourrain18,Sauer16a}. 
It associates to each sum of $r$ monomials 
$f(x) = \sum_\alpha a_\alpha x^\alpha, \mbox{ where } x^{\alpha} = x_1^{\alpha_1} \ldots x_n^{\alpha_n}$ and $a_\alpha$ in a field $\K$,   a linear form $\Omega:\K[x,x^{-1}] \rightarrow \K$ given by $\Omega(p) = \sum_\alpha a_\alpha p(\zeta_\alpha)$ where $\zeta_\alpha = (\xi^{\alpha_1}, \ldots ,  \xi^{\alpha_n})$ for suitable $\xi$. This linear form allows us to define a Hankel operator from $K[x,x^{-1}] $ to its dual (see Section~\ref{HankelMultiplication}) whose kernel is an ideal $I$ having precisely the $\zeta_\alpha$ as its  zeroes. 
The  $\zeta_\alpha$ can be recovered  as eigenvalues of multiplication maps on 
$\K[x,x^{-1}]/I$. The matrices of these multiplication maps can actually
 be calculated directly in terms of the matrices of a Hankel operator,
   without explicitly calculating $I$. One can then find the 
$\zeta_\alpha$ and the $a_\alpha$  using  only linear algebra and evaluation 
of the original polynomial $f(x)$ at well-chosen points. The calculation of
the $(\alpha_1, \ldots ,  \alpha_n)$ is then reduced to the calculation of logarithms.


The usual Hankel or mixed Hankel-Toepliz matrices 
that appeared in the literature on sparse interpolation  \cite{Ben-Or88,Lakshman95} are actually 
 the matrices of the Hankel operator mentioned above in the different univariate polynomial bases considered. 
The recovery of the support of a linear form with 
 this type of technique also appears in 
 optimization, tensor decomposition and cubature 
\cite{Abril-Bucero16,Bernardi18,Brachat10,Collowald18,Lasserre10,Laurent09}.
We present new developments to take advantage of the invariance or 
semi-invariance of the linear form. This allows us to reduce the size of the matrices involved by a factor equal to the order of the Weyl group (\secr{hankelinv}).

For sparse interpolation in terms of Chebyshev polynomials (\secr{anxoves} and 
\ref{bocarones}), one again recasts this problem in terms of a linear form on a Laurent polynomial ring. We define an action of the Weyl group on this ring as well as on the underlying ambient space and note that the linear form is invariant or semi-invariant according to whether we consider generalized Chebyshev polynomials of the first or second kind. Evaluations, at specific points, of the function to interpolate provide the 
knowledge of the linear form on a linear basis of the invariant subring  or semi-invariant module. In the case of interpolation of sparse sums of Laurent monomials the seemingly trivial yet important fact that 
$(\xi^\beta)^\alpha = (\xi^\alpha)^\beta$ is crucial to the algorithm. 
In the multivariate Chebyshev case we  identify a family of evaluation 
points that provides a similar commutation property in 
the Chebyshev polynomials (\lemr{basic}). 

Since the linear form is invariant, or semi-invariant,
 the support consists  of points grouped into orbits of the 
action of the Weyl group. Using tools developed in analogy to the Hankel formulation above, we show how to recover the values 
of the fundamental invariants (\algr{invsupport}) on each of these orbits and, 
 from these, the values of the Chebyshev polynomials that appear in the sparse sum.  Furthermore, we show how to recover each Chebyshev polynomial from its values at $n$ carefully selected points (Theorem~\ref{kumquat}).

The relative cost of our algorithms depends on the linear algebra operations used in recovering the support of the linear form and the number of evaluations needed.
Recovering the support of a linear form  on the Laurent polynomial ring 
is solved with linear algebra after introducing the appropriate 
Hankel operators. 
Symmetry reduces the size of matrices, as expected,
 by a factor the order of the group.
Concerning evaluations of the function to recover, 
we need evaluations  to determine certain sunbmatrices of maximum rank used in the linear algebra component of the algorithms.  
To bound the number of evaluations needed, we rely on the
interpolation property of sets of polynomials indexed by the  
hyperbolic cross (\prpr{bounty}, \colr{chufa}), 
a result generalizing the case of monomials in \cite{Sauer16a}.
 The impact of this on the relative costs of the algorithms  
 is discussed in Section~\ref{sec:eval}.

The paper is organized as follows. In Section~\ref{chebyshev}, we begin by describing the connection between univariate Chebyshev polynomials and the representation theory of traceless $2\times 2$ matrices.  
We then turn to the multivariate case and review the theory of root systems needed to define and work with generalized Chebyshev polynomials.
The section concludes with the first original contribution: 
we show how an $n$-variable Chebyshev polynomial, of the first or second kind,
is determined by its values on $n$ special points. 
In Section~\ref{interpolation} we 
show how multivariate sparse interpolation can be reduced 
to  retrieving the support 
of certain linear forms on a Laurent polynomial ring.   
 For sparse interpolation in terms of multivariate Chebyshev polynomials of the first and second kind, we show how we can consider the restriction  
 of the linear form  to the ring of invariants of the Weyl group or the module of semi-invariants. In addition, we discuss some of the costs of our algorithm as compared to treating generalized Chebyshev polynomials as sums of monomials.
 In \secr{hankel} we introduce Hankel operators and their use in determining algorithmically the support of a linear form  through linear algebra operations. After reviewing the definitions of Hankel operators and multiplication matrices in the context of linear forms on a Laurent polynomial ring, we extend these tools to apply to linear forms invariant under a Weyl group and show how these developments allow one to 
 scale down the size of the matrices by a factor equal to the order of this group.
 Throughout these sections we provide examples to illustrate the theory and the algorithms. In Section~\ref{final} we discuss the global algorithm and point out some directions of further improvement.

\textbf{Acknowledgment:} 
The authors wish to thank the Fields institute 
and the organizers of the thematic program on computer algebra 
where this research was initiated.
They also wish to thank  Andrew Arnold for discussions on sparse interpolation 
and the timely pointer on the use of the hypercross in the multivariate case.

\newpage 

\section{Chebyshev polynomials}\label{chebyshev} \label{cheby}

In this section we first  discuss how the usual Chebyshev polynomials arise from considerations concerning root systems and their Weyl group.
This approach allows us to give  higher dimensional generalizations of these polynomials  
\cite{Hoffman88,Munthe-Kaas13}. 
We review the results about root systems and representation theory allowing us to define the generalized Chebyshev polynomials of the first and second kind.
This section concludes with the first original result in this article necessary to our purpose: we show how one can determine the degree of a Chebyshev polynomial from its values at few well chosen points. 

\subsection{Univariate Chebyshev polynomials}\label{unicheb}

The univariate Chebyshev polynomials of the first and second kind  arise in many contexts; approximation theory, polynomial interpolation, and quadrature  formulas are examples.  A direct and simple way to define these polynomials 
 is as follows.  
\begin{definition} \begin{enumerate}
\item The \emph{Chebyshev polynomials of the first kind}, $\{\tilde{T}_n(x) \ | \ n=0, 1, 2, \ldots \}$,   are the unique monic polynomials satisfying 
\[ \tilde{T}_n(\cos(\theta)) = \cos(n\,\theta) \quad\mbox{ or }\quad 
\tilde{T}_n\left(\frac{x + x^{-1}}{2}\right) = \frac{x^{n} + x^{-n}}{2}
.\]
\item The \emph{Chebyshev polynomials  of the second kind}, $\{\tilde{U}_n(x) \ | \ n=0, 1, 2, \ldots \}$, are the unique monic polynomials satisfying 
\[
\tilde{U}_n(\cos(\theta)) 
= \frac{\sin((n+1)\,\theta)}{\sin(\theta)} 
\quad\mbox{ or }\quad 
{
\tilde{U}_n\left(\frac{x + x^{-1}}{2}\right) 
= \frac{ x^{(n+1)} - x^{-(n+1)} }{x - x^{-1}}
= x^{n} + x^{n-2} + \dots + x^{2-n}+ x^{-n}}
.\]
\end{enumerate}
\end{definition}
{The second set of equalities for $\tilde{T}_n$ and $\tilde{U}_n$ are  familiar when written in terms of 
$x=e^{i\theta}$ since $\cos{n\theta}=\frac{1}{2}\left(e^{in\theta} +e^{-in\theta}\right)$ and 
$\sin(n\theta)=\frac{1}{2}\left(e^{in\theta} - e^{-in\theta}\right)$. We introduced these equalities in terms of $x$ for a clearer connection with the following sections.}

These polynomials also arise naturally when one studies the representation theory of the Lie algebra 
$\ssl[2](\C)$ of $2\times 2$-matrices with zero trace 
\cite{Dieudonne_79,Vilenkin68}.
  Any representation $\pi:\ssl[2](\C) \rightarrow \ggl[n](\C)$ is a direct sum of irreducible representations.  For each nonnegative integer $n$, there is a  unique irreducible representation  $\pi_n: \ssl[2](\C) \rightarrow \ggl[n+1](\C)$ of dimension $n+1$  (see \cite[Capitre IV]{Serre66} for a precise description). Restricting this representation to the diagonal matrices 
 $\left\{ \diag( a, -a) \ | \ a\in \C\right\}$,  
   this map is given by
$\pi_n\left(\diag( a,\; -a)\right)=\diag(na,\;(n-2)\,a,\;\ldots ,\;(2-n)\,a,\; -n\,a)$. 
Each of the maps 
$\diag( a,\; -a) \mapsto m\,a$, 
for $m = n,\,n-2,\, \ldots ,\, 2-n ,\, -n$ 
is called a {\it weight} of this representation.  The set of weights appearing in the  representations of $\ssl[2](\C)$ may therefore be identified with the lattice of integers  in the one-dimensional vector space $\R$.  The group of automorphisms of this vector space that preserves this lattice is precisely the two element group $\{id, \sigma\}$ where $id(m) = m$ and $\sigma(m) = -m$.  This group is called the {\it Weyl group} $\WeylG$.   

We now make the connection between Lie theory and  Chebyshev polynomials.  Identify the weight corresponding to the integer $m$ with the {\it weight monomial} $x^m$ in the Laurent polynomial ring $\Z[x,x^{-1}]$ and let the generator $\sigma$ of the group $\WeylG$ act on this ring via the map $\sigma\cdot x^m = x^{\sigma(m)}$.   For each weight monomial $x^m, m\geq 0$, we can define the {\it orbit polynomial} 
\[\orb{m}(x)= x^m + x^{-m}\]
 and the {\it character polynomial} 
 \[\cha{m}(x)=  x^m + x^{m-2} + \ldots + x^{2-m} + x^{-m}.\]
 Note that for each $m$, both of these polynomials lie in the ring of invariants $\Z[x,x^{-1}]^\WeylG = \Z[x+x^{-1}]$ of the Weyl group. Therefore there exist polynomials ${T}_n(X)$ and ${U}_n(X)$ such that $\orb{n}(x) = {T}_n(x+x^{-1})$ and  $\cha{n}(x) = {U}_n(x+x^{-1})$.  The Chebyshev polynomials of the first and second kind can be recovered using the formulas
\[\tilde{T}_n(X) = \frac{1}{2}{T}_n(2X) 
   \mbox{ and } \tilde{U}_n(X) = {U}_n(2X).\]

The previous discussion shows how the classical Chebyshev polynomials arise from 
representation of a semisimple Lie algebra and
the action of the Weyl group on a Laurent polynomial ring. As noted above, this discussion could have started just with the associated root system and its Weyl group  and weights. This is precisely what we do in Section~\ref{cheby1} and \ref{cheby2} where we define a generalization of these polynomials for any (reduced)  root system. 


\subsection{Root systems and Weyl groups} \label{primaldual:roots} 

We review the definition and results on root systems that are needed to define generalized Chebyshev polynomials.
These are taken from 
\cite[Chapitre VI]{Bourbaki_4_5_6},\cite[Chapter 8]{Hall15} or \cite[Chapitre V]{Serre66} where complete expositions can be found.

\begin{definition} \label{root:def} 
 Let $V$ be a finite dimensional real vector space with an inner product $\langle\cdot \, , \cdot \rangle$ and $\mR $ a finite  subset of $V$. We say $\mR$ is a {\em root system} in $V$ if
\begin{enumerate}
\item $\mR$ spans $V$ and does not contain $0$.
\item If $\rho, \tilde{\rho} \in \mR$, then $s_\rho(\tilde{\rho}) \in \mR$, where $s_\rho$ is the  reflection defined by 
\(\displaystyle s_\rho(\gamma) = \gamma - 2\frac{\langle\gamma,\rho\rangle}{\langle\rho,\rho\rangle}\rho, \ \ \gamma \in  V.\)
\item For all $\rho, \tilde{\rho} \in \mR$, 
\(\displaystyle 2\frac{\langle\tilde{\rho},\rho\rangle}{\langle\rho,\rho\rangle} \in \Z.\)
\item If $\rho \in \mR$, and $c \in \R$, then $c\rho \in R$ if and only if $c = \pm 1$.

\end{enumerate}
\end{definition}

The definition of $s_\rho$ above implies that $\langle s_\rho(\mu),s_\rho(\nu)\rangle=\langle \mu,\nu\rangle$ for any $\mu,\nu \in V$.

In many texts, a root system is defined only using the first three of the above conditions and the last condition is used to define a reduced root system. All root systems in this paper are reduced so we include this last condition  in our definition and dispense with the adjective ``reduced''. Furthermore, some texts define a root system without reference to an inner product (c.f.~\cite[Chapitre VI]{Bourbaki_4_5_6},\cite[Chapitre V]{Serre66}) and only introduce an inner product later in their exposition.  The inner product allows one to identify $V$ with its dual $V^*$ in a canonical way and this  helps us with many computations.  



\begin{definition}The \emph{Weyl group} 
$\WeylG$ of a root system $\mR$ in $V$ is  
 the subgroup of the orthogonal group, with respect to the inner product $\langle\cdot \, , \cdot \rangle$,  generated by the reflections $s_\rho$, $\rho \in \mR$.
\end{definition}

One can find a useful basis of the ambient vector space $V$ sitting inside the set of roots :
\begin{definition} Let $\mR$ be a root system. \begin{enumerate}
\item A subset $\mB = \{\rho_1, \ldots , \rho_n\}$ of $\mR$ is a {\rm base} if 
\begin{enumerate}
\item $\mB$ is a basis of the vector space $V$.
\item Every root $\mu \in \mR$ can be written as $\mu = \alpha_1\rho_1 + \ldots +\alpha_n \rho_n$ or $\mu = -\alpha_1\rho_1 - \ldots -\alpha_n \rho_n$ for some $\alpha \in \N^n$.
\end{enumerate}
\item If $\mB$ is a base, the  roots of the form $\mu = \alpha_1\rho_1 + \ldots +\alpha_n \rho_n$ for some $\alpha \in \N^n$ are called the {\rm positive roots} and the set of positive roots is denoted by $\mR^+$.
\end{enumerate}
\end{definition}

A standard way to show bases exist 
(c.f. \cite[Chapter 8.4]{Hall15},\cite[Chapitre V,\S8]{Serre66})  
is to start by selecting a hyperplane $H$ that does not contain any of the roots  and letting $v$ be an element perpendicular to $H$.  One defines $R^\plus = \{ \rho \in \mR \ | \ \langle v,\rho\rangle >0\}$ and then shows that  $\mB = \{ \rho \in \mR^\plus \ | \ \rho \neq \rho' + \rho '' \mbox{ for any pair } \rho',\rho''\in \mR^\plus\}$, the {\it indecomposable}  positive roots, forms a base.  For any two bases $\mB$ and $\mB'$ there exists a $\sigma \in \WeylG$ such that $\sigma(\mB) = \mB'$.   We fix once and for all a base $\mB$ of $\mR$.

 The base can be used to define the following important cone in $V$. 

\begin{definition} The {\rm closed  fundamental Weyl chamber} 
 in $V$  relative to the base  $\mB = \{\rho_1, \ldots , \rho_n\}$ is 
 $\PC=\{ v \in V \ | \ \langle v,\rho_i\rangle \geq0 \}$. 
The interior of $\PC$ is called the {\rm open fundamental Weyl chamber}.
\end{definition}

Of course, different bases have different open fundamental Weyl chambers. If $L_i$ is the hyperplane perpendicular to an element $\rho_i$ in the base $\mB$, then the connected components of $V - \bigcup_{i=1}^n L_i$ correspond to the possible open fundamental Weyl chambers.  Furthermore, the Weyl group acts transitively on these components.

The element 
\[\rho^{\vee} = 2\frac{\rho}{\langle\rho,\rho\rangle}\]
that appears in the definition of $s_\rho$  is called the {\it coroot} of $\rho$.  
The set of all coroots  is denoted by $\mR^\vee$ and this set is again a root system called the \emph{dual root system}
with the same Weyl group as 
$\mR$ \cite[Chapitre VI, \S 1.1]{Bourbaki_4_5_6},\cite[Proposition 8.11]{Hall15}.
If $\mB$ is a base of $\mR$ then $\mB^\vee$ is a base of $\mR^\vee$.

A root system defines  the following lattice in $V$, called the lattice of weights. This lattice and  related concepts  play an important role in the representation theory of semisimple Lie algebras. 

\begin{definition}\label{defwt} Let $\mB= \{\rho_1, \ldots , \rho_n\}$ 
the base of $\mR$ and $\mB^\vee= \{\rho_1^\vee, \ldots , \rho_n^\vee\}$ its dual.  
\begin{enumerate}
\item An element $\mu$ of $V$ is called a {\rm weight} if \[\langle\mu,\rho_i^\vee\rangle  = 2\frac{\langle \mu,\rho_i\rangle}{\langle \rho_i,\rho_i\rangle} \in \Z\] for $i = 1, \ldots ,n$.  The set of weights forms a lattice called the {\rm weight lattice} $\Lambda$.
\item The {\rm fundamental weights}  are elements $\{ \omega_1, \ldots , \omega_n\}$ such that $\langle \omega_i, \rho_j^\vee\rangle = \delta_{i,j}, i,j = 1, \ldots ,n$.
\item A weight $\mu$ is \emph{ strongly dominant } if $\langle \mu, \rho_i\rangle  >  0$   for all $\rho_i \in \mB$. A weight $\mu$ is \emph{ dominant } if $\langle \mu, \rho_i\rangle \geq  0$   for all $\rho_i \in \mB$, i.e., $\mu \in \PC$. 
\end{enumerate}
\end{definition}

Weights are occasionally referred to as integral elements,
 \cite[Chapter 8.7]{Hall15}.  
 In describing the properties of their lattice it is useful to first define the following partial order on elements of $V$   \cite[Chapter 10.1]{Humphreys72}. 

\begin{definition} For $v_1, v_2 \in V$, we define  $v_1 \succ v_2$ if $v_1 - v_2 $ is a sum of positive roots or $v_1 = v_2$, that is, $v_1 - v_2 = \sum_{i=1}^nn_i \rho_i$ for some $n_i \in \N$.\end{definition}

The following proposition states three  key properties of weights and  of dominant weights which we will use later.
\begin{proposition}  \label{domprop} 
\begin{enumerate} 
\item The weight lattice $\Lambda$ is invariant under the action of the Weyl group $\WeylG$.
\item Let $\mB= \{\rho_1, \ldots ,\rho_n\}$ be a base. If $\mu$ is a dominant weight and $\sigma \in \WeylG$, then $\mu \succ \sigma(\mu)$.   If $\mu$ is a strongly dominant weight, then $\sigma(\mu) = \mu$ if and only if $\sigma $ is the identity.
\item $\ds \delta = \frac{1}{2} \sum_{\rho\in\mR^+}\rho$ is a strongly dominant weight 
equal to  $\ds \sum_{i=1}^n \omega_i$.
\item If $\mu_1$ and $\mu_2$ are dominant weights, then $\langle \mu_1, \mu_2\rangle \geq 0$.
\end{enumerate}
\end{proposition}
\begin{proof} The proofs of items 1.,~2.,~and 3.~may be found in \cite[Section 13.2 and 13.3]{Humphreys72}. For item 4. it is enough to show this when $\mu_1$ and $\mu_2$ are fundamental weights since dominant weights are nonnegative integer combinations of these. The fact for fundamental weights follows from Lemma 10.1 and Exercise 7 of Section 13 of \cite{Humphreys72} (see also  \cite[Proposition 8.13, Lemma 8.14]{Hall15}). \end{proof}


\begin{example} \label{ex:roots} The (reduced) root systems have been classified and presentations of these can be found in many texts.  We give  three examples, $\WeylA[1], \WeylA[2], \WeylB[2]$, here.  In most texts, these examples are given so that the inner product is the usual inner product on Euclidean space.  We have chosen the following representations because we want the associated weight lattices (defined below) to be the integer lattices in the ambient vector spaces. Nonetheless there is an isomorphism of the underlying inner product spaces identifying these representations. 

\begin{description}
\item[{$\WeylA[1]$}.]  This system has two elements   $[2], [-2]$ in  $V = \R^1$. 
The inner product given by $\isp{u}{v} = \frac{1}{2}uv$. 
A base is given by $\rho_1=[2]$. The Weyl group has two elements, given by the matrices $\begin{bmatrix} 1 \end{bmatrix}$ and $\begin{bmatrix} -1 \end{bmatrix}$.

\item[{$\WeylA[2]$.}] This system has 6 elements 
$\pm \tr{ \begin{bmatrix} 2 & -1 \end{bmatrix} }$, 
$\pm \tr{ \begin{bmatrix}-1 & 2 \end{bmatrix} }$, 
$\pm \tr{ \begin{bmatrix} 1 & 1 \end{bmatrix} } \in \R^2$
when   
the inner product is given by $\isp{u}{v}=\tr{u}S\,v$ where
$$S = \frac{1}{3}\begin{bmatrix} 2 & 1
  \\ 1 & 2 \end{bmatrix}.$$
 A base is given by 
 $\rho_1=\tr{ \begin{bmatrix} 2 & -1 \end{bmatrix} }$ 
 and $\rho_2=\tr{ \begin{bmatrix}-1 & 2 \end{bmatrix} }$.
 We have $\isp{\rho_i}{\rho_i}=2$ so that $\rho_i^\vee=\rho_i$ for $i=\{1,2\}$. 

 The Weyl group is of order $6$ and represented by the matrices
 \[ 
\underbrace{\left[ \begin {array}{cc} -1&0\\ \noalign{\medskip}1&1\end {array} \right]}_{A_1}, 
\;
\underbrace{\left[ \begin {array}{cc} 1&1\\ \noalign{\medskip}0&-1\end {array}\right]}_{A_2},
\;
\left[ \begin {array}{cc} 0&-1\\ \noalign{\medskip}-1&0\end {array} \right],
\;
\left[ \begin {array}{cc} 1&0\\ \noalign{\medskip}0&1\end {array}\right],
\; 
 \left[ \begin {array}{cc} -1&-1\\ \noalign{\medskip}1&0\end {array} \right],
 \; 
\left[ \begin {array}{cc} 0&1\\ \noalign{\medskip}-1&-1\end {array} \right].
 \]
where $A_1$ and $A_2$ are the reflections associated with $\rho_1$ and $\rho_2$.
We implicitly made choices so that the fundamental weights are  
$\omega_1=\tr{ \begin{bmatrix} 1 & 0 \end{bmatrix} }$ and 
$\omega_2=\tr{ \begin{bmatrix} 0 & 1 \end{bmatrix} }$.
The lattice of weights is thus the integer lattice in $\R^2$ 
and orbits of weights are represented in Figure~\ref{A2orbits}.

\item[{$\WeylB[2]$.}]  This system has 8 elements 
$\pm\tr{ \begin{bmatrix} 2 & -2 \end{bmatrix} }$,
$\pm\tr{\begin{bmatrix}  -1 & 2\end{bmatrix} }$, 
$\pm\tr{\begin{bmatrix}  0 & 2\end{bmatrix} }$, 
$\pm \tr{\begin{bmatrix}  1 & 0\end{bmatrix}}$
when 
the inner product is given by $\isp{u}{v}=\tr{u}S\,v$ where 
$$S = \frac{1}{2}\,\begin{bmatrix} 2 & 1\\ 1 & 1 \end{bmatrix}.$$
A base is given by 
$\rho_1=\tr{ \begin{bmatrix} 2 & -2 \end{bmatrix} }$ 
and 
$\rho_2=\tr{ \begin{bmatrix}-1 & 2 \end{bmatrix} }$.
We have $\isp{\rho_1}{\rho_1}=2$ and $\isp{\rho_2}{\rho_2}=1$. Hence 
	 $\rho_1^\vee=\rho_1$  and $\rho_2^\vee=2\,\rho_2$. 
 The Weyl group is of order $8$ and represented by the matrices
 \[ 
\left[ \begin {array}{cc} 1&0\\ \noalign{\medskip}0&1\end {array}
 \right] , \left[ \begin {array}{cc} -1&0\\ \noalign{\medskip}2&1
\end {array} \right] , \left[ \begin {array}{cc} 1&1
\\ \noalign{\medskip}0&-1\end {array} \right] , \left[ \begin {array}
{cc} 1&1\\ \noalign{\medskip}-2&-1\end {array} \right] , \left[ 
\begin {array}{cc} -1&-1\\ \noalign{\medskip}2&1\end {array} \right] ,
 \left[ \begin {array}{cc} -1&-1\\ \noalign{\medskip}0&1\end {array}
 \right] , \left[ \begin {array}{cc} 1&0\\ \noalign{\medskip}-2&-1
\end {array} \right] , \left[ \begin {array}{cc} -1&0
\\ \noalign{\medskip}0&-1\end {array} \right]
.
 \]
 We implicitly made choices so that the fundamental weights are  
$\omega_1=\tr{ \begin{bmatrix} 1 & 0 \end{bmatrix} }$ and 
$\omega_2=\tr{ \begin{bmatrix} 0 & 1 \end{bmatrix} }$.
The lattice of weights is thus the integer lattice in $\R^2$ 
and orbits of weights are represented in Figure~\ref{A2orbits}.
 \end{description}
\end{example}

{\bf Convention:} We will always assume that the root systems are presented in such a way that the associated weight lattices are the integer lattice.  This implies that the associated Weyl group lies in $\GLZ$. 

We may assume that there is a matrix $S$ with rational entries such that 
$<v,w> = \tr{v}Sw$.  This is not obvious from the definition of a root system but follows from the classification of irreducible root systems. Any root system is the direct sum of orthogonal irreducible root systems (\cite[Section 10.4]{Humphreys72}) and these are isomorphic to root systems given by vectors with rational coordinates where the inner product is the usual inner product on affine space \cite[Ch.VI, Planches I-IX]{Bourbaki_4_5_6}. Taking the direct sum of these inner product spaces one gets an inner product on the ambient space with $S$ having rational entries. 
For the examples we furthermore choose $S$ so as to have the 
longest roots to be of norm $2$.

\begin{figure}[h] 
\includegraphics[width=0.4\textwidth]{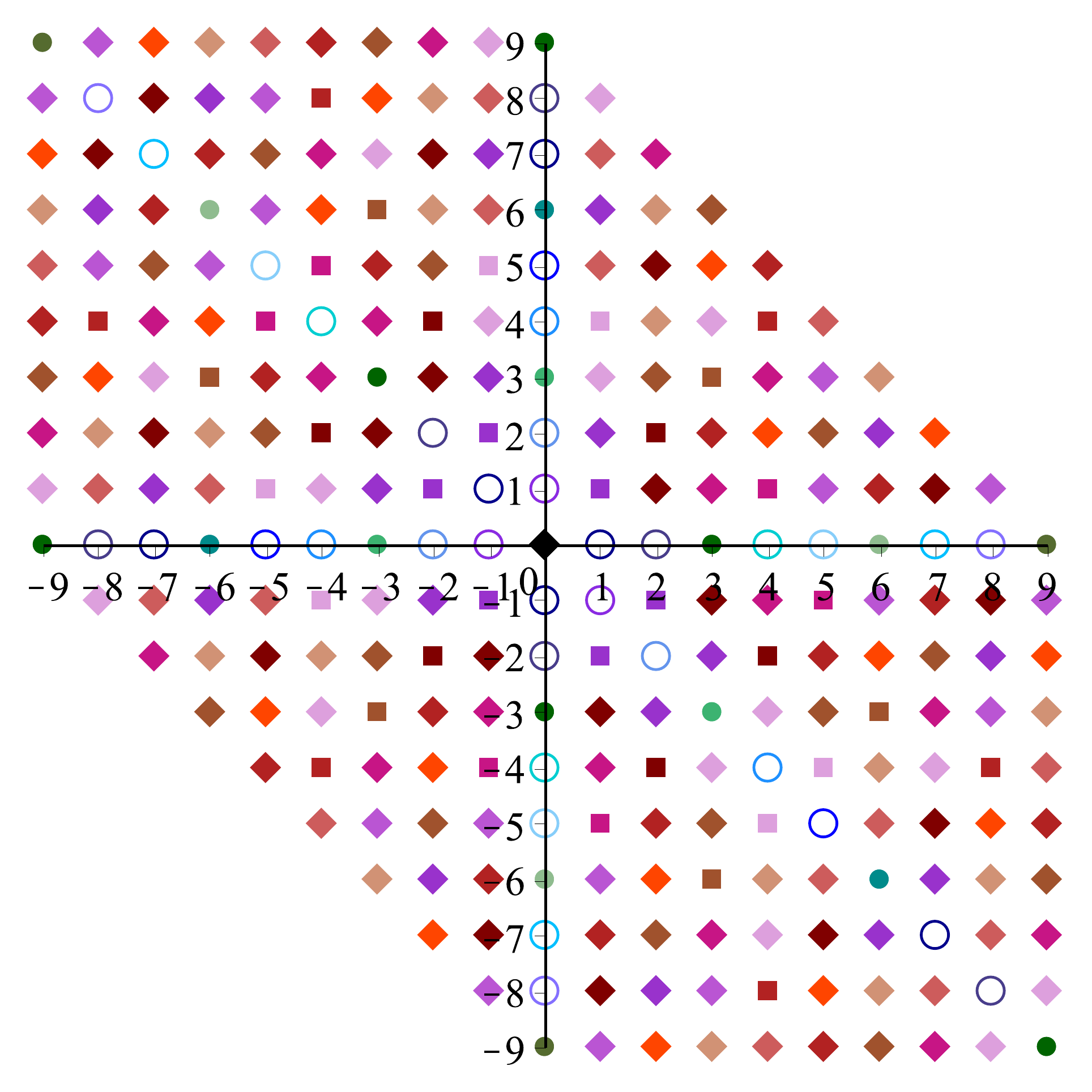}
\includegraphics[width=0.55\textwidth]{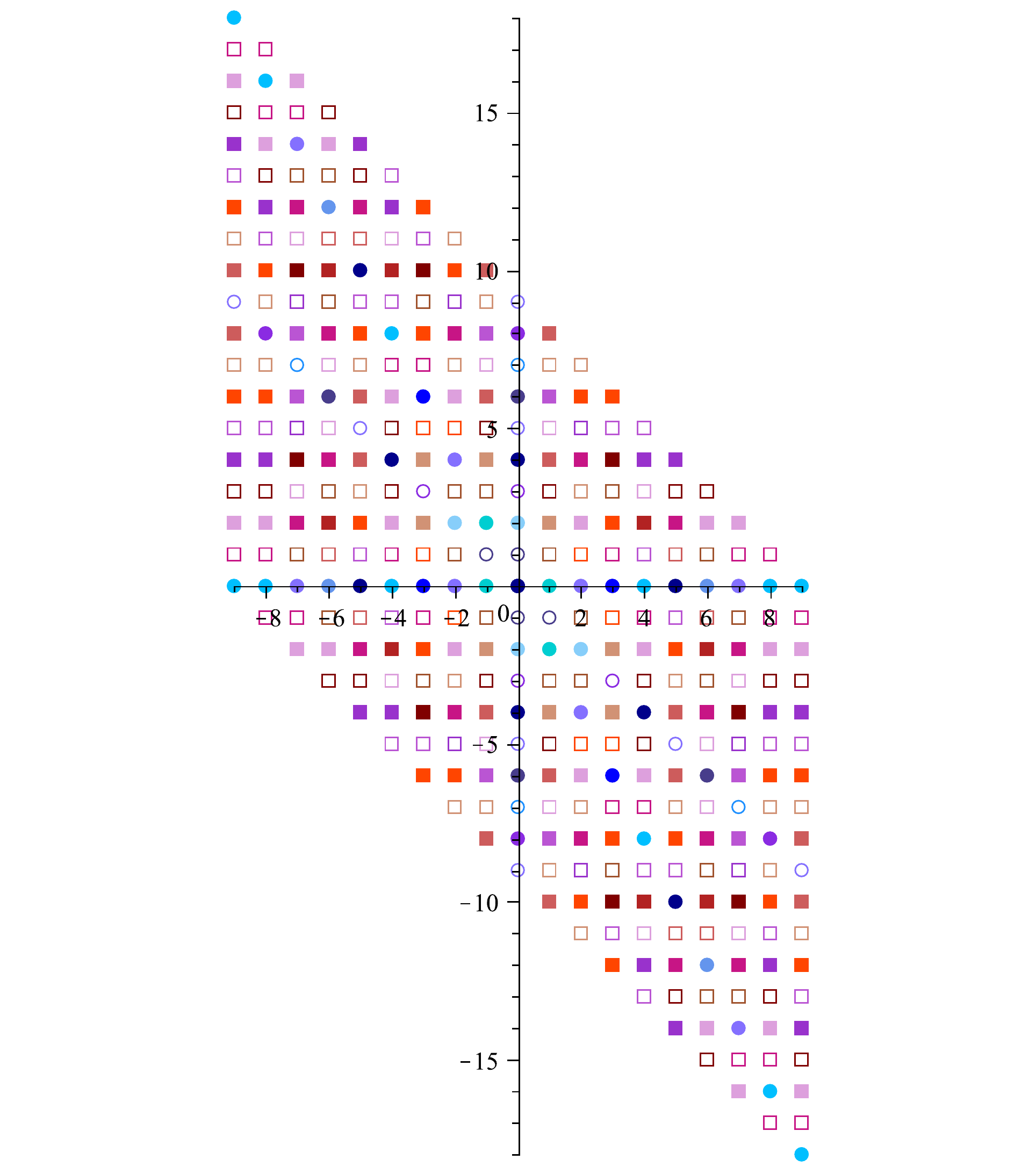}
\caption{$\WeylA[2]$-orbits and 
$\WeylB[2]$-orbits of all $\alpha\in \N^2$ with $|\alpha|\leq 9$.
         Elements of an orbit have the same shape and color. 
         The orbits with $3$, or $4$, elements are represented by circles, 
         the orbits with $6$, or  $8$, elements by diamonds or squares. 
        Squares and solid disc symbols are on the sublattice generated by the roots.} 
\label{A2orbits}
\end{figure}
 

\subsection{Generalized Chebyshev polynomials of the first kind} \label{cheby1}

As seen in Section~\ref{unicheb}, 
the usual Chebyshev polynomials can be defined{}
by considering a Weyl group acting on the 
exponents of monomials  in a ring of Laurent polynomials.  
We shall use this approach to define  Chebyshev
polynomials of several variables  as in \cite{Hoffman88,Munthe-Kaas13}.
This section defines the generalized Chebyshev polynomials of the first kind. 
The next section presents how those of the second kind 
appear in the representations of simple Lie algebras.

Let $\Lambda$ and $\WeylG$ be the weight lattice and 
the Weyl group associated to a root system.
With  $\omega_1,\ldots,\omega_n$ the fundamental weights,
we identify $\Lambda$ with $\Z^n$ through 
$\omega \rightarrow\alpha =  \tr{\left[\alpha_1, \ldots , \alpha_n\right]}$ where
$\omega= \alpha_1\omega_1+ \ldots + \alpha_n\omega_n$.

An arbitrary weight 
$\omega = \alpha_1\omega_1+\ldots +\alpha_n\omega_n\in \Lambda$ 
is associated with the \emph{weight monomial}
$x^\alpha=x_1^{\alpha_1}\ldots x_n^{\alpha_n}$.  
In this way one  sees that the group algebra $\Z[\Lambda]$ 
can be identified with the Laurent polynomial ring 
$\Z[x_1, \ldots , x_n, x_1^{-1}, \ldots ,x_n^{-1}]=\Z[x,x^{-1}]$. 
The action of $\WeylG$ on $\Lambda$ makes us identify $\WeylG$ 
with subgroup of $\GL[n](\Z)$.



Let $\K$ be a field of characteristic $0$ and denote $\K\setminus\{0\}$ by $\K^*$.
The linear action of  $\gva$ 
on $\Kx=\K[x_1,\,\ldots,\, x_n,x_{1}^{-1},\,\ldots,\,x_n^{-1}]$
is defined by 
\begin{equation} \label{WeylLaurentAction}
 \cdot : \begin{array}[t]{ccc}
   \gva \times \Kx & \rightarrow &  \Kx \\
   \left(A, \,x^\alpha\right) & \mapsto&  A\cdot x^\alpha =x^{A\,\alpha} .
   \end{array}.
   \end{equation}
 We have $(A \cdot f) (x) = f\left( x^A \right)$.
One can see the above  action on $\Kx$ as 
induced by the (nonlinear) action on $(\Ks)^n$
defined by the monomial maps:
\begin{equation} \label{multaction}
\begin{array}[t]{ccl}
  \gva \times (\Ks)^n & \rightarrow &  (\Ks)^n \\
    \left(A, \,\zeta \right) & \mapsto &
    A\star \zeta =\left[\zeta_1,\,\ldots,\,\zeta_n\right] ^{A^{-1}} 
    = \left[\zeta^{A^{-1}_{\cdot 1}},\,\ldots,\,\zeta^{A^{-1}_{\cdot n}}\right]
   \end{array}
\end{equation}
where $A^{-1}_{\cdot i}$ is the $i$-th column vector of $A^{-1}$.
 Such actions are
sometimes called multiplicative actions \cite[Section 3]{Lorenz05}.


For a group morphism $\chi:\gva \rightarrow \C^{*}$,
$\alpha,\beta\in\Z^n$ we define
\begin{equation}\label{corb}
\corb{\chi}{\alpha} 
=\sum_{B\in\gva} \chi(B^{-1})\, x^{B\alpha}.\end{equation}
One sees that 
\( A\cdot \corb{\chi}{\alpha} = \corb{\chi}{A\alpha}=\chi(A) \,\corb{\chi}{\alpha}.\)
Two morphisms are of particular interest: 
$\chi(A)=1$ and $\chi(A)= \det(A)$. In either case $(\chi(A))^2=1$ for all $A\in\gva$.
In the former case we define the \emph{orbit polynomial} $\orb{\alpha}$.
In the latter case we use the notation $\aorb{\alpha}$.
\begin{equation} \label{orbpoly}
 \orb{\alpha} =\sum_{B\in\gva} x^{B\alpha}, \quad \hbox{ and } \quad
 \aorb{\alpha} =\sum_{B\in\gva} \det(B)\,x^{B\alpha},
 \end{equation}
where we used the simplificaion $\det(B^{-1})=\det(B)$.

\begin{proposition}\label{fmlas} We have
\[ \orb{\alpha} \, \orb{\beta}= \sum_{B\in\gva}  \orb{\alpha+B\beta},
\quad
 \aorb{\alpha} \, \orb{\beta}= \sum_{B\in\gva}  \aorb{\alpha+B\beta} ,
\quad
 \aorb{\alpha} \, \aorb{\beta}= \sum_{B\in\gva}  \det(B)\,\orb{\alpha+B\beta}.
\]
\end{proposition}

\begin{proof}
This  follows in a straightforward manner  
 from the definitions. 
\end{proof}

Note that  $\orb{\alpha}$ is invariant under the Weyl group action: 
$\orb{\alpha}=A\cdot \orb{\alpha}=\orb{A\alpha}$, for all $A\in\gva$. 
The ring of all invariant Laurent polynomials is denoted $\Z[x,x^{-1}]^\WeylG$.
This ring is isomorphic to a polynomial ring for which generators are known
\cite[Chapitre VI, \S3.3 Th\'eor\`eme 1]{Bourbaki_4_5_6}.

\begin{proposition}\label{prop:invariants} 
Let $\{\omega_1, \ldots , \omega_n\}$ be the
  fundamental weights.
\begin{enumerate}
\item $\{\orb{\omega_1}, \ldots , \orb{\omega_n}\}$ 
  is an algebraically
  independent set of invariant Laurent polynomials.
\item $\Z[x,x^{-1}]^\WeylG = \Z[\orb{\omega_1}, \ldots , \orb{\omega_n}]$
\end{enumerate}
\end{proposition}

We can now define the multivariate generalization of the 
Chebyshev polynomials of the first kind (cf.~\cite{Hoffman88}, \cite{LyUv2013}, \cite{Moody87}, \cite{Munthe-Kaas13})

\begin{definition}\label{Cheb1Def} Let $\alpha\in\N^n$ be a  dominant weight.
 The {\rm Chebyshev polynomial of the first kind associated to $\alpha$} is the polynomial 
 $T_\alpha$ in $\K[X]=\K[X_1, \ldots X_n]$ 
 such that $\orb{\alpha} = T_\alpha(\orb{\omega_1}, \ldots , \orb{\omega_n})$.
\end{definition}

{We shall usually drop the phrase ``associated to $\alpha$'' and just refer to Chebyshev polynomials of the first kind with the understanding that we have fixed a root systems and  each of these polynomials  is associated to a dominant weight of this root system.}

\begin{example} \label{ex:chebi1} Following up on \exmr{ex:roots}.
\begin{description}
  \item[{$\WeylA[1]$ :}] As we have seen in \secr{unicheb},
  these are not the classical Chebyshev polynomials strictly speaking,  
  but become these after a scaling. 
  \item[{$\WeylA[2]$} :] We can deduce from \prpr{fmlas} 
  the following recurrence formulas that allow us to write the 
multivariate Chebyshev polynomials associated to $\WeylA[2]$ in the monomial basis of $\K[X,Y]$.
We have 
\[ T_{0,0} = 6;\qquad
 T_{1,0}= X,  \quad  T_{0,1} = Y;\qquad
 4\, T_{1,1}   =  XY - 12;\]
and for $a,b>0$
\[
2\, T_{a+2,0}  = X\,T_{a+1,0}-4T_{a,1}, \qquad
2\, T_{0, b+2}   =  Y\,T_{b+1}-4T_{1,b} ;
\]
\[2\, T_{a+1,b}   =  X\, T_{a,b}- 2 T_{a,b+1}-2T_{a,b-1},
\qquad
2\, T_{a,b+1}   =  Y\, T_{a,b}- 2 T_{a+1,b-1}-2T_{a-1,b-1} .
\]

For instance
{\small \[\begin{array}{c}
T_{{2,0}}=\frac{1}{2}\,{X}^{2}-2\,Y,
\quad
T_{{1,1}}=\frac{1}{4}\,YX-3,
\quad 
T_{{0,2}}=\frac{1}{2}\,{Y}^{2}-2\,X;
\\ \medskip
T_{{3,0}}=\frac{1}{4}\,{X}^{3}-\frac{3}{2}\,YX+6,
\;
T_{{2,1}}=\frac{1}{8}\,{X}^{2}Y-\frac{1}{2}\,{Y}^{2}-\frac{1}{2}\,X,
\;
T_{{1,2}}=\frac{1}{8}\,X{Y}^{2}-\frac{1}{2}\,{X}^{2}-\frac{1}{2}\, Y,
\;
T_{{0,3}}=\frac{1}{4}\,{Y}^{3}-\frac{3}{2}\,YX+6;
\\ \medskip
T_{{4,0}}=\frac{1}{8}\,{X}^{4}-{X}^{2}Y+{Y}^{2}+4\,X,  
\quad
T_{{0,4}}=\frac{1}{8}\,{Y}^{4}-X{Y}^{2}+{X}^{2}+4\,Y,
\\ \medskip 
T_{{3,1}}=\frac{1}{16}\,{X}^{3}Y-\frac{3}{8}\,X{Y}^{2}-\frac{1}{4}\,{X}^{2}+\frac{5}{2}\,\,Y,
\quad 
T_{{1,3}}=\frac{1}{16}\,{Y}^{3}X-\frac{3}{8}\,\,{X}^{2}Y-\frac{1}{4}\,{Y}^{2}+\frac{5}{2}\,X,
\\ \medskip
T_{{2,2}}=\frac{1}{16}\,{X}^{2}{Y}^{2}-\frac{1}{4}\,{X}^{3}-\frac{1}{4}\,{Y}^{3}+YX-3.
  \end{array}\] }
    
  \item[{$\WeylB[2]$} :]  Similarly we determine
  {\small \[\begin{array}{c}
T_{{0,0}}=8; \qquad  T_{{1,0}}=X,\quad T_{{0,1}}=Y;
\\ \medskip
T_{{2,0}}=\frac{1}{2}\,{X}^{2}-{Y}^{2}+4\,X+8,
\quad
T_{{1,1}}=\frac{1}{4}\,YX-Y,
\quad 
T_{{0,2}}=\frac{1}{2}\,{Y}^{2}-2\,X-8;
\\ \medskip
T_{{3,0}}=\frac{1}{4}\,{X}^{3}-\frac{3}{4}\,X{Y}^{2}+3\,{X}^{2}+9\,X,
\;
T_{{0,3}}=\frac{1}{4}\,{Y}^{3}-\frac{3}{2}\,XY-3\,Y,
\\ \medskip
T_{{2,1}}=\frac{1}{8}\,{X}^{2}Y+\frac{3}{4}\,XY-\frac{1}{4}\,{Y}^{3}+3\,Y,
\;
T_{{1,2}}=\frac{1}{8}\,X{Y}^{2}-\frac{1}{2}\,{X}^{2}-3\,X;
\\ \medskip
T_{{4,0}}=\frac{1}{8}\,{X}^{4}-\frac{1}{2}\,{X}^{2}{Y}^{2}+2\,{X}^{3}+10\,{X}^{2}-2\,
X{Y}^{2}+\frac{1}{4}\,{Y}^{4}-4\,{Y}^{2}+16\,X+8
\quad
T_{{0,4}}=\frac{1}{8}\,{Y}^{4}-X{Y}^{2}-2\,{Y}^{2}+{X}^{2}+8\,X+8,
\\ \medskip
T_{{3,1}}=\frac{5}{8}\,{X}^{3}Y+\frac{1}{16}\,{X}^{3}Y-\frac{3}{16}\,X{Y}^{3}+\frac{3}{2}\,XY-3\,Y+\frac{1}{4}\,{Y}^{3}
,
\quad 
T_{{1,3}}=\frac{1}{16}\,{Y}^{3}X-\frac{3}{8}\,{X}^{2}Y-XY+Y,
\\ \medskip
T_{{2,2}}=X{Y}^{2}
+\frac{1}{16}\,{X}^{2}{Y}^{2}-\frac{1}{8}\,{Y}^{4}+\frac{5}{2}\,{Y}^{2}-\frac{1}{4}\,{X}^{3}-3\,{X}^{2
}-10\,X-8 .
  \end{array}\] }
\end{description}
\end{example}

\subsection{Generalized Chebyshev polynomials of the second kind} \label{cheby2}

We now describe the role that root systems play in the representation theory of semisimple Lie algebras and how the Chebyshev polynomials of the second kind arise  in this context
 \cite[Chapitre VIII, \S2,6,7]{Bourbaki_7_8}, 
 \cite[Chapter 14]{Fulton91}, \cite[Chapter 19]{Hall15}.

\begin{definition}  Let $\frakg \subset \ggl[n](\C)$ be a semisimple Lie algebra 
and let $\frakh$ be a Cartan subalgebra, that is, a maximal diagonalizable subalgebra of $\frakg$. 
 Let $\pi:\frakg \rightarrow \ggl[](W)$  be a representation of $\frakg$. 
\begin{enumerate} 
\item An element $\nu\in \frakh^*$  is called a \emph{weight} of $\pi$ if    
$W_\nu = \{w\in W \ | \ \pi(h)w = \nu(h) w \mbox{ for all } h  \in \frakh\} $   is different from $\{0\}$.  
\item The subspace $W_\nu$ of $W$ is a \emph{weight space} and the dimension of $W_\nu$ is called the {\rm multiplicity} of $\nu$ in $\pi$.
\item $\nu \in \frakh^*$ is called a \emph{weight}  if it appears as  the weight of some representation.  
\end{enumerate}
\end{definition}

 An important representation of $\frakg$ is the adjoint representation  $\ad:\frakg \rightarrow \ggl[](\frakg)$ given by $\ad(g)(h) = [g,h] = gh-hg$. 
For the adjoint representation, $\frakh$ is the weight space of $0$.
The  nonzero weights of this representation are called \emph{roots} and the set of roots is denoted by $\mR$.  Let $V$ be the {\it real} vector space spanned by $\mR$ in $\frakh^*$.  One can show that there is a unique (up to constant multiple) inner product on $V$ such that $\mR$ is a root system for $V$ in the sense of Section~\ref{primaldual:roots}   The weights of this root system are the weights defined above coming from representations of $\frakg$ so there should be no confusion in using the same term for both concepts. In particular, the weights coming from representations form a lattice.   The following is an important result concerning weights and representation. 

\begin{proposition} \label{weightprop}\cite[\S VII-5,Th\'eor\`eme 1;\S VII-12, Remarques]{Serre66}
 Let $\frakg \subset \ggl[n](\C)$ be a semisimple Lie algebra  
and  $\pi:\frakg \rightarrow \ggl[](W)$  be a representation of $\frakg$. 
Let $E = \{\mu_1,  \ldots, \mu_r\}$ be the weights of $\pi$ and let $n_i$ be the multiplicity of $\mu_i$.

\begin{enumerate} 
  \item The sum $\ds\sum _{i=1}^r n_i \mu_i \in \Lambda$ is invariant under the action of the Weyl group.
\item If $\pi$ is an irreducible representation then there is a unique $\mu \in E$ such that $\mu \succ \mu_i$ for $i = 1, \ldots, r$. This weight is called the {\rm highest weight of $\pi$} and is a dominant weight for $\mR$. Two irreducible representations are isomorphic if and only if they have the same highest weight. 
\item Any dominant weight $\mu$ for $\mR$ appears as the highest weight  of an irreducible representation of $\frakg$.
\end{enumerate}
\end{proposition}

Note that property 1.~implies that all weights in the same Weyl group orbit appear with the same multiplicity and so this sum  is an integer combination of Weyl group orbits.


In the usual expositions   one denotes   
a basis of the group ring $\Z[\Lambda]$ by
$\{e^\mu \ | \ \mu \in \Lambda\}$ \cite[Chapitre
VIII, \S9.1]{Bourbaki_7_8} or $\{e(\mu) \ | \ \mu \in \Lambda\}$
(\cite[\S24.3]{Humphreys72}) where $e^\mu\cdot e^\lambda =
e^{\mu+\lambda}$ or  $e(\mu)\cdot e(\lambda) = e({\mu+\lambda})$.
With the  conventions introduced in the previous section, we define the character polynomial and state Weyl's character formula. 

\begin{definition} \label{Cheb2Def} Let $\omega$ be a dominant weight. 
The \emph{character polynomial} associated to $\omega$ 
is the polynomial in $\Z[x,x^{-1}]$
\[\cha{\omega} = \sum_{\lambda \in \Lambda_\omega} n_\lambda x^\lambda\]
where 
 $\Lambda_\omega$ is the set of weights for the irreducible
  representation associated with $\omega$ and 
$n_\lambda$ is the multiplicity of $\lambda$ in this representation.
\end{definition}

From Proposition~\ref{weightprop} and the comment following it, one sees that $\cha{\alpha} = \sum_{\beta \prec \alpha} n_\beta \orb{\beta}$. 
Here we abuse notation and include all $\orb{\beta}$ with $\beta \prec \alpha$ 
even if $\beta \notin \Lambda_\alpha$ in which case we let $n_\beta = 0$.

\begin{theorem} \label{weyl} (Weyl character formula)
$\delta = \frac{1}{2} \sum_{\rho \in \mR^+} \rho$ is a strongly dominant weight 
and 
\begin{equation*}
\aorb{\delta} \, \cha{\omega} =  \aorb{\omega+\delta}
\quad \hbox{ where } \aorb{\alpha} =\sum_{B\in\gva} \ \det(B)\,x^{B\alpha}
\end{equation*}
\end{theorem}

The earlier cited
\cite[Chapitre VI, \S3.3 Th\'eor\`eme 1]{Bourbaki_4_5_6} 
that provided \prpr{prop:invariants} allows the following definition of the
 generalized Chebyshev polynomials of the  second kind. 

\begin{definition}\label{ChebDef2} Let $\omega$ be a  dominant weight.
 The \emph{Chebyshev polynomial of the second kind associated to $\omega$} 
 is the polynomial $U_\omega$ in $\K[X]=\K[X_1,\ldots X_n]$ such that 
 $\cha{\omega} = U_\omega(\orb{\omega_1}, \ldots , \orb{\omega_n})$.
\end{definition}

This is the definition proposed in \cite{Munthe-Kaas13}. 
In \cite{LyUv2013}, the Chebyshev polynomial of the second kind are 
defined as the polynomial $\tilde{U}_\omega$ such that 
$\cha{\omega} = \tilde{U}_\omega(\cha{\omega_1}, \ldots , \cha{\omega_n})$. 
This is made possible thanks to \cite[Chapitre VI, \S3.3 Th\'eor\`eme 1]{Bourbaki_4_5_6} 
that also provides the following result.

\begin{proposition}\label{prop:invariants2} 
Let $\{\omega_1, \ldots , \omega_n\}$ be the
  fundamental weights.
\begin{enumerate}
\item 
  $\{\cha{\omega_1}, \ldots , \cha{\omega_n}\}$ is an algebraically
  independent set of invariant Laurent polynomials.
\item $\Z[x,x^{-1}]^\WeylG = \Z[\cha{\omega_1}, \ldots , \cha{\omega_n}]$
\end{enumerate}
\end{proposition}

One sees from \cite[Chapitre VI, \S3.3 Th\'eor\`eme 1]{Bourbaki_4_5_6} 
that an invertible affine map takes the basis 
$\{\orb{\omega_1}, \ldots , \orb{\omega_n}\}$ 
to the basis $\{\cha{\omega_1}, \ldots , \cha{\omega_n}\}$ 
so results using one definition can easily be applied to situations using the other definition.
The sparse interpolation algorithms to be presented in this article can also be directly  modified to  work for this latter  definition as well. The only change is in \algr{SKInterp} where the evaluation points 
should be 
$$\left(\cha{\omega_1}(\xi^{\tr{\alpha}S}),\ldots,\cha{\omega_n}(\xi^{\tr{\alpha}S})
\right) 
\quad \hbox{ instead of }\quad
\left(\orb{\omega_1}(\xi^{\tr{\alpha}S}),\ldots,\orb{\omega_n}(\xi^{\tr{\alpha}S})
\right).$$

As with  Chebyshev polynomials of the first kind, we shall usually drop the phrase ``associated to $\omega$'' and just refer to Chebyshev polynomials of the second  kind with the understanding that we have fixed a root systems and  each of these polynomials  is associated to a dominant weight of this root system.

\begin{example} \label{ex:chebi2} Following up on \exmr{ex:roots}
\begin{description}
  \item[{$\WeylA[1]$ :}] As we have seen in \secr{unicheb}, the Chebyshev polynomials of the second kind associated to $\WeylA[1]$ are the 
   classical Chebyshev polynomials of the second kind    after a scaling. 
  \item[{$\WeylA[2]$} :] We can deduce from \prpr{fmlas} (as done in the proof of Proposition~\ref{gooseberry}) the following recurrence formulas that allow us to write the 
multivariate Chebyshev polynomials associated to $\WeylA[2]$ in the monomial basis of $\K[X,Y]$.
We have $U_{0,0} = 1$
and, for  $a,b\geq 1$,
for  $a,b\geq 1$,
\begin{eqnarray*}
 2\,U_{{a+1,0}}=XU_{{a,0}}-2\,U_{{a-1,1}}, & \qquad & 
 2\,U_{{0,b+1}}=YU_{{0,b}}-2\,U_{{a+1,b-1}}
\\
 2\,U_{{a+1,b}}=XU_{{a,b}}-2\,U_{{a-1,b+1}}-2\,U_{{a,b-1}}, & \qquad & 
 2\,U_{{a,b+1}}=YU_{{a,b}}-2\,U_{{a+1,b-1}}-2\,U_{{a-1,b}}
\end{eqnarray*}
For instance
{\small \[\begin{array}{c}
U_{{1,0}}=\frac{1}{2}\,X,\quad
U_{{0,1}}=\frac{1}{2}\,Y;
\\\medskip
U_{{2,0}}=\frac{1}{4}\,{X}^{2}-\frac{1}{2}\,Y,\quad
U_{{1,1}}=\frac{1}{4}\,XY-1,\quad
U_{{0,2}}=\frac{1}{4}\,{Y}^{2}-\frac{1}{2}\,X;
\\\medskip
U_{{3,0}}=\frac{1}{8}\,{X}^{3}-\frac{1}{2}\,XY+1,\quad
U_{{0,3}}=\frac{1}{8}\,{Y}^{3}-\frac{1}{2}\,XY+1,
\\\medskip
U_{{2,1}}=\frac{1}{8}\,{X}^{2}Y-\frac{1}{4}\,{Y}^{2}-\frac{1}{2}\,X,\quad
U_{{1,2}}=\frac{1}{8}\,X{Y}^{2}-\frac{1}{4}\,{X}^{2}-\frac{1}{2}\,Y;
\\\medskip
U_{{4,0}}=\frac{1}{16}\,{X}^{4}-\frac{3}{8}\,{X}^{2}Y+\frac{1}{4}\,{Y}^{2}+X,\quad
U_{{0,4}}=\frac{1}{16}\,{Y}^{4}-\frac{3}{8}\,X{Y}^{2}+\frac{1}{4}\,{X}^{2}+Y,
\\\medskip
U_{{3,1}}=\frac{1}{16}\,{X}^{3}Y-\frac{1}{4}\,X{Y}^{2}-\frac{1}{4}\,{X}^{2}+Y,\quad
U_{{1,3}}=\frac{1}{16}\,X{Y}^{3}-\frac{1}{4}\,{X}^{2}Y-\frac{1}{4}\,{Y}^{2}+X,
\\\medskip
U_{{2,2}}=\frac{1}{16}\,{X}^{2}{Y}^{2}-\frac{1}{8}\,{X}^{3}-\frac{1}{8}\,{Y}^{3}.
  \end{array}\] }
    
\item[{$\WeylB[2]$} :]  Similarly we determine
 {\small \[\begin{array}{c}
U_{{0,0}}=1;\quad
U_{{1,0}}=\frac{1}{2}\,X-1,\quad
U_{{0,1}}=\frac{1}{2}\,Y;
\\\medskip
U_{{2,0}}=\frac{1}{4}\,{X}^{2}-\frac{1}{2}\,X-\frac{1}{4}\,{Y}^{2},\quad
U_{{1,1}}=\frac{1}{4}\,XY-Y,\quad
U_{{0,2}}=\frac{1}{4}\,{Y}^{2}-\frac{1}{2}\,X;
\\\medskip
U_{{3,0}}=\frac{1}{8}\,{X}^{3}
,\quad
U_{{0,3}}=\frac{1}{8}\,{Y}^{3}-\frac{1}{2}\,XY+\frac{1}{2}\,Y,
\\\medskip
U_{{2,1}}=\frac{1}{8}\,{X}^{2}Y-\frac{1}{8}\,{Y}^{3}-Y,\quad
U_{{1,2}}=\frac{1}{8}\,X{Y}^{2}-\frac{1}{2}\,{Y}^{2}-\frac{1}{4}\,{X}^{2}+\frac{1}{2}\,X+1;
\\\medskip
U_{{4,0}}=\frac{1}{16}\,{X}^{4}-\frac{1}{8}\,{X}^{3}-\frac{3}{16}\,{X}^{2}{Y}^{2}+\frac{1}{2}\,X{Y}^{
2}-\frac{1}{2}\,{X}^{2}+\frac{1}{16}\,{Y}^{4}+1+\frac{1}{2}\,X
,\quad
U_{{0,4}}=\frac{1}{16}\,{Y}^{4}-\frac{3}{8}\,X{Y}^{2}+\frac{1}{2}\,{Y}^{2}+\frac{1}{4}\,{X}^{2}-1
,\\\medskip
U_{{3,1}}=\frac{1}{16}\,{X}^{3}Y-\frac{1}{8}\,X
{Y}^{3}-\frac{5}{4}\,XY+\frac{1}{4}\,{Y}^{3}+\frac{1}{4}\,{X}^{2}Y-Y
,\quad
U_{{1,3}}=\frac{1}{16}\,X{Y}^{3}-\frac{1}{4}\,{Y}^{3}-\frac{1}{4}\,{X}^{2}Y+\frac{1}{2}\,
XY+2\,Y
,\\\medskip
U_{{2,2}}=\frac{1}{16}\,{X}^{2}{Y}^{2}-\frac{1}{16}\,{Y}^{4}-\frac{1}{2}\,{Y}^{2}-\frac{1}{8}\,{X}^{3}+\frac{1}{4}\,{X}^{2}+\frac{1}{8}\,X{Y}^{2}+\frac{1}{2}\,X-1
.
  \end{array}\]
  }
\end{description}
\end{example}

We note that the elements $\aorb{\alpha}$ appearing in Theorem~\ref{weyl} are {\it not} invariant polynomials but are skew-symmetric polynomials, that is, polynomials $p$ such that $A\cdot p = \det(A)p$. The $\K$-span of  all such polynomials form a module over $\IKx$ which has a nice description.

\begin{theorem} \label{iraty} \cite[Ch.~VI,\S3,Proposition 2]{Bourbaki_4_5_6}
With $\delta = \frac{1}{2} \sum_{\rho \in \mR^+} \rho$, the map
\[\begin{array}{rcl} \IKx & \rightarrow & \Kx\\
p &\mapsto & \aorb{\delta}\, p \end{array}\]
is a $\IKx$-module isomorphism between $\IKx$   and the $\IKx$-module of skew-symmetric polynomials.
\end{theorem}

This theorem allows us to denote the module of skew-symmetric polynomials by $\aorb{\delta}\IKx$.

\subsection{Orders} \label{additional}

In this section we gather properties 
about generalized Chebyshev polynomials 
that relate to  orders on $\N^n$.
They are  needed in some of the proofs that underlie 
the sparse interpolation  algorithms developed in this article. 

\begin{proposition} \label{raspberry}
For any $\alpha,\beta\in\N^n$ there exist some 
$a_\nu\in \N$ with $a_{\alpha+\beta}\neq 0$ such that
\[ \orb{\alpha} \, \orb{\beta}= 
\sum_{\substack{\nu\in \N^n\\ 
\nu\prec \alpha+\beta }} a_{\nu}  \orb{\nu},
\quad
\aorb{\alpha} \, \orb{\beta}=  
\sum_{\substack{\nu\in \N^n\\ \nu\prec \alpha+\beta } } a_{\nu}  \aorb{\nu}.
\]
and the cardinality of the supports 
$\left\{\nu  \in \N^n \, |\, \nu \prec \alpha+\beta,\;  \hbox{ and } a_{\nu}\neq  0 \right\}$
is at most $|\gva|$.
\end{proposition}

\begin{proof} From \prpr{fmlas} we have 
\(\orb{\alpha} \, \orb{\beta}= \sum_{B\in\gva}  \orb{\alpha+B\beta}\),
\( \aorb{\alpha} \, \orb{\beta}= \sum_{B\in\gva}  \aorb{\alpha+B\beta}.\)
If  $\mu\in \N^n$ is the unique dominant weight in the orbit of 
$\alpha+B\beta$ then 
$\orb{\alpha+ B\beta}=\orb{\mu}$ and
$\aorb{\alpha+ B\beta}=\aorb{\mu}$.
We next prove that $\mu\prec \alpha+\beta$.

Let $A\in\gva$ be such that $A(\alpha+B\beta)=\mu$.
Since $A, AB \in \gva$ we have $A\alpha \prec \alpha$ and 
$AB\beta \prec \beta$ (Proposition~\ref{domprop}.2).
Therefore $A\alpha = \alpha - \sum m_i\rho_i$ 
an $AB\beta= \beta - \sum n_i\rho_i$ for some $m_i,n_i \in \N$. 
This implies
\[ A(\alpha+ B\beta) = A\alpha + AB\beta = \alpha - \sum m_i\rho_i + \beta - \sum n_i\rho_i = \alpha + \beta - \sum(m_i+n_i)\rho_i\]
so $\mu=A(\alpha+ B\beta) \prec \alpha+\beta $.
\end{proof}


\begin{proposition} \label{gooseberry}
For all $\alpha\in \N^n$, 
$\ds T_\alpha = \sum_{\beta\prec \alpha} t_\beta X^\beta$  and
$\ds U_\alpha = \sum_{\beta\prec \alpha} u_\beta X^\beta$  
where $t_\alpha\neq 0$ and $u_\alpha \,\neq 0$.
\end{proposition}

\begin{proof} 
Note that the $\isp{\omega_i}{\omega_j}$ are nonnegative rational numbers \prpr{domprop}.  Therefore the set of nonnegative integer combinations of these rational numbers forms a well ordered subset of the rational numbers. This allows us to proceed by induction on $\langle \delta,\alpha \rangle$ to prove the first statement of the above proposition.

Consider $\delta = \frac{1}{2} \sum_{\rho \in \mR^+} \rho =\sum_{i=1}^r \omega_i$ (\prpr{domprop}). 
As a strongly dominant weight $\delta$ satisfies $\langle \delta , \rho \rangle > 0$ for all $\rho\in \mR^+$.
Furthermore, for any dominant weight $\omega \neq 0$,
\( \langle \delta, \omega \rangle > 0\)
since $\langle \rho ,\omega\rangle \geq 0$ for all 
$\rho\in \mR^+$, with at least one inequality being a strict inequality.
Hence $ \langle \delta, \alpha -\omega_i\rangle <\langle \delta, \alpha \rangle $.

The property is true for $T_0$ and $U_0$. 
Assume it is true for all $\beta\in\N^n$ such that 
$\langle \delta , \beta \rangle < \langle \delta, \alpha\rangle$, $\alpha\in\N^n$.
There exists $1\leq i\leq n$ such that $\alpha_i \geq 1$.
By \lemr{fmlas},
$\orb{\omega_i}\orb{\alpha-\omega_i} = \sum_{\nu \prec \alpha} a_{\nu} \orb{\nu}$ with $a_\alpha \neq 0$.
Hence 
$a_\alpha \,T_\alpha = X_i T_{\alpha-\omega_i} - 
\sum_{\substack{\nu \prec \alpha\\ \nu\neq \alpha}} a_{\nu} T_{\nu}$.
Since $\nu \prec \alpha$, $\nu\neq \alpha$, implies that $\isp{\delta}{\nu} < \isp{\delta}{\alpha}$,
the property thus holds by recurrence for $\left\{T_\alpha\right\}_{\alpha\in\N}$.

 By \prpr{weightprop},
$\cha{\alpha}$ is invariant under the action of the Weyl group.  Furthermore, any orbit of the Weyl group will contain a unique highest  weight.  Therefore 
$\cha{\alpha}=\sum_{\beta\prec \alpha} n_{\beta} \orb{\beta}$ with $n_{\alpha}\neq 0$.
 Hence $U_\alpha = \sum_{\beta \prec \alpha} n_\beta T_\alpha$ and so the result follows from the above.
The property holds  for $\left\{U_\alpha\right\}_{\alpha\in\N}$ as it holds for $\left\{T_\alpha\right\}_{\alpha\in\N}$.
\end{proof}


The following result shows that the partial order $\prec$ can be extended to an admissible order on $\N^n$. Admissible order on $\N^n$ define  term orders on the polynomial ring $\K[X_1,\ldots,X_n]$ 
upon which Gr\"obner bases can be defined \cite{Becker93,Cox15}.
In the proofs of Sections~\ref{sec:sparse} and~\ref{sec:support}
some arguments stem from there. 

\begin{proposition} \label{termorder}
Let $\mB=\left\{\rho_1,\ldots,\rho_n\right\}$ be the base for $\mR$
and consider $\delta = \frac{1}{2} \sum_{\rho \in \mR^+} \rho$. 
Define the relation $\leq$ on $\N^n$ by 
\[ \alpha \leq \beta \; \Leftrightarrow\; \left\{\begin{array}{l}
   \isp{\delta}{\alpha} < \isp{ \delta}{\beta}  \quad \hbox{ or } \\
   \isp{\delta}{\alpha} = \isp{ \delta}{\beta}  \hbox{ and }
     \isp{\rho_2}{\alpha} < \isp{ \rho_2}{\beta}  \quad \hbox{ or }\\
     \vdots \\
      \isp{\delta}{\alpha}  = \isp{ \delta}{\beta}  \hbox{ and }
     \isp{\rho_2}{\alpha}  = \isp{ \rho_2}{\beta} , \ldots,
     \isp{\rho_{n-1}}{\alpha}  = \isp{ \rho_{n-1}}{\beta} ,
    \isp{\rho_{n}}{\alpha} < \isp{\rho_{n}}{\beta} \quad \hbox{ or }\\
     \alpha = \beta \end{array}\right.\]

Then $\leq$ is an admissible order on  $\N^n$,
that is, for any $\alpha,\beta,\gamma\in \N^n$
\[ \tr{\begin{bmatrix} 0 & \ldots & 0\end{bmatrix}}
   \leq \gamma, \quad \hbox{ and }
 \alpha \leq  \beta \;\Rightarrow \; \alpha+\gamma \leq \beta+\gamma.
 \]
Furthermore $\alpha \prec \beta \;\Rightarrow \; \alpha\leq \beta .$
\end{proposition}

\begin{proof} 
We have that  $\isp{\rho}{\alpha}\geq  0$  
for all $\rho\in\mB$  and $\alpha\in\N^n$. 
Hence, since $\delta = \frac{1}{2} \sum_{\rho \in \mR^+} \rho$,
$\isp{\delta}{\alpha} >0$ for all dominant weights $\alpha$.
Furthermore since $\left\{\rho_1,\rho_2\ldots, \rho_n\right\}$ 
is a basis for $V=\K^n$,
so is $\left\{\delta,\rho_2,\ldots, \rho_n\right\}$.
Hence $\leq$ is an admissible order.

We have already seen that  $\delta$ 
is a strongly dominant weight (\prpr{domprop}). 
As such $\isp{\delta}{\rho} > 0$ for all $\rho \in \mR^+$. 
Hence, if $\alpha\prec \beta$, with $\alpha\neq\beta$, 
then $\beta=\alpha-m_1\rho_1-\ldots-m_n\rho_n$ with $m_i\in\N$, 
at least one positive, so that $\isp{\delta}{\alpha}<\isp{\delta}{\beta}$ 
and thus $\alpha \leq \beta$.
\end{proof}




\subsection{Determining Chebyshev polynomials from their values}   \label{testweight} 

The algorithms for recovering the support of a linear combination of generalized Chebyshev polynomials will first determine the values 
      $\orb{\omega} \left( \xi^{\tr{\alpha}S}\right)$
  for certain  $\alpha$ but for unknown $\omega$. To complete the determination, we will need to determine  $\omega$.
 We will show below  that if $\{\mu_1, \ldots, \mu_n\}$ are strongly dominant weights that form a basis of the ambient vector space $V$, then one can 	 choose an integer $\xi$ that allows one to 
effectively determine $\omega$ from the values 
 \[\left(\, \orb{\omega} \left( \xi^{\tr{\mu_i}S}\right) 
         \;|\;  1\leq i\leq n\right)\]
 or from the values
  \[\left(\, \cha{\omega} \left( \xi^{\tr{\mu_i}S}\right) 
         \;|\;  1\leq i\leq n\right)\]

We begin with two facts concerning strongly dominant weights which are crucial in what follows. \begin{itemize}
\item If $\mu_1$ and $\mu_2$ are dominant weights, then $\langle \mu_1, \mu_2\rangle \geq 0$ (Proposition~\ref{domprop}).  
\item If $\mB$ is a base of the roots, $\rho \in \mB$ and $\mu$ is a strongly dominant weight, then $\langle \mu, \rho\rangle >0$.  This follows from the facts that  $\langle \mu, \rho^*\rangle >0$ by definition and that $\rho^*$ is a positive multiple of $\rho$.
\end{itemize}
         
Also recall our convention (stated at the end of 
Section~\ref{primaldual:roots})  that the
entries of $S$ are  in $\Q$. We shall denote by $D$ their least common denominator. Note that with this notation we have that $D \langle \mu, \nu\rangle$ is an integer for any weights $\mu, \nu$.

 \begin{lemma} \label{falafel} 
 Let  $\mu$ be a strongly dominant weight and let $\xi = \xi_0^D$ where $\xi_0 \in \N$ satisfies 
 \[ \xi_0 > (\frac{3}{2}|\WeylG|)^2.\]
 \begin{enumerate}
 \item If $\omega$ be is a dominant weight then 
\[ D\cdot\langle \mu, \omega\rangle = \lfloor \log_{\xi_0}(\orb{\omega}(\xi^{\tr{\mu} S}))\rfloor\]
where $\lfloor \cdot \rfloor$ is the usual floor function.
\item If $\omega$ is  a \underline{strongly} dominant weight then 
%
\[ D\cdot \langle \mu,\omega \rangle={\rm nint}[\log_{\xi_0}(\aorb{\omega}(\xi^{\tr{\mu} S})]\]
where {\rm nint} denotes the nearest integer\footnote{in the proof we show that the distance to the nearest integer is less than $\frac{1}{2}$ so this is well defined.}
\end{enumerate}\end{lemma}

\begin{proof}

1.~Let $s$ be the size of the stabilizer of $\omega$ in  $\WeylG$. We have the following
\begin{eqnarray*}
\orb{\omega} \left( \xi^{\tr{\mu}S}\right) & = & \sum_{\sigma\in\WeylG}   \xi^{\tr{\mu}S\sigma(\omega)} = \sum_{\sigma\in\WeylG}   \xi^{\langle\mu, \, \sigma(\omega)\rangle} \\
 & = & s \sum_{\sigma\in C}   \xi_0^{D\langle\mu, \, \sigma(\omega)\rangle} \ \ \mbox{ where $C$ is a set of coset representatives of $\WeylG/{\mathrm Stab}(\omega)$.}\\
 & = &s\xi_0^{D\langle\mu, \, \omega \rangle}  (1+\sum_{\sigma \neq 1, \sigma\in C}   \xi_0^{D\langle\mu, \, \sigma(\omega) - \omega\rangle})
 \end{eqnarray*}
 We now use the fact that  for $\sigma \in \WeylG$, 
 $\sigma(\omega) - \omega = -\sum_{\rho \in \mB} n_{\rho}^\sigma \rho$ 
for some nonnegative integers 
$n_{\rho}^\sigma $. If $\sigma \in C, \sigma \neq 1$ 
we have that not all the $n_{\rho}^\sigma $ are zero. Therefore we have
 \begin{eqnarray*}
\orb{\omega} \left( \xi^{\tr{\mu}S}\right) & = &
s \xi_0^{D\langle\mu, \, \omega \rangle}  
(1+\sum_{\sigma \neq 1,\sigma \in C}
         \xi_0^{D\langle\mu,\,-\sum_{\rho \in \mB} n_{\rho}^\sigma \rho\rangle})
\\
& = & 
s\xi_0^{D\langle\mu, \, \omega \rangle}  (1+\sum_{\sigma \neq 1,\sigma \in C} \xi_0^{-m_\sigma}),
\end{eqnarray*}
where each $m_\sigma$ is a positive integer.  This follows from the fact that $D\langle \mu, \rho\rangle$ is always a positive integer for $\mu$ a strongly dominant  weight and $\rho \in \mB$. It is now immediate that 
\begin{eqnarray}\label{eq:orb1}
s\xi_0^{D\langle\mu, \, \omega \rangle} &\leq &\orb{\omega} \left( \xi^{\tr{\mu}S}\right).
\end{eqnarray}
Since  $\xi_0 > (\frac{3}{2}|\WeylG|)^2 >  \frac{9}{4}|\WeylG|$ 
we have
\[ 1+\sum_{\sigma \neq 1,\sigma \in C}   \xi_0^{-m_\sigma} \leq 1 + |\WeylG| \xi_0^{-1} < \frac{3}{2}\]
and so
\begin{eqnarray}\label{eq:orb2}
s\xi_0^{D\langle\mu, \, \omega \rangle}  (1+\sum_{\sigma \neq 1,\sigma \in C}   \xi_0^{-m_\sigma}) < \frac{3}{2}s\xi_0^{D\langle\mu, \, \omega \rangle} 
\end{eqnarray}
To prove the final claim, apply $\log_{\xi_0}$ to (\ref{eq:orb1}) and (\ref{eq:orb2}) to yield 
\[  D\langle \mu, \omega\rangle   +\log_\xi s \leq \log_{\xi_0} (\orb{\omega} \left( \xi_0^{\tr{\mu}S}\right)) \leq D\langle \mu, \omega\rangle   +\log_{\xi_0} (\frac{3}{2}s)\]
Using the hypothesis on the lower bound for $\xi_0$, we have 
\[ s \leq |\WeylG| < \frac{2}{3}(\xi_0)^{1/2} \Longrightarrow \log_{\xi_0} (\frac{3}{2}s) < \frac{1}{2}.\]
Therefore
\[  D\langle \mu, \omega\rangle  \leq \log_\xi (\orb{\omega} \left( \xi^{\tr{\mu}S}\right)) <D\langle \mu, \omega\rangle   +\frac{1}{2}\]
which yields the final claim. 

2.~Since $\omega$ is a stongly dominant weight, we have that for any $\sigma \in \WeylG, \omega \prec \sigma(\omega)$ and $\sigma(\omega) = \omega$ if and only if  $\sigma$ is the identity (c.f. Proposition~\ref{domprop}).  In particular, the stabilizer of $\omega$ is trivial.  The proof begins in a similar manner as above.

  We have 
\begin{eqnarray*}
\aorb{\omega} \left( \xi^{\tr{\mu}S}\right) & = & \sum_{\sigma\in\WeylG}  \det(\sigma) \xi^{\tr{\mu}S\sigma(\omega)} = \sum_{\sigma\in\WeylG} \det(\sigma)  \xi^{\langle\mu, \, \sigma(\omega)\rangle} \\
  & = &\xi_0^{D\langle\mu, \, \omega \rangle}  (1+\sum_{\sigma \neq 1, \sigma\in \WeylG} \det(\sigma)  \xi_0^{D\langle\mu, \, \sigma(\omega) - \omega\rangle})
 \end{eqnarray*}
 We now use the fact that  for $\sigma \in \WeylG$, $\sigma(\omega) - \omega = -\sum_{\rho \in \mB} n_{\rho}^\sigma \rho$ 
for some nonnegative integers $n_{\rho}^\sigma $. If $\sigma \in \WeylG, \sigma \neq 1$ we have that not all the $n_{\rho}^\sigma $ are zero. Therefore we have
 \begin{eqnarray*}
\aorb{\omega} \left( \xi^{\tr{\mu}S}\right) & = & \xi^{\langle\mu, \, \omega \rangle}  (1+\sum_{\sigma \neq 1,\sigma \in \WeylG} \det(\sigma)  \xi^{\langle\mu, \, -\sum_{\rho \in \mB} n_{\rho}^\sigma \rho\rangle
})\\
& = & \xi_0^{D\langle\mu, \, \omega \rangle}  (1+\sum_{\sigma \neq 1,\sigma \in |\WeylG|}  \det(\sigma)   \xi_0^{-m_\sigma}),
\end{eqnarray*}
where each $m_\sigma$ is a positive integer.  This again follows from the fact that $D\langle \mu, \rho\rangle$ is always a positive integer for $\mu$ a strongly dominant  weight and $\rho \in \mB$. At this point the proof diverges from the proof of 1. Since for any $\sigma \in \WeylG, \det(\sigma) = \pm 1$ we have 
\begin{eqnarray*}
1-|\WeylG| \xi_0^{-1} \ \leq  \ 1-\sum_{\sigma \neq 1,\sigma \in \WeylG}  \xi_0^{-m_\sigma} &\leq & 1+\sum_{\sigma \neq 1,\sigma \in \WeylG}  \det(\sigma)   \xi_0^{-m_\sigma} \   \leq  \  1+\sum_{\sigma \neq 1,\sigma \in \WeylG}     \xi_0^{-m_\sigma} \leq 1+ |\WeylG| \xi_0^{-1}.
\end{eqnarray*}
Therefore
\begin{eqnarray} \label{eq:aorb1}
 \xi_0^{D\langle\mu, \, \omega \rangle} (1-|\WeylG| \xi_0^{-1})& \leq &  \aorb{\omega} \left( \xi^{\tr{\mu}S}\right) \ \leq \  \xi_0^{D\langle\mu, \, \omega \rangle}(1+|\WeylG| \xi_0^{-1}).
\end{eqnarray}

We will now show that $1-|\WeylG| \xi_0^{-1} > \xi_0^{-\frac{1}{4}}$ and  $1+|\WeylG| \xi_0^{-1}< \xi_0^{\frac{1}{4}}$.

\underline{$1-|\WeylG| \xi_0^{-1} > \xi_0^{-\frac{1}{4}}$:} This is equivalent to $\xi_0 - \xi_0^\frac{3}{4} = \xi_0(1-\xi_0^{-\frac{1}{4}}) > |\WeylG|$.  Since $\xi_0 >  (\frac{3}{2}|\WeylG|)^2$, it is enough to show that $1-\xi_0^{-\frac{1}{4}} > \frac{4}{9}|\WeylG|^{-1}$. To achieve this it suffices to show $1-(\frac{3}{2} |\WeylG|)^{-\frac{1}{2}} > \frac{4}{9}|\WeylG|^{-1}$ or equivalently, that $f(x) = x - (\frac{2}{3} x)^\frac{1}{2} -  \frac{4}{9} >0$ when $x\geq2$. Observing that $f(2) >0$ and that $f'(x) = 1-\frac{1}{2}(\frac{2}{3} x)^{-\frac{1}{2}} > 0$ for all $x \geq 2$ yields this latter conclusion.

\underline{$1+|\WeylG| \xi_0^{-1}< \xi_0^{\frac{1}{4}}$:} This is equivalent to $\xi_0^{\frac{5}{4}} - \xi_0= \xi_0(\xi_0^\frac{1}{4} -1)> |\WeylG| $. In a similar manner as before, it suffices to show $\frac{9}{4} |\WeylG|^2(\xi_0^\frac{1}{4} -1)> |\WeylG| $ or $\xi_0^\frac{1}{4} -1> \frac{4}{9}|\WeylG|^{-1}$.  To achieve this it suffices to show $\frac{3}{2}^\frac{1}{2} |\WeylG|^\frac{1}{2} - 1 > \frac{4}{9}|\WeylG|^{-1}$ or equivalently, $f(x) = {\frac{3}{2}}^\frac{1}{2} x^\frac{3}{2} - x - \frac{4}{9} >0$ for $x \geq 2$. Observing that $f(2) > 0 $ and $f'(x) >0 $ for all $x \geq 2$ yields the latter conclusions.

Combining these last two inequalities, we have
\begin{eqnarray*}
 \xi_0^{D\langle\mu, \, \omega \rangle} \xi_0^{-1/4} \ <  \xi_0^{D\langle\mu, \, \omega \rangle} (1-|\WeylG| \xi^{-1})& \leq &  \aorb{\omega} \left( \xi^{\tr{\mu}S}\right) \ \leq  \xi_0^{D\langle\mu, \, \omega \rangle}(1+|\WeylG| \xi^{-1}) \ < \ \xi_0^{D\langle\mu, \, \omega \rangle} \xi^{1/4} .
\end{eqnarray*}
Taking logarithms base $\xi_0$, we have
\begin{eqnarray*}
D\langle\mu, \, \omega \rangle -\frac{1}{4} &<&  \log_{\xi_0} (\aorb{\omega} \left( \xi^{\tr{\mu}S}\right)) \ < \ D\langle\mu, \, \omega \rangle +\frac{1}{4}
\end{eqnarray*}
which yields the conclusion of 2.
\end{proof}

The restriction in 2.~that $\omega$ be a strongly dominant weight is necessary  
as $\aorb{\omega}=0$ when $\omega$ belongs to the walls of the Weyl chamber 
 \cite[Ch.~VI,\S3]{Bourbaki_4_5_6}.
Furthermore,
the proof of Lemma~\ref{falafel}.2 yields the following result 
which is needed in \algr{SKInterp}. 
\begin{corollary}\label{falafelcor} If $\beta$ is a dominant weight and $\xi = \xi_0^D$ with $\xi_0 \in \N$ and $\xi_0 > |\WeylG|$,   then $\aorb{\delta}(\xi^{\tr{(\delta+\beta)}S})\neq 0$.
\end{corollary}
\begin{proof} Note that both $\delta$ and $\delta + \beta$ are strongly dominant weights.  As in the proof of Lemma~\ref{falafel}.2, we have 
\begin{eqnarray*}
\aorb{\delta} \left( \xi^{\tr{(\delta+\beta)}S}\right) & = & \sum_{\sigma\in\WeylG}  \det(\sigma) \xi^{\tr{(\delta+\beta)}S\sigma(\delta)} = \sum_{\sigma\in\WeylG} \det(\sigma)  \xi^{\langle \delta+\beta, \, \sigma(\delta)\rangle} \\
  & = &\xi_0^{D\langle\delta+\beta, \, \delta \rangle}  (1+\sum_{\sigma \neq 1, \sigma\in \WeylG} \det(\sigma)  \xi_0^{D\langle\delta+\beta, \, \sigma(\delta) - \delta\rangle})
 \end{eqnarray*}
 We now use the fact that  for $\sigma \in \WeylG$, $\sigma(\delta) - \delta = -\sum_{\rho \in \mB} n_{\rho}^\sigma \rho$ 
for some nonnegative integers $n_{\rho}^\sigma $. If $\sigma \in \WeylG, \sigma \neq 1$ we have that not all the $n_{\rho}^\sigma $ are zero. Therefore we have
 \begin{eqnarray*}
\aorb{\delta} \left( \xi^{\tr{(\delta+\beta)}S}\right) & = & \xi^{\langle\delta+\beta, \, \delta \rangle}  (1+\sum_{\sigma \neq 1,\sigma \in \WeylG} \det(\sigma)  \xi^{\langle\delta+\beta, \, -\sum_{\rho \in \mB} n_{\rho}^\sigma \rho\rangle
})\\
& = & \xi_0^{D\langle\delta+\beta, \, \delta \rangle}  (1+\sum_{\sigma \neq 1,\sigma \in |\WeylG|}  \det(\sigma)   \xi_0^{-m_\sigma}),
\end{eqnarray*}
where each $m_\sigma$ is a positive integer.  This follows from the fact that $D\langle \delta+\beta, \rho\rangle$ is always a positive integer $\rho \in \mB$ since $\delta+\beta$ is a strongly dominant weight.  Therefore we have 
\[\aorb{\delta} \left( \xi^{\tr{(\delta+\beta)}S}\right)  = \xi_0^{D\langle\delta+\beta, \, \delta \rangle}  (1+\sum_{\sigma \neq 1,\sigma \in |\WeylG|}  \det(\sigma)   \xi_0^{-m_\sigma}) \geq  \xi_0^{D\langle\delta+\beta, \, \delta \rangle}(1 - |\WeylG|\xi_0^{-1}) > 0.\] \end{proof}

\begin{theorem} \label{kumquat}  Let $\{\mu_1, \ldots , \mu_n\}$ be a basis of strongly 
dominant weights
 and let  $\xi = (\xi_0)^{D}$ with $\xi_0 > (\frac{3}{2}|\WeylG|)^2$. One can effectively determine the dominant weight $\omega$ from either of the sets of the numbers
\begin{eqnarray}\label{eqn:weq1}
\{\orb{\omega} \left( \xi^{\tr{\mu_i}S}\right) \;|\;  i=1,  \ldots , n\} 
\quad \mbox{ or } \quad 
\{\cha{\omega}\left( \xi^{\tr{\mu_i}S}\right) \;|\;  i=1,  \ldots , n\}
\end{eqnarray}
\end{theorem}

\begin{proof} Lemma~\ref{falafel}.1 allows us to determine the rational numbers $\{ \langle \mu_i, \omega\rangle \ | \  i = 1, \ldots , n\}$ from  $\{\orb{\omega} \left( \xi^{\tr{\mu_i}S}\right) \;|\;  i=1,  \ldots , n.\} $.  Since the $\mu_i$ are linearly independent, this allows us to determine $\omega$.  

To determine the rational numbers $\{ \langle \mu_i, \omega\rangle \ | \  i = 1, \ldots , n\}$ from  $\{\cha{\omega} \left( \xi^{\tr{\mu_i}S}\right) \;|\;  i=1,  \ldots , n\} $ we proceed  as follows. Since we know $\delta$ and the $\mu_i$ we can  evaluate  the elements of the set $\{\aorb{\delta}\left( \xi^{\tr{\mu_i}S}\right) \;|\;  i=1,  \ldots , n\}$.  The Weyl Character Formula (Theorem~\ref{weyl}) then allows us to evaluate $\aorb{\delta+\omega}\left( \xi^{\tr{\mu_i}S}\right) = \aorb{\delta}\left( \xi^{\tr{\mu_i}S}\right) \cha{\omega}\left( \xi^{\tr{\mu_i}S}\right)$ for $i = 1, \ldots, n$. Since $\delta+\omega$ is a dominant weight, Lemma~\ref{falafel}.2 allows us to determine $\{ \langle \mu_i, \delta + \omega\rangle \ | \  i = 1, \ldots , n\}$ from  $\{\aorb{\delta + \omega} \left( \xi^{\tr{\mu_i}S}\right) \;|\;  i=1,  \ldots , n\} $. Proceeding as above we can determine  $\omega$. \end{proof}




\begin{example} \label{ex:testweight} Following up on \exmr{ex:roots}.
\begin{description}
  \item[{$\WeylA[1]$:}] 
  We consider the strongly dominant weight $\mu=1$. 
  Then  for $\beta\in \N$ we have
$\orb{\beta}(\xi)= \xi^{\beta} +\xi^{-\beta} = \xi^{\beta}\left( 1+\xi^{-2\beta}\right)$
from which we can deduce how to retrieve $\beta$ for $\xi$ sufficiently large.

  \item[{$\WeylA[2]$:}]  
We can choose $\mu_1 = \tr{[1,1]}$ and 
$\mu_2 = \tr{[1,2]}$ as the elements of our basis of strongly dominant weights.  
To illustrate \thmr{kumquat}  and the proof  of \lemr{falafel}, for 
$\beta=\tr{\begin{bmatrix} \beta_1 & \beta_2\end{bmatrix}}$ :
\begin{eqnarray*}
\orb{\beta}\left(\xi^{\tr[1]{\mu}S}\right) & = & 
     \orb{\mu_1}\left(\xi^{\tr{\beta}S}\right) 
     = {\xi}^{\beta_{{2}}+\beta_{{1}}} 
     \left( 1+{\xi}^{-2\,\beta_{{1}}-\beta_{{2}}}+{\xi}^{-\beta_{{1}}
       -2\,\beta_{{2}}}+{\xi}^{-\beta_{{2}}}+{\xi}^{-\beta_{{1}}}
       +{\xi}^{-2\,\beta_{{1}}-2\,\beta_{{2}}} \right),
 \\
 \orb{\beta}\left(\xi^{\tr[2]{\mu}S}\right) & = & 
     \orb{\mu_2}\left(\xi^{\tr{\beta}S}\right) = 
     {\xi}^{\frac{1}{3}(4\,\beta_{{1}}+5\,\beta_{{2}})} 
     \left( 1+{\xi}^{-\beta_{{1}}}+{\xi}^{-2\,\beta_{{2}}}
             +{\xi}^{-\beta_{{1}}-3\,\beta_{{2}}}
             +{\xi}^{-3\,\beta_{{1}}-2\,\beta_{{2}}}+{\xi}^{-3\,\beta_{{1}}
             -3\,\beta_{{2}}}
 \right) 
.
\end{eqnarray*}

For $\xi=\xi_0^3$ sufficiently large, the integer part of
 $\log_{\xi_0}(\orb{\beta}\left(\xi^{\tr[1]{\mu}S}\right) )$ is $3\beta_1+3\beta_2$ 
 and the integer part of $\log_{\xi_0}(\orb{\beta}\left(\xi^{\tr[2]{\mu}S}\right) )$ 
 is $4\beta_1+5\beta_2$.  From these we can determine $\beta_1$ and $\beta_2$.

\item[{$\WeylB[2]$:}] We choose again 
$\{\mu_1 = \tr{[1,1]}, \mu_2 = \tr{[1,2]}\}$.
\begin{eqnarray*}
\orb{\mu_1}(\xi^{\tr{\beta}S}) & = & 
{\xi}^{\frac{1}{2}\,(3\beta_{{1}}+2\beta_{{2}})} \left( 1+{\xi}^{-\beta_{{1}}}+{
\xi}^{-\frac{1}{2}\,\beta_{{2}}}+{\xi}^{-\beta_{{1}}-\frac{3}{2}\,\beta_{{2}}}
+{\xi}^{-2\,\beta_{{1}}-\frac{1}{2}\,\beta_{{2}}}
+{\xi}^{-3\,\beta_{{1}}-\frac{3}{2}\,\beta_{{2}}}
+{\xi}^{-2\,\beta_{{1}}-2\,\beta_{{2}}}+{\xi}^{-3\,\beta_{{1}}-2\,\beta_{{2}}} \right) 
\\
\orb{\mu_2}(\xi^{\tr{\beta}S}) & = & 
{\xi}^{\frac{1}{2}\,(4\,\beta_{{1}}+{3}\,\beta_{{2}})} \left( 1+{\xi}^{-\beta_{{1}}}+
{\xi}^{-\beta_{{2}}}+{\xi}^{-\beta_{{1}}-2\,\beta_{{2}}}+{\xi}^{-3\,
\beta_{{1}}-\beta_{{2}}}+{\xi}^{-4\,\beta_{{1}}-2\,\beta_{{2}}}+{\xi}^
{-3\,\beta_{{1}}-3\,\beta_{{2}}}+{\xi}^{-4\,\beta_{{1}}-3\,\beta_{{2}}
} \right).
\end{eqnarray*}
\end{description}
\end{example}

\newpage
\section{Sparse multivariate interpolation} \label{sec:sparse}
  \label{interpolation} 

We turn to the problem of sparse multivariate interpolation, that is, 
finding the support (with respect to a given basis) and the coefficients of a
multivariate  polynomial  from its values at chosen points.  
In Section~\ref{pulpo}, we 
consider the case of Laurent polynomials written with respect to the monomial basis.
In Sections~\ref{anxoves} and \ref{bocarones} we consider the interpolation of  a sparse sum of generalized 
 Chebyshev polynomials, of the  first and second kind respectively.
 In Section~\ref{sec:eval}, we discuss an important measure of the complexity of the algorithms: the number of evaluations to be made.

The goal in this section is to recast sparse interpolation into
the problem of finding the suport of a (semi-invariant) linear form
 on the ring of Laurent polynomials.
Evaluation of the function to interpolate, at specific points, gives 
the values of the linear form on certain polynomials.

Multivariate sparse interpolation has been often addressed by reduction 
to the univariate case 
\cite{Arnold14Roche,Ben-Or88,Giesbrecht09,Kaltofen89,Kaltofen00}. 
The essentially univariate sparse interpolation method initiated in \cite{Ben-Or88} 
is known to be reminiscent of Prony's method \cite{Riche95}. 
The function $f$ is evaluated at 
$(p_1^k,\,\ldots,p_n^k)$, for $k=0,1,2,\ldots$, 
where the $p_i$ are chosen as distinct prime numbers \cite{Ben-Or88},
or roots of unity \cite{Arnold14Giesbrecht,Giesbrecht09}.

Our approach builds on a multivariate generalization of Prony's 
interpolation of sums of exponentials \cite{Kunis16,Mourrain18,Sauer16a}.
It is designed to take the group invariance into account.
This latter is  destroyed when reducing to a univariate problem. 
The evaluation points to be used for sparsity in the monomial basis
are 
$(\xi^{\alpha_1},\ldots,\xi^{\alpha_n})$ for a chosen {$\xi\in\Q$, $\xi>1$}, and
for $\alpha$ ranging in an appropriately chosen finite subset of $\N^n$ related to 
the positive orthant of the hypercross  
$$\hcross[n]{r}= 
\left\{\alpha \in \N^n \;\left|\; \prod_{i=1}^r(\alpha_i+1) \leq r\right.\right\}.$$ 
The hypercross and related relevant sets that will appear below 
are illustrated for $n=2$ and $n=3$ in Figure~\ref{Hcross2} and Figure~\ref{Hcross3}. 

\begin{figure}[h]
\hspace*{0.05\textwidth}
\includegraphics[width=0.25\textwidth]{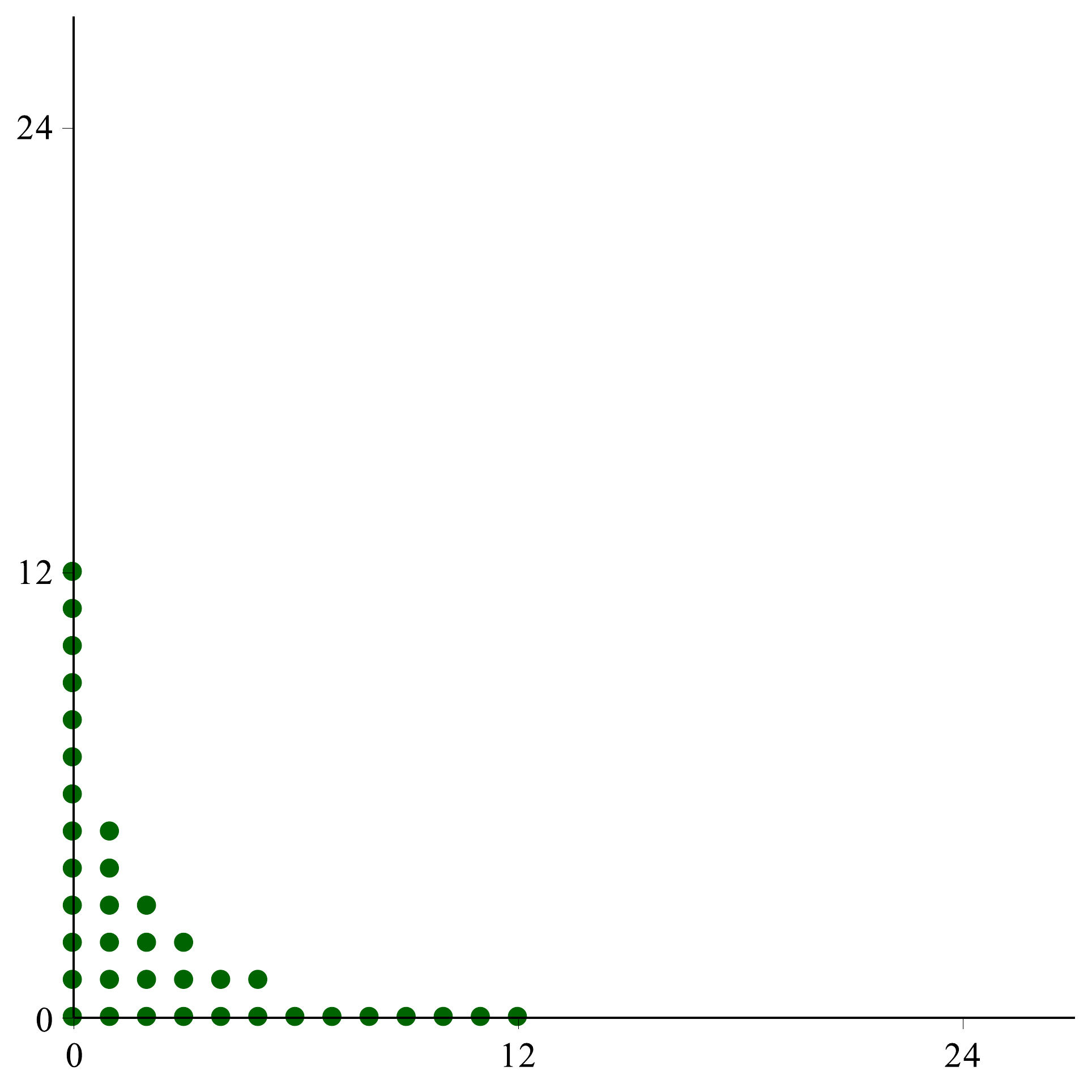} 
\hspace{0.1\textwidth}
\includegraphics[width=0.25\textwidth]{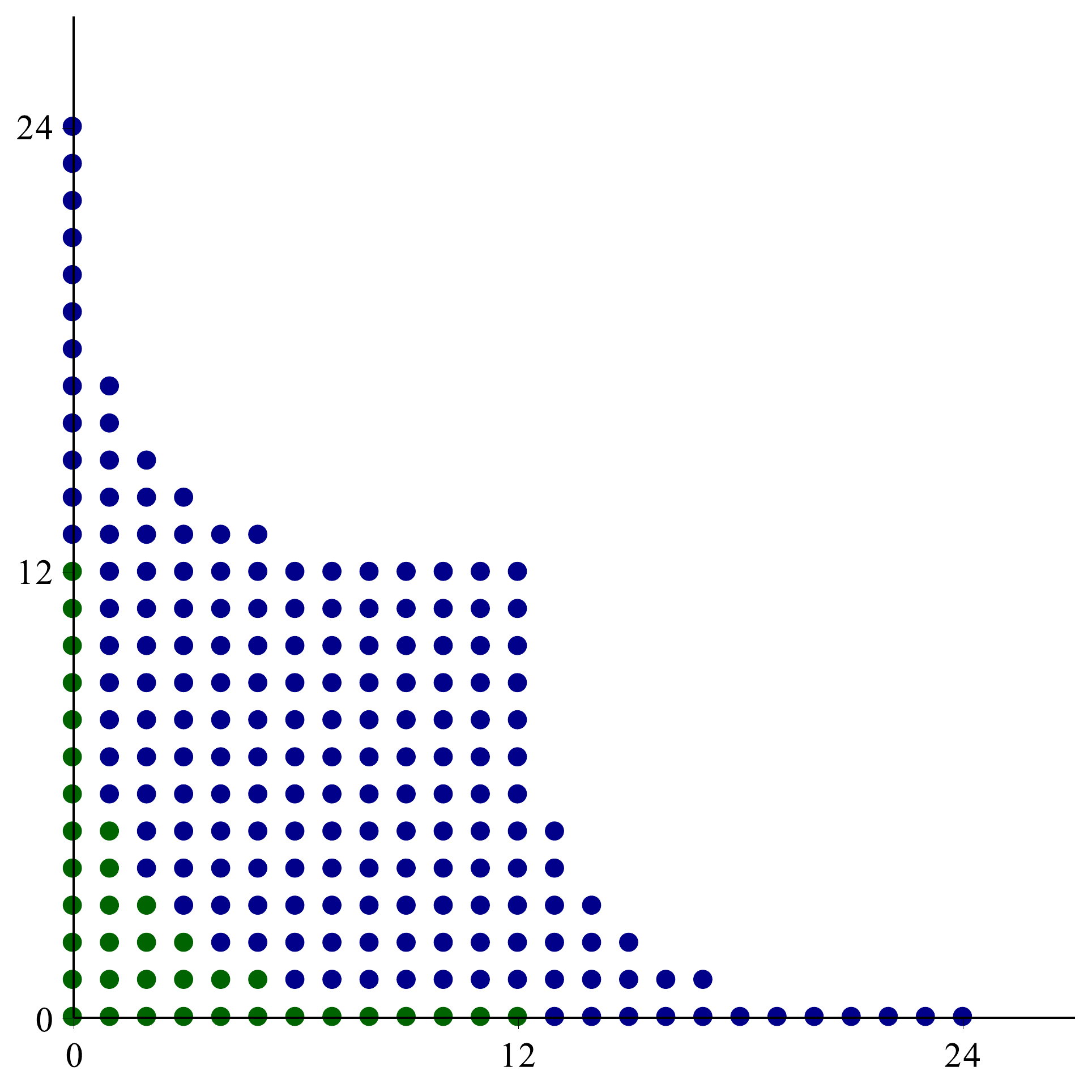}
\hspace{0.1\textwidth}
\includegraphics[width=0.25\textwidth]{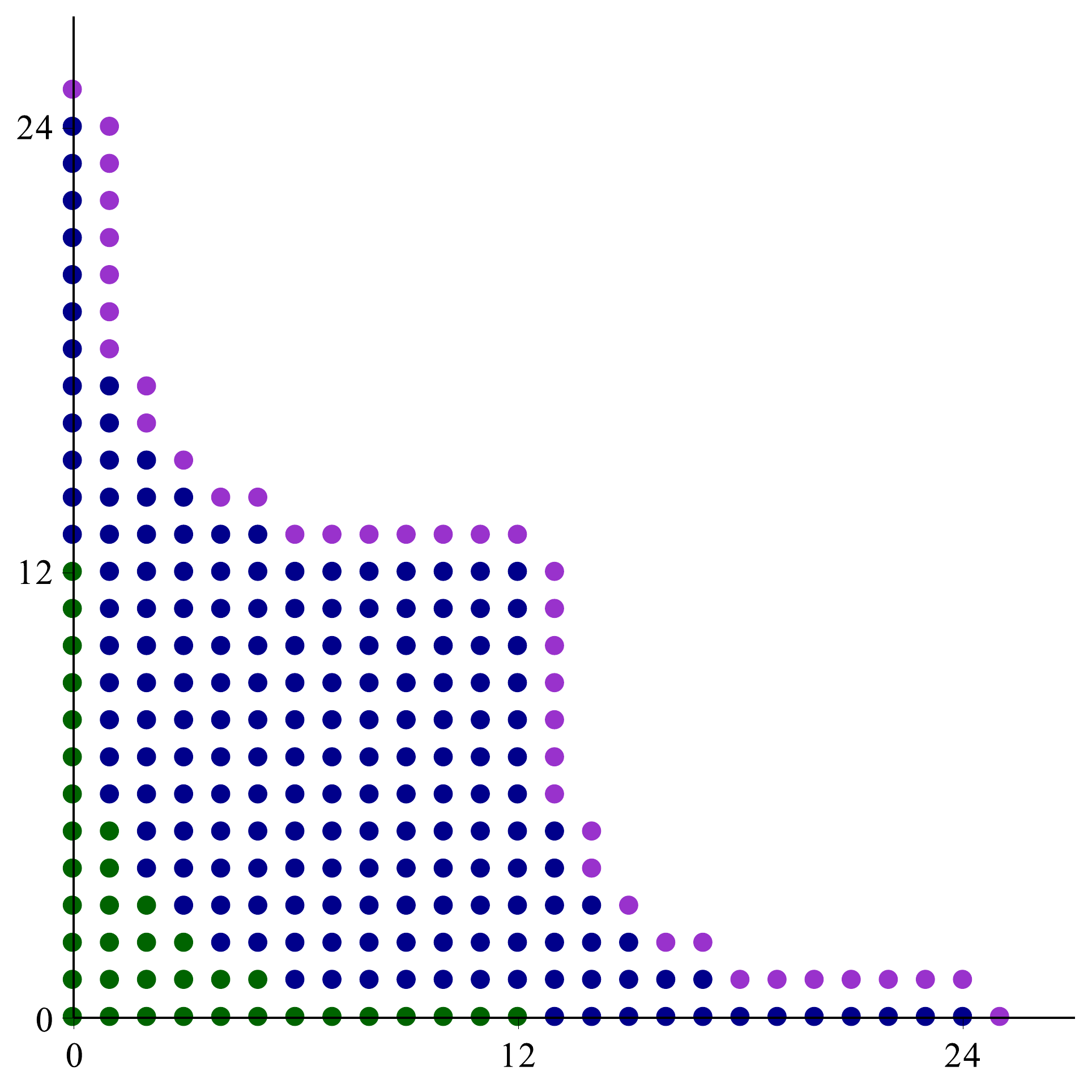}
\caption{$\hcross[2]{13}$, $\hcross[2]{13}+\hcross[2]{13}$ and $\hcross[2]{13}+\hcross[2]{13}+\hcross[2]{2}$.
}\label{Hcross2}
\end{figure}

The evaluation points to be used for sparsity in the generalized Chebyshev basis 
are 
$\left(\orb{\omega_1}(\xi^{\tr{\alpha}S}),\ldots,\orb{\omega_n}(\xi^{\tr{\alpha}S})\right)$ 
for  $\xi=\xi_0^D, \xi_0 \in\N$, $\xi_0>\frac{9}{4}|\gva|^2$ as described in \thmr{kumquat}.
This can be recognized to generalize 
sparse interpolation in  terms of univariate Chebyshev polynomials
 \cite{Arnold15,Giesbrecht04,Lakshman95,Potts14}.

In \secr{hankel} we then show how to  recover the support of a linear form.
More precisely, we provide the 
algorithms to solve the following two problems. Given $r\in\N$:
\begin{enumerate}
  \item Consider the unknowns $\zeta_1,\ldots,\zeta_r \in \K^n$ 
      and $a_1,\ldots, a_r\in \K$.  They define the linear form
         $$\lifo:\begin{array}[t]{ccl} \Kx &  \rightarrow & \K \\ 
             p & \mapsto &\sum_{i=1}^r a_i \,p(\zeta_i)\end{array}$$
         that we write as 
         $\lifo=\sum_{i=1}^{r} a_i \, \eval[\zeta_i],$ where $ \eval[\zeta_i](p) = p(\zeta_i)$.
      From the values of $\lifo$ on 
      $\left\{x^{\alpha+\beta+\gamma} \,|\, \alpha \in \hcross{r},  |\gamma|\leq 1 \right\}$,  
      \algr{support} retrieves the set of pairs 
      $\left\{(a_1,\zeta_1),\ldots,(a_r,\zeta_r)\right\}$.

\item Consider the Weyl group $\gva$ acting on $(\Ks)^n$ as in \Ref{multaction} , 
      the unknowns $\zeta_1,\ldots,\zeta_r \in \K^n$ 
      and $a_1,\ldots, a_r\in \K^*$.  
      They define the $\chi$-invariant linear form
         $$\lifo:\begin{array}[t]{ccl} 
          \Kx &  \rightarrow & \K \\ p & \mapsto & \ds \sum_{i=1}^r a_i \sum_{A\in\gva} \chi(A)\, p(A\star\zeta_i)\end{array}$$
         that we write as $\lifo=\sum_{i=1}^{r} a_i \sum_{A\in\gva}\chi(A)\,\eval[A\star\zeta_i].$
      From the values of $\lifo$ on 
      $\left\{\orb{\alpha}\orb{\beta}\orb{\gamma}\ \,|\, 
         \alpha,\beta \in \hcross{r}, |\gamma|\leq 1  \right\}$ if $\chi(A)=1$,
         or $\left\{\aorb{\delta+\alpha}\orb{\beta}\orb{\gamma}\ \,|\, 
         \alpha,\beta \in \hcross{r}, |\gamma|\leq 1  \right\}$ if $\chi(A)=\det(A)$,
      \algr{invsupport} retrieves the set of pairs 
      $\left\{(\tilde a_1,\vartheta_1),\ldots,(\tilde a_r,\vartheta_r)\right\}$, where 
      \begin{itemize}
        \item $\tilde{a}_i=a_i \,\orb{0}(\zeta_i)=a_i\, |\gva|$
      or $\tilde{a}_i=a_i\,\aorb{\delta}(\zeta_i)\neq 0$ depending whether $\chi=1$ or $\det$;
      \item $\vartheta_i = \left[\orb{\omega_1}(\xi^{\tr[i]{\beta}S}), 
         \;\ldots,\;\orb{\omega_n}(\xi^{\tr[i]{\beta}S})\right]$.
       \end{itemize}


\end{enumerate}
The second  problem appears as a special case of the first one, 
yet the special treatment allows one to reduce the size of the matrices by a factor 
$|\gva|$.
These algorithms rely solely on linear algebra operations and evaluations of polynomial functions:
\begin{itemize}
  \item Determine a nonsingular principal submatrix of size $r$ 
        in a matrix of size $|\hcross{r}|$ ;
  \item Compute the generalized eigenvectors of a pair of  
        matrices of size $r\times r$;
  \item Solve a nonsingular square linear system of size $r$.
\end{itemize}

\begin{figure}[h]
\hspace*{0.05\textwidth}
\includegraphics[width=0.3\textwidth]{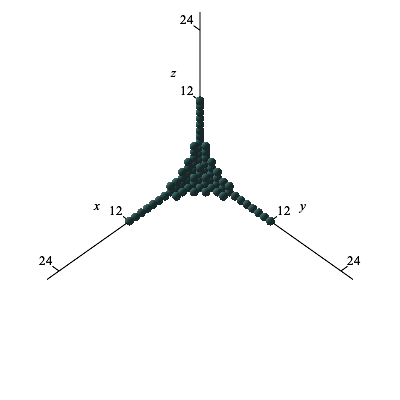}
\includegraphics[width=0.3\textwidth]{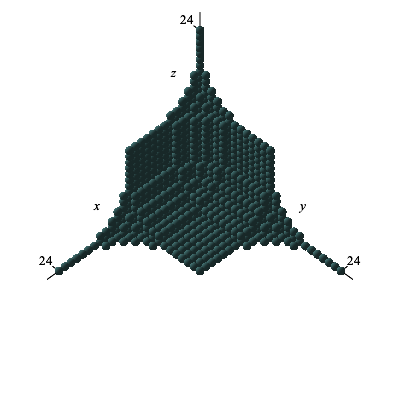}
\includegraphics[width=0.3\textwidth]{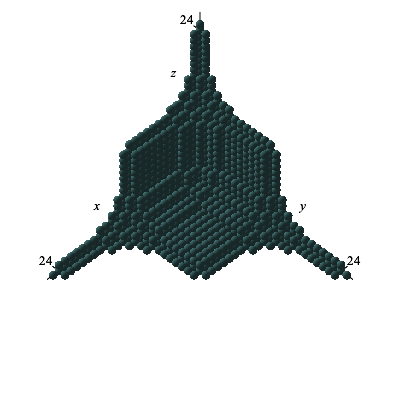}
\caption{$\hcross[3]{13}$, $\hcross[3]{13}+\hcross[3]{13}$ and $\hcross[3]{13}+\hcross[3]{13}+\hcross[3]{2}$.
}\label{Hcross3}
\end{figure}

\subsection{Sparse interpolation of a Laurent polynomial in the monomial basis} 
\label{pulpo}

Consider a Laurent polynomial in $n$ variables 
that is $r$-sparse in the monomial basis. This means that  
\[ \ds f  = \sum_{i=1}^r a_i \,x^{\beta_i} ,\]  
for some $a_i\in \Ks$ and $\beta_i\in\Z^n$.  
The function it defines is a \emph{black box}: 
we can evaluate it at chosen points but know neither 
its coefficients $\left\{a_1,\ldots,a_r\right\} \subset \Ks$
nor its support $\left\{\beta_1,\ldots,\beta_r\right\} \subset\Z^n$; 
only the size $r$ of its support.
The problem we address is to find the pairs 
$(a_i,\beta_i)\in \Ks \times \Z^n$ from 
a small set of evaluations of $f$

To \( \ds f  = \sum_{i=1}^r a_i \,x^{\beta_i} \)  and 
{$\xi\in \Q,\; \xi>1$}, we associate the linear form 
   $$\lifo:\begin{array}[t]{ccl} 
          \Kx &  \rightarrow & \K \\ p & \mapsto & \ds 
          \sum_{i=1}^r a_i\,   p(\zeta_i)
\quad \hbox{ where } \quad 
\zeta_i=\xi^{\tr[i]{\beta}}
   =\begin{bmatrix}\xi^{\beta_{i,1}} &\ldots & \xi^{\beta_{i,n}}\end{bmatrix}\in (\Ks)^n. 
   \end{array}$$
By denoting  $\eval[\zeta]$ the linear form that is the evaluation at $\zeta\in\K^n$ we can write   $\lifo=\sum_{i=1}^{r} a_i \eval[\zeta_i].$
We observe that 
$$\lifo\left(x^\alpha \right)
=\sum_{i=1}^r a_i \left(\xi^{ \tr[i]{\beta} }\right)^{\alpha} 
=\sum_{i=1}^r a_i \left(\xi^{ \tr{\alpha} }\right)^{\beta_i} 
=f\left(\xi^{\tr{\alpha}}\right)$$
since $\tr{\beta}\alpha=\tr{\alpha}\beta$. 
In other words, the value of $\lifo$ on the monomial basis 
$\left\{x^\alpha\,|\, \alpha\in \N^n\right\}$ 
is known from the evaluation of $f$ at the set of points
$\left\{\xi^{\tr{\alpha}}\,|\, \alpha\in \N^n\right\} \subset (\Ks)^n$.
Though trite in the present case, a commutation property 
such as $\tr{\beta}\alpha=\tr{\alpha}\beta$ 
is at the heart of sparse interpolation algorithms.

\begin{algoy} \label{LaurInt} {\Algf LaurentInterpolation}

\In{ $r\in \N_{>0}$, {$\xi\in\Q$, $\xi>1$}, and 
   a function $f$ that can be evaluated at arbitrary points and is known to be a sum of $r$ monomials.
    }

\Out{ The pairs $(a_1,\beta_1),\ldots, (a_r,\beta_r)\in \Ks \times \Z^n$ such that 
\( \ds f  = \sum_{i=1}^r a_i \,x^{\beta_i} \)}
       
    \begin{italg} 
  \item Perform the evaluations of $f$ on 
    $ \left\{\left(\xi^{\tr{(\gamma+\alpha+\beta)}}\right) \;|\; 
            \alpha, \beta \in \hcross{r},\, |\gamma|\leq 1\right\}\subset \Q^n$.

  \item Apply Algorithm~\ref{support}(Support \& Coefficients) 
    to determine the pairs 
     $(a_1,\zeta_1),\ldots,(a_r,\zeta_r)\in\Ks\times(\Ks)^n$
     such that the linear form $\lifo = \sum_{i=1}^r a_i\,\eval[\zeta_i]$
     satisfies 
       \( \lifo( x^{\alpha}) = f(\xi^\alpha) .\)


  \item For $1\leq i\leq r$, determine $\beta_i$ from $\zeta_i$
  by taking logarithms. Indeed $\zeta_i=\xi^{\tr[i]{\beta}}$. Hence 
 for $1\leq i\leq r$ and $1\leq j\leq n$
\[ \zeta_{i,j} = \xi^{\beta_{i,j}} \hbox{ so that } 
    \beta_{i,j}=\frac{\ln (\zeta_{i,j})}{\ln( \xi)} \]

\end{italg}
\end{algoy}

\begin{example} \label{spinterp}
In $\K[x,y,x^{-1},y^{-1}]$, let us consider a 2-sparse polynomial in the monomial basis. 
Thus $f(x,y)=a\, x^{\alpha_1}y^{\alpha_2}+b\,x^{\beta_1}y^{\beta_2}$.
We have 
$$\hcross[2]{2}=\left\{ 
   \tr{\begin{bmatrix}0 & 0\end{bmatrix}},\tr{\begin{bmatrix}1 & 0\end{bmatrix}},
   \tr{\begin{bmatrix}0 & 1\end{bmatrix}}\right\}.$$
Hence 
{ $$ \left\{\alpha+\beta+\gamma \,|\, \alpha,\beta\in \hcross[2]{2}, 
|\gamma|\leq 1 \right\}
=
  \left\{ [0,0],[0,1],[0,2],[0,3],[1,0],[1,1],[1,2],[2,0],[2,1],[3,0]
 \right\}.$$}
To retrieve the pairs $(a,\alpha)$ and $(b,\beta)$ in 
$\Ks\times\N^2$ one thus need to evaluate $f$ at the points
$$\left\{ [1,1],[1,\xi],[1,{\xi}^{2}],[1,{\xi}^{3}],[\xi,1],[\xi,\xi],[
\xi,{\xi}^{2}],[{\xi}^{2},1],[{\xi}^{2},\xi],[{\xi}^{3},1] \right\}
\subset\Q^2 $$
From these values, Algorithm~\ref{support} will recover the pairs
\[ \left(a, [\xi^{\alpha_1}, \xi^{\alpha_2}] \right),
    \left(b, [\xi^{\beta_1}, \xi^{\beta_2}]  \right).
\]
Taking some logarithms on this output we get
\( (a, \alpha), (b, \beta).\)
\end{example}


\newenvironment{proof2}{\textsc{proof  of proposition~\ref{parsnip}:} }{\ensuremath{\blacksquare}}
\def \qed{\hspace*{2mm} \hfill $\Box $\bigskip} 

\subsection{Sparse interpolation with  Chebyshev polynomials of the first kind} \label{anxoves}

We consider now the polynomial ring $\K[X]=\K[X_1,\ldots,X_n]$ and a \emph{black box} function $F$ that is a $r$-sparse polynomial in the basis of Chebyshev polynomials $\left\{T_\beta\right\}_{\beta\in\N^n}$ of the first kind associated to the Weyl group $\gva$: 
$$F(X_1,\ldots,X_n) = \sum_{i=1}^r a_i \,T_{\beta_i}(X_1,\ldots,X_n) 
    \in \K[X_1,\ldots, X_n].$$

By \dfnr{Cheb1Def}, 
$T_{\beta}\left(\orb{\omega_1}(x), \ldots,\orb{\omega_n}(x)\right) =\orb{\beta}(x)$
where $\orb{\beta}(x)=\sum_{A\in\gva} x^{A\beta}$.
Upon introducing
$$f(x)= F(\orb{\omega_1}(x), \ldots,\orb{\omega_n}(x))
=\sum_{i=1}^r a_i \sum_{A\in\gva} x^{A\beta_i} $$
we could apply \algr{LaurInt} to recover the pairs $(a_i,A\beta_i)$.
Instead we examine how to recover the pairs $(a_i,\beta_i)$ only.
For that we associate to $F$ and $\xi\in\N$, $\xi>0$, the linear form
$$\lifo:\begin{array}[t]{ccl} 
          \Kx &  \rightarrow & \K \\ p & \mapsto & \ds 
          \sum_{i=1}^r a_i   \sum_{A\in \gva} p(A\star\zeta_i)
\quad \hbox{ where } \quad 
\zeta_i= \xi^{\tr[i]{\beta} S} \in (\Ks)^n. 
   \end{array}$$


The linear form $\lifo$ is $\gva$-invariant, that is 
$\lifo \left(A\cdot p\right)=\lifo \left( p\right)$. 
The property relevant to sparse interpolation is that the value of $\lifo$
on $\left\{\orb{\alpha}\right\}_{\alpha\in \N^n}$ is obtained by evaluating $F$.

  \begin{proposition} \label{parsnip}
$\ds \lifo(\orb{\alpha})
=
F\left(\orb{\omega_1}(\xi^{\tr{\alpha}S}),\ldots,\orb{\omega_n}(\xi^{\tr{\alpha}S})\right)
$.
 \end{proposition}

The proof of this proposition  is a consequence of 
the following \emph{commutation property}.

\begin{lemma} \label{basic} 
Consider $\chi:\gva\rightarrow \K^*$ a group morphism 
such that $\chi(A)^2 = 1$ for all $A \in \gva$,
and \(\corb{\chi}{\alpha}=\sum_{A\in\gva} \chi(A)^{-1}\, x^{A\alpha}.\)
If $S$ is a positive definite symmetric matrix such 
that $\tr{A} S A = S$ for all $A\in\gva$,
then for any $\xi \in \Ks$
\[\ds \corb{\chi}{\alpha}\left(\xi^{\tr{\beta}S} \right)  
= \corb{\chi}{\beta}  \left(\xi^{\tr{\alpha}S} \right),\]
where
$\ds\corb{\chi}{\alpha} =\sum_{B\in\gva} \chi(B)^{-1}\, x^{B\alpha}$ 
as defined in  \eqref{corb}. 
\end{lemma}

\begin{proof} We have
 $$\corb{\chi}{\alpha}  \left(\xi^{\tr{\beta}S} \right) 
   =\sum_{A\in\gva}\chi(A)^{-1}\,\left(\xi^{\tr{\beta}S}\right)^{A \alpha} 
   =\sum_{A\in\gva}\chi(A)^{-1}\, \xi^{\tr{\beta}S {A \alpha}}  .$$ 
Since $\tr{A}SA = S$, we have $SA=\ti{A} S$ so that 
$$\corb{\chi}{\alpha}  \left(\xi^{\tr{\beta}S} \right)   
    = \sum_{A\in\gva}\chi(A)^{-1}\,\xi^{\tr{\beta}\ti{A} S \alpha} 
    = \sum_{A\in\gva}\chi(A)^{-1}\,\xi^{\tr{(A^{-1}\beta)}S \alpha}.$$
Since, trivially, $\tr{\beta}S\alpha=\tr{\alpha}S\beta$ for all $\alpha,\beta\in\Z^n$, we have 
$$\corb{\chi}{\alpha}  \left(\xi^{\tr{\beta}S} \right) 
    = \sum_{A\in\gva}\chi(A)^{-1}\,\xi^{\tr{\alpha} S {(A^{-1}\beta)}}
    = \sum_{A\in\gva}\chi(A)^{-1}\left(\xi^{\tr{\alpha} S}\right)^{A^{-1}\beta}.$$
The conclusion comes from the fact that $\chi(A)^2 = 1$ implies that $\chi(A)=\chi(A)^{-1}$. 
\end{proof}

\begin{proof2} 
When $\chi(A)=1$ for all $A\in \gva$ we have 
$\corb{1}{\alpha}=\orb{\alpha}$. Therefore \lemr{basic} implies
\begin{eqnarray*}
F\left(\orb{\omega_1}(\xi^{\tr{\alpha}S}),\ldots,\orb{\omega_n}(\xi^{\tr{\alpha}S})\right)
& = &
\sum_{i=1}^r a_i \,T_{\beta_i}\left(\orb{\omega_1}(\xi^{\tr{\alpha}S}),\ldots,\orb{\omega_n}(\xi^{\tr{\alpha}S})\right)
\\ & = & 
\sum_{i=1}^r a_i \,\orb{\beta_i}\left(\xi^{\tr{\alpha}S}\right)
 \; = \;  
\sum_{i=1}^r a_i \,\orb{\alpha}\left(\xi^{\tr[i]{\beta}S}\right)
\; = \;   
\lifo(\orb{\alpha}).
\end{eqnarray*}
\end{proof2}

In the following algorithm to recover the support of $F$ we need to have the value of 
$\lifo$ on the polynomials $\orb{\alpha}\orb{\beta}\orb{\gamma}$ for $\alpha,\beta\in\hcross{r}$ and $|\gamma|\leq 1$.
We have access to the values of $\lifo$ on $\orb{\mu}$, for any $\mu\in\N^n$, by evaluating 
$F$ at $\left( \orb{\omega_1}(\xi^{\tr{\mu}S}),\ldots,\orb{\omega_n}(\xi^{\tr{\mu}S})\right)$.
To get the values  of $\lifo$ on $\orb{\alpha}\orb{\beta}\orb{\gamma}$ we consider the relationships stemming from \prpr{raspberry}
 \[ 
  \orb{\gamma}\orb{\alpha}\orb{\beta} = 
       \sum_{\nu \in S(\alpha,\beta,\gamma)} a_\nu\,\orb{\nu} \]
where $S(\alpha,\beta,\gamma)$ is a  finite subset 
of $\left\{\mu \in \N^n\,|\,\mu\prec \alpha+\beta+\gamma\right\}$.
Then the set 
\begin{equation} \label{wcross}
 \wcross{r} = \bigcup_{\substack{\alpha,\beta\in \hcross{r}\\|\gamma|\leq  1}} S(\alpha,\beta,\gamma) \end{equation}
 indexes the 
evaluations  needed to determine the support of a $r$-sparse sum of 
Chebyshev polynomials associated to the Weyl group $\gva$.

%
%
%

{As we noted in the paragraph preceding Lemma~\ref{falafel}, the entries of $S$ are  in $\Q$ and we shall denote by $D$ the least common denominator of these entries.}

 \begin{algoy} \label{FKInterp} 
{\Algf FirstKindInterpolation}

\In{ $r \in \N_{>0}$, $\xi_0 \in \N_{>0}$,  where 
     $\xi_0 > \left(\frac{3}{2}|\WeylG|\right)^2$  and  $\xi= \xi_0^D$, and a function $F$ that can be evaluated at arbitrary points and is known to be the sum of $r$ generalized Chebyshev polynomials of the first kind.
     }

\Out{The pairs $(a_1,\beta_1), \ldots , (a_r, \beta_r) \in \Ks \times \Z^n$  such that
$$F(X_1,\ldots,X_n) = \sum_{i=1}^r a_i \,T_{\beta_i}(X_1,\ldots,X_n).$$}

\begin{italg} 
  \item From the evaluations $\left\{\left. F\left( 
  \orb{\omega_1}(\xi^{\tr{\alpha}S}),\ldots,\orb{\omega_n}(\xi^{\tr{\alpha}S})\right) 
        \;\right|\; 
  \alpha \in \wcross{r} \right\}$
        determine 
  $\left\{ \left. \lifo(\orb{\alpha}\orb{\beta}\orb{\gamma}) \;\right|\; 
          \alpha,\beta \in \hcross{r},\, |\gamma|\leq 1 \right\}$ 
          
          {\small \emph{\hspace*{\stretch{3}}  \% 
          The hypothesis on $\xi= \xi_0^D$ guarantees that $\xi^{\tr{\alpha}S}$ is a row vector of integers. }}

   \item Apply \algr{invsupport} (Invariant Support \& Coefficients) 
    to calculate the vectors \[ \vartheta_i = \left[\orb{\omega_1}(\xi^{\tr[i]{\beta}S}), 
         \;\ldots,\;\orb{\omega_n}(\xi^{\tr[i]{\beta}S})\right], \quad 1\leq i\leq r\]
       and the vector
         \(\left[\tilde{a}_1,\ldots, \tilde{a}_r\right]= 
         \left[|\gva|\,{a}_1,\ldots, |\gva|\,{a}_r\right].
         \) 
         
\item Deduce \(\left[{a}_1,\ldots, {a}_r\right].\)

   \item  Calculate    $$\left[\orb{\beta_i}\left(\xi^{\tr[1]{\mu}S}\right),
         \;\ldots,\;\orb{\beta_i}\left(\xi^{\tr[n]{\mu}S}\right)\right]
 =\left[\Tche{\mu_1}(\vartheta_i),\;\ldots,\;\Tche{\mu_n}(\vartheta_i)\right]
 $$
using the precomputed Chebyshev polynomials $\{T_{\mu_1}, \ldots , T_{\mu_n}\}$, where $\mu_1,\ldots,\mu_n$ are linearly independent strongly dominant weights. 
  
  \item Using Theorem~\ref{kumquat}, recover each $\beta_i$ from 
  $$\left[\orb{\beta_i}\left(\xi^{\tr[1]{\mu}S}\right),
         \;\ldots,\;\orb{\beta_i}\left(\xi^{\tr[n]{\mu}S}\right)\right].$$

\end{italg}
\end{algoy}

As will be remarked after its description, \algr{invsupport} may,
in some favorable cases, return directly the vector  
$$\left[\orb{\mu_j}(\xi^{\tr[1]{\beta}S}),\ldots, \orb{\mu_j}(\xi^{\tr[r]{\beta}S})\right]
=
\left[\orb{\beta_1}(\xi^{\tr[j]{\mu}S}),\ldots, \orb{\beta_r}(\xi^{\tr[j]{\mu}S})\right].
$$
for some or all  $1\leq j\leq n$. 
This then saves on evaluating $T_{\mu_j}$ at the points 
$\vartheta_1,\, \ldots,\, \vartheta_r$.

\begin{example}\label{spinterp1} 
We consider the Chebyshev polynomials  of the first kind 
$\left\{T_\alpha\right\}_{\alpha\in \N^2}$ associated to 
the Weyl group $\WeylA[2]$ and 
a 2-sparse polynomial  $F(X,Y)=a\, T_{\alpha}(X,Y)+b\,T_\beta(X,Y)$ in this basis of $\K[X,Y]$.

We need to consider  
$$\hcross[2]{2}=\left\{ 
   \tr{\begin{bmatrix}0 & 0\end{bmatrix}},\tr{\begin{bmatrix}1 & 0\end{bmatrix}},
   \tr{\begin{bmatrix}0 & 1\end{bmatrix}}\right\}.$$ 
The following relations
\[\begin{array}{l}
{\orb{{0,0}}}^{2}=6\,\orb{{0,0}},\quad
\orb{{0,0}}\orb{{0,1}}=6\,\orb{{0,1}},\quad
\orb{{0,0}}\orb{{1,0}}=6\,\orb{{1,0}},\\
{\orb{{0,1}}}^{2}=2\,\orb{{0,2}}+4\,\orb{{1,0}},\quad
\orb{{0,1}}\orb{{1,0}}=4\,\orb{{1,1}}+2\,\orb{{0,0}},\quad
{\orb{{1,0}}}^{2}=2\,\orb{{2,0}}+4\,\orb{{0,1}}
\end{array}
\] 
and
\[\begin{array}{l}
\orb{{2,0}}\orb{{0,0}}=6\,\orb{{2,0}},\quad 
\orb{{2,0}}\orb{{1,0}}=2\,\orb{{3,0}}+4\,\orb{{1,1}}, \quad
\orb{{2,0}}\orb{{0,1}}=4\,\orb{{2,1}}+2\,\orb{{1,0}},\\
\orb{{1,1}}\orb{{0,0}}=6\,\orb{{1,1}},\quad 
\orb{{1,1}}\orb{{1,0}}=2\,\orb{2, 1} + 2\,\orb{0, 2} + 2\,\orb{1, 0}, \quad
\orb{{1,1}}\orb{{0,1}}=2\,\orb{1, 2} + 2\,\orb{2, 0} + 2\,\orb{0, 1},\\
\orb{{0,2}}\orb{{0,0}}=6\,\orb{{0,2}},\quad
\orb{{0,2}}\orb{{1,0}}=4\,\orb{{1,2}}+2\,\orb{{0,1}},\quad
\orb{{0,2}}\orb{{0,1}}=2\,\orb{{0,3}}+4\,\orb{{1,1}},
\end{array}
\]
allow one to express any product 
$\orb{\alpha}\orb{\beta} \orb{\gamma}$, $\alpha,\beta \in \hcross{2}$, $|\gamma|\leq 1$ 
as a linear combination of elements from 
$\{\orb{\alpha} \ | \ \alpha \in \wcross[{\WeylA[2]}]{2}\}$ where 
\[\wcross[{\WeylA[2]}]{2}= \left\{ [0,0],[0,1],[0,2],[0,3],[1,0],[1,1],[1,2],[2,0],[2,1],[3,0]
 \right\} .\]
 For example $\orb{1,0}\orb{0,1}^2 = 8\orb{1,2}+20 \orb{0,1} + 8 \orb{2,0}$.  

We consider 
\begin{eqnarray*}
f(x,y) & = & F\left(\orb{\omega_1}(x,y),\orb{\omega_2}(x,y)\right) 
\end{eqnarray*}
where 
\[\orb{\omega_1}(x,y)= 2\,x+2\,{y}{x}^{-1}+2\,{y}^{-1}, \hbox{ and }
\orb{\omega_2}(x,y)= 2\,y+2\,{x}{y}^{-1}+2\,{x}^{-1}.\]

We introduce the invariant linear form $\lifo$ on $\K[x,y,x^{-1},y^{-1}]$ 
determined by 
$\lifo(\orb{\gamma})=
f\left(\xi^{\frac{2}{3}\gamma_1+\frac{1}{3} \gamma_2},
       \xi^{\frac{1}{3}\gamma_1+\frac{2}{3} \gamma_2} \right)$
The first step of the algorithm requires us to determine 
  $\left\{ \left. \lifo(\orb{\alpha}\orb{\beta}\orb{\gamma}) \;\right|\; 
          \alpha,\beta \in \hcross[2]{2},\, |\gamma|\leq 1 \right\}$.  
Expanding these triple products as linear combinations of orbit polynomials, we see from Proposition~\ref{parsnip} that to determine these values it is enough to evaluate $f(x,y)$ at 
the 10 points $\{\xi^{\tr{\alpha}S} \ | \ \alpha \in \wcross[{\WeylA[2]}]{2}\}$, that is, at the points
\[ \left\{  [ 1,1 ] , [\xi^\frac{1}{3}, \xi^\frac{2}{3}], [\xi^\frac{2}{3}, \xi^\frac{4}{3}], [\xi, \xi^2], [\xi^\frac{2}{3}, \xi^\frac{1}{3}], [\xi, \xi], [\xi^\frac{4}{3}, \xi^\frac{5}{3}], [\xi^\frac{4}{3}, \xi^\frac{2}{3}], [\xi^\frac{5}{3}, \xi^\frac{4}{3}], [\xi^2, \xi]  \right\}\]

Note that $D = 3$ so  $\xi = ({\xi}_0)^3$ for some $\xi_0 \in \N_{>0}$. Therefore the above vectors have integer entries.


%

From these values, Algorithm~\ref{invsupport} will recover the pairs
$(a, \vartheta_\alpha)$ and $(b, \vartheta_\beta)$
where 
\[\vartheta_\alpha = [\orb{\omega_1}(\xi^{\tr{\alpha}S}), \orb{\omega_2}(\xi^{\tr{\alpha}S})] \; 
\mbox{ and } \; \vartheta_\beta =  [\orb{\omega_1}(\xi^{\tr{\beta}S}), \orb{\omega_2}(\xi^{\tr{\beta}S})].\]
One can  then form 
$$[T_{\mu_1}(\vartheta_\alpha), T_{\mu_2}(\vartheta_\alpha)]
= \left[\orb{\alpha}(\xi^{\tr[1]{\mu}S}), \orb{\alpha}(\xi^{\tr[2]{\mu}S})\right]
$$ and 
$$[T_{\mu_1}(\vartheta_\beta), T_{\mu_2}(\vartheta_\beta)]
= \left[\orb{\beta}(\xi^{\tr[1]{\mu}S}), \orb{\beta}(\xi^{\tr[2]{\mu}S})\right]
$$ 
using the polynomials calculated in \exmr{ex:chebi1}  and find $\alpha$ and $\beta$ as illustrated in 
\exmr{ex:testweight}.


Note that the function $f$ is a 12-sparse polynomial in the monomial basis
\begin{eqnarray*}
f(x,y) & = &   
a\, \orb{\alpha}(x)+b\, \orb{\beta}(x)\\
& = &  a \left( {x}^{\alpha_{{1}}}{y}^{\alpha_{{2}}}+{x}^{-\alpha_{{1}}}{y}^{
\alpha_{{1}}+\alpha_{{2}}}+{x}^{\alpha_{{1}}+\alpha_{{2}}}{y}^{-\alpha
_{{2}}}+{x}^{\alpha_{{2}}}{y}^{-\alpha_{{1}}-\alpha_{{2}}}+{x}^{-
\alpha_{{1}}-\alpha_{{2}}}{y}^{\alpha_{{1}}}+{x}^{-\alpha_{{2}}}{y}^{-
\alpha_{{1}}} \right) 
\\ & & +b \left( {x}^{\beta_{{1}}}{y}^{\beta_{{2}}}+{x}
^{-\beta_{{1}}}{y}^{\beta_{{1}}+\beta_{{2}}}+{x}^{\beta_{{1}}+\beta_{{
2}}}{y}^{-\beta_{{2}}}+{x}^{\beta_{{2}}}{y}^{-\beta_{{1}}-\beta_{{2}}}
+{x}^{-\beta_{{1}}-\beta_{{2}}}{y}^{\beta_{{1}}}+{x}^{-\beta_{{2}}}{y}
^{-\beta_{{1}}} \right),
\end{eqnarray*}
Yet to retrieve its support we only need to evaluate $f$ at
 points indexed by $\wcross[{\WeylA[2]}]{2}$, which is  equal to 
 $\hcross[2]{2}+\hcross[2]{2}+\hcross[2]{2}$ and has cardinality $10$.

Note though that $12$ is actually an upper bound on the sparsity of $f$ in the monomial basis.
 If $\alpha$ or $\beta$ has a component that is zero then the actual sparsity 
 can be $4$, $6$, $7$ or $9$. We shall comment on dealing with upper bounds on the sparsity rather than the exact sparsity in \secr{final}.
\end{example}

\subsection{Sparse interpolation with  Chebyshev polynomials of the second kind} 
\label{bocarones}

We consider now the polynomial ring $\K[X]=\K[X_1,\ldots,X_n]$ 
and a \emph{black box} function $F$ that is an $r$-sparse polynomial 
in the basis of Chebyshev polynomials 
$\left\{U_\beta\right\}_{\beta\in\N^n}$ of the second kind 
associated to the Weyl group $\gva$. Hence
$$F(X_1,\ldots,X_n) = \sum_{i=1}^r a_i\, U_{\beta_i}(X_1,\ldots,X_n) 
    \in \K[X_1,\ldots, X_n].$$

By \dfnr{Cheb2Def}  and thanks to \thmr{weyl}
$U_{\beta}\left(\orb{\omega_1}(x), \ldots,\orb{\omega_n}(x)\right) 
=\cha{\beta}(x) =\frac{\aorb{\delta+\beta}(x)}{\aorb{\delta}(x)}$.
Hence upon introducing
$$f(x)= \aorb{\delta}(x)\, F(\orb{\omega_1}(x), \ldots,\orb{\omega_n}(x))
= \sum_{i=1}^r a_i \, \aorb{\delta+\beta_i}(x)
=\sum_{i=1}^r a_i \, \sum_{A\in\gva} \det(A)^{-1}\, x^{A(\delta+\beta_i)} $$
we could apply \algr{LaurInt} to recover the pairs $\left(a_i,A(\delta+\beta_i)\right)$.
We examine how to recover only the pairs $(a_i,\delta+\beta_i)$.
For that we define
$$\lifo:\begin{array}[t]{ccl} 
          \Kx &  \rightarrow & \K \\ p & \mapsto & \ds 
          \sum_{i=1}^r a_i   \sum_{A\in \gva} \det(A) \, p(\zeta_i^A)
\quad \hbox{ where } \quad 
\zeta_i= \xi^{\tr{(\delta+\beta_i)} S} \in (\Ks)^n. 
   \end{array}$$
The linear form $\lifo$ is skew invariant, \textit{i.e.}
$\lifo \left(A\cdot p\right)=\det(A)^{-1} \,\lifo \left( p\right)$. 
The property relevant to sparse interpolation is that the value of $\lifo$
on $\left\{\aorb{\alpha}\right\}_{\alpha\in \N^n}$ is obtained by evaluating $F$.

\begin{proposition} \label{turnip}
$\ds \lifo(\aorb{\alpha})
=
\aorb{\delta}\left(\xi^{\tr{\alpha}S}\right)\,F\left(\orb{\omega_1}(\xi^{\tr{\alpha}S}),\ldots,\orb{\omega_n}(\xi^{\tr{\alpha}S})\right)
 $
 \end{proposition}


\begin{proof} 
Note that
$\ds F(\orb{\omega_1}(x), \ldots, \orb{\omega_n}(x))
= \sum_{i=1}^r {a}_{i} \, \cha{\beta_i}(x) $ so that
$\ds \aorb{\delta}(x) \, F(\orb{\omega_1}(x), \ldots, \orb{\omega_n}(x))
= \sum_{i=1}^r {a}_{i} \, \aorb{\delta+\beta_i}(x)
$.

 \lemr{basic} implies
\begin{eqnarray*}
\aorb{\delta}(\xi^{\tr{\alpha}S})\,
F\left(\orb{\omega_1}(\xi^{\tr{\alpha}S}),\ldots,\orb{\omega_n}(\xi^{\tr{\alpha}S})\right)
& = & 
\sum_{i=1}^r a_i \,\aorb{\delta+\beta_i}\left(\xi^{\tr{\alpha}S}\right) 
\\ & = & 
\sum_{i=1}^r a_i \,\aorb{\alpha}\left(\xi^{\tr{(\delta+\beta_i)}S}\right)
\; = \;   
\lifo(\aorb{\alpha}).
\end{eqnarray*}
\end{proof}

We are now in a position to describe the algorithm 
to recover the support of $F$ from its evaluations at a set of points 
$\left\{\left( \orb{\omega_1}(\xi^{\tr{\alpha}S}),\ldots,\orb{\omega_n}(\xi^{\tr{\alpha}S})\right)\;|\;{\alpha\in\wcroxx{r}}\right\}$.  The set $\wcroxx{r}$ is defined similarly to the set $\wcross{r}$  in the previous section 
(Equation~\Ref{wcross}).   
\begin{equation} \label{wcroxx}
 \wcroxx{r} = \bigcup_{\substack{\alpha,\beta\in \hcross{r}\\|\gamma|\leq  1}} \check{S}(\alpha,\beta,\gamma) \end{equation}
 where the subsets $\check{S}(\alpha,\beta,\gamma)$ of 
 $\left\{\mu \in \delta+\N^n\,|\,\mu\prec \delta+\alpha+\beta+\gamma\right\}$
 are defined by the fact that
 \[ 
  \aorb{\delta+\alpha}\orb{\beta}\orb{\gamma} = 
       \sum_{\nu \in \check{S}(\alpha,\beta,\gamma)} a_\nu\,\aorb{\nu}.\]


 \begin{algoy} \label{SKInterp} 
{\Algf SecondKindInterpolation}

\In{ $r \in \N_{>0}$, $\xi_0 \in \N_{>0}$,  where 
   $\xi_0> (\frac{3}{2}|\WeylG|)^2$  and  $\xi= \xi_0^D$ , and a function $F$ that can be evaluated at arbitrary points and is known to be the sum of $r$ generalized Chebyshev polynomials of the second kind.
    }

\Out{The pairs $(a_1,\beta_1), \ldots , (a_r, \beta_r) \in \Ks \times \Z^n$  such that
$$F(X_1,\ldots,X_n) = \sum_{i=1}^r a_i\, U_{\beta_i}(X_1,\ldots,X_n).$$}

\begin{italg} 
  \item From $\left\{\left. \aorb{\delta}(\xi^{\tr{\alpha}S}) F\left( 
  \orb{\omega_1}(\xi^{\tr{\alpha}S}),\ldots,\orb{\omega_n}(\xi^{\tr{\alpha}S})\right) 
        \;\right|\; 
  \alpha \in \wcroxx{r} \right\}$
        determine 
  $\left\{ \left. \lifo(\aorb{\delta+\alpha}\orb{\beta}\orb{\gamma}) \;\right|\; 
          \alpha,\beta \in \hcross{r},\, |\gamma|\leq 1 \right\}$ 

   \item Apply \algr{invsupport} (Invariant Support \& Coefficients) 
    to calculate the vectors 
    \[ \check{\vartheta}_i = \left[\orb{\omega_1}(\xi^{\tr{(\delta+\beta_i)}S}), 
         \;\ldots,\;\orb{\omega_n}(\xi^{\tr{(\delta+\beta_i)}S})\right], \quad 1\leq i\leq r\]
        and the vector
         \(\left[\tilde{a}_1,\ldots, \tilde{a}_r\right]= 
         \left[\aorb{\delta}\left(\xi^{\tr{(\delta+\beta_1)} S}\right)\,{a}_1,
         \ldots, 
         \aorb{\delta}\left(\xi^{\tr{(\delta+\beta_r)} S}\right)\,{a}_r\right]
         \)

   \item  Calculate   $$\left[\orb{\delta+\beta_i}\left(\xi^{\tr[1]{\mu}S}\right),
         \;\ldots,\;\orb{\delta+\beta_i}\left(\xi^{\tr[n]{\mu}S}\right)\right]
 =\left[\Tche{\mu_1}(\check{\vartheta}_i),\;\ldots,\;\Tche{\mu_n}(\check{\vartheta}_i)\right]
 $$
using the  Chebyshev polynomials $\{T_{\mu_1}, \ldots , T_{\mu_n}\}$, where $\mu_1,\ldots,\mu_n$ are linearly independent strongly dominant weights. 
  
  \item Using Theorem~\ref{kumquat}, recover each $\delta+\beta_i$, and hence $\beta_i$, 
  from $$\left[\orb{\delta+\beta_i}\left(\xi^{\tr[1]{\mu}S}\right),
         \;\ldots,\;\orb{\delta+\beta_i}\left(\xi^{\tr[n]{\mu}S}\right)\right].$$

  \item  Compute          \(\
         \left[\aorb{\delta}\left(\xi^{\tr{(\delta+\beta_1)}S}\right),
         \ldots, 
         \aorb{\delta}\left(\xi^{\tr{(\delta+\beta_r)} S}\right)\right]
         \) and deduce \(\left[{a}_1,\ldots, {a}_r\right].\)
        Our hypothesis for $\xi_0$ imply that $\xi_0 > |\WeylG|$ so Corollary~\ref{falafelcor} implies that none of the components are zero.

  
\end{italg}
\end{algoy}

\begin{example}\label{spinterp2} 
We consider the Chebyshev polynomials  of the second kind 
$\left\{U_\gamma\right\}_{\gamma\in \N^2}$ associated to 
the Weyl group $\WeylA[2]$ and 
a 2-sparse polynomial  $F(X,Y)=a\, U_{\alpha}(X,Y)+b\,U_\beta(X,Y)$ in this basis of $\K[X,Y]$.

We need to consider  
$$\hcross[2]{2}=\left\{ 
   \tr{\begin{bmatrix}0 & 0\end{bmatrix}},\tr{\begin{bmatrix}1 & 0\end{bmatrix}},
   \tr{\begin{bmatrix}0 & 1\end{bmatrix}}\right\}.$$ 
The following relations
\[\begin{array}{l}
\aorb{{1,1}}\orb{{0,0}}=6\,\aorb{{1,1}}, \;
\aorb{{1,1}}\orb{{1,0}}=2\,\aorb{{2,1}}, \; 
\aorb{{1,1}}\orb{{0,1}}=2\,\aorb{{1,2}}, \; 
\\
\aorb{{2,1}}\orb{{0,0}}=6\,\aorb{{2,1}}, \; 
\aorb{{2,1}}\orb{{1,0}}=2\,\aorb{{3,1}}+2\,\aorb{{1,2}},\;
\aorb{{2,1}}\orb{{0,1}}=2\,\aorb{{2,2}}+2\,\aorb{{1,1}},
\\
\aorb{{1,2}}\orb{{0,0}}=6\,\aorb{{1,2}},\;
\aorb{{1,2}}\orb{{1,0}}=2\,\aorb{{2,2}}+2\,\aorb{{1,1}},\;
\aorb{{1,2}}\orb{{0,1}}=2\,\aorb{{1,3}}+2\,\aorb{{2,1}}
\end{array}
\] 
and
\[\begin{array}{l}
\aorb{{3,1}}\orb{{0,0}}=6\,\aorb{{3,1}},\quad
\aorb{{3,1}}\orb{{1,0}}=2\,\aorb{{4,1}}+2\,\aorb{{2,2}},\quad
\aorb{{3,1}}\orb{{0,1}}=2\,\aorb{{3,2}}+2\,\aorb{{2,1}}
\\
\aorb{{2,2}}\orb{{0,0}}=6\,\aorb{{2,2}},\quad
\aorb{{2,2}}\orb{{1,0}}=2\,\aorb{{3,2}}+2\,\aorb{{1,3}}+2\,\aorb{{2,1}},\quad
\aorb{{2,2}}\orb{{0,1}}=2\,\aorb{{2,3}}+2\,\aorb{{3,1}}+2\,\aorb{{1,2}}
\\
\aorb{{1,3}}\orb{{0,0}}=6\,\aorb{{1,3}},\quad
\aorb{{1,3}}\orb{{1,0}}=2\,\aorb{{2,3}}+2\,\aorb{{1,2}},\quad
\aorb{{1,3}}\orb{{0,1}}=2\,\aorb{{1,4}}+2\,\aorb{{2,2}}
\end{array}
\]
allow one to express any product 
$\aorb{\delta+\alpha}\orb{\beta} \orb{\gamma}$, $\alpha,\beta \in \hcross{2}$, $|\gamma|\leq 1$ 
as a linear combination of elements from 
$\{\aorb{\alpha} \ | \ \alpha \in \wcroxx[{\WeylA[2]}]{2}\}$ where 
\[\wcroxx[{\WeylA[2]}]{2}= \left\{ [1,1],[2,1],[1,2],[3,1],[2,2],[1,3],[4,1],[3,2],[2,3],[1,4]
 \right\} .\]

We consider 
\begin{eqnarray*}
f(x,y) & = & \aorb{\delta}(x,y)\, F\left(\orb{\omega_1}(x,y),\orb{\omega_2}(x,y)\right) 
\end{eqnarray*}
where 
\[\orb{\omega_1}(x,y)= 2\,x+2\,{y}{x}^{-1}+2\,{y}^{-1}, \quad
\orb{\omega_2}(x,y)= 2\,y+2\,{x}{y}^{-1}+2\,{x}^{-1}, \hbox{ and }
\aorb{\delta}(x,y)= xy-x^{-1}{y}^{2}-{x}^{2}y^{-1}+{x}{y}^{-2}+{y}{x}^{-2}-x^{-1}y^{-1}
.\]

We introduce the $\chi$-invariant linear form $\lifo$ on $\K[x,y,x^{-1},y^{-1}]$ 
determined by 
$\lifo(\aorb{\gamma})= 
f\left(\xi^{\frac{2}{3}\gamma_1+\frac{1}{3} \gamma_2},
       \xi^{\frac{1}{3}\gamma_1+\frac{2}{3} \gamma_2} \right)$
The first step of the algorithm requires us to determine 
  $\left\{ \left. \lifo(\aorb{\delta+\alpha}\orb{\beta}\orb{\gamma}) \;\right|\; 
          \alpha,\beta \in \hcross[2]{2},\, |\gamma|\leq 1 \right\}$.  
Expanding these products as linear combinations of skew orbit polynomials, 
we see  
that 
it is enough to evaluate $f$ at 
the 10 points $\{\xi^{\tr{(\delta+\alpha)}S} \ | \ \alpha \in \wcroxx[{\WeylA[2]}]{2}\}$, that is, at the points
\[ \left\{  [ \xi,\xi ] , [\xi^\frac{4}{3}, \xi^\frac{5}{3}], [\xi^\frac{5}{3}, \xi^\frac{7}{3}], [\xi^2, \xi^3], [\xi^\frac{5}{3}, \xi^\frac{4}{3}], [\xi^2, \xi^2], [\xi^\frac{47}{3}, \xi^\frac{8}{3}], [\xi^\frac{7}{3}, \xi^\frac{5}{3}], [\xi^\frac{8}{3}, \xi^\frac{7}{3}], [\xi^3, \xi^2]  \right\}\]

Note that $D = 3$ so  $\xi = (\xi_0)^3$ for some $\xi_0 \in \N$ and therefore the above vectors have integer entries.


%

From these values, Algorithm~\ref{invsupport} will recover the pairs
$(\aorb{\delta}(\xi^{\tr{(\delta+\alpha)}S})\, a, \check{\vartheta}_\alpha)$ and 
$(\aorb{\delta}(\xi^{\tr{(\delta+\beta)}S})\, b, \check{\vartheta}_\beta)$
where 
\[
\check{\vartheta}_\alpha = [\orb{\omega_1}(\xi^{\tr{(\delta+\alpha)}S}), \orb{\omega_2}(\xi^{\tr{(\delta+\alpha)}S}) 
\; \mbox{ and } \; 
\check{\vartheta}_\beta =  [\orb{\omega_1}(\xi^{\tr{(\delta+\beta)}S}), \orb{\omega_2}(\xi^{\tr{(\delta+\beta)}S})].\]

One then can form 
$$[T_{\mu_1}(\check{\vartheta}_\alpha), T_{\mu_2}(\check{\vartheta}_\alpha)]
= \left[\orb{\delta+\alpha}(\xi^{\tr[1]{\mu}S}), \orb{\delta+\alpha}(\xi^{\tr[2]{\mu}S})\right]
$$ and 
$$[T_{\mu_1}(\check{\vartheta}_\beta), T_{\mu_2}(\check{\vartheta}_\beta)]
= \left[\orb{\delta+\beta}(\xi^{\tr[1]{\mu}S}), \orb{\delta+\beta}(\xi^{\tr[2]{\mu}S})\right]
$$ 
using the polynomials calculated in \exmr{ex:chebi1}  and find 
$\delta+\alpha$ and $\delta+\beta$ as illustrated in 
\exmr{ex:testweight}.
We can then compute $\aorb{\delta}(\xi^{\tr{(\delta+\alpha)}S})$ and $\aorb{\delta}(\xi^{\tr{(\delta+\beta)}S})$ and hence $a$ and $b$.


Note that the function $f$ is a 12-sparse polynomial in the monomial basis since 
\(
   f(x,y) = a\, \aorb{\delta+\alpha}(x)+b\, \aorb{\delta+\beta}(x).
\)
Yet to retrieve its support we only need to evaluate $f$ at
 points indexed by $\wcroxx[{\WeylA[2]}]{2}$ that has cardinality $10$.
\end{example}

\color{black}


%
%
%
%


\subsection{Relative costs of the algorithms}\label{sec:eval}

There are two factors that are the main contributions to the cost of the algorithms described above: the cost of the linear algebra operations in \algr{support} or \algr{invsupport} and the needed number of function evaluations. 

For \algr{LaurInt},  one calls upon the linear algebra operations of \algr{support} to calculate the support and coefficients of the sparse polynomial that is being interpolated. This involves 
 one $|\hcross{r}|\times|\hcross{r}|$ matrix and several of 
 its $r \times r$ submatrices.  
\algr{support} is fed with the evaluation at the points 
\[\left\{\xi^{(\gamma+\alpha+\beta)^T} \;|\; 
            \alpha, \beta \in \hcross{r},\, |\gamma|\leq 1\right\} \subset \Q^n.\]
Since 
 $|\hcross{r}|\leq r\log^{n-1} (r)$ \cite[Lemma 1.4]{Lubich08}, 
 $|\hcross{r}+\hcross{r}|\leq r^2\log^{2n-2} (r)$
 and  
 $|\hcross{r}+\hcross{r}+\hcross{2}|\leq (n+1)\,r^2\log^{2n-2} (r)$.{}
This latter number is a crude upper bound on the number of 
evaluations of $f$ in  \algr{LaurInt}.
This bound was given in \cite{Sauer16a} in the context of 
the multivariate generalization of Prony's method.


Turning to sums of Chebyshev polynomials of the first kind, 
we wish to compare  
the cost of the interpolation of the $r$-sparse polynomial 
$F = \sum_{i=1}^r a_i T_{\beta_i}$,
with \algr{FKInterp},
to the cost of the the $r|\gva|$-sparse polynomial 
$f(x)=\sum_{i=1}^r\sum_{A\in\gva} a_i x^{A\beta_i} $,  
with \algr{LaurInt}.
The analysis for  the sparse interpolation of 
$F=\sum_{i=1}^r a_i U_{\beta_i}$ with \algr{SKInterp}
compared with the sparse interpolation of 
$f(x)=\sum_{i=1}^r\sum_{A\in\gva} a_i \det(A) x^{A\beta_i} $
with \algr{LaurInt} is the same.

First note that   \algr{invsupport} will involve a matrix 
of the size  $|\hcross{r}|$ and some of its submatrices of size $r$.
This is to be constrasted with \algr{LaurInt} involving  in theses cases
a matrix of size $|\hcross{|\gva| r}|$ and some of its submatrices of size 
 $|\gva| r$.

 The number of evaluations is  the cardinality of
 $\wcross{r}$ defined by Equation~\eqref{wcross}. 
$\wcross{r}$ is a superset of  $\hcross{r}+\hcross{r}+\hcross{2}$.
In the case where $\gva$ is $\WeylB[2]$ or $\WeylA[3]$, 
$\wcross{r}$ is a {\it proper} superset
and the discrepancy is illustrated in 
Figure~\ref{B2outliers} and \ref{A3outliers}. 
On the other hand there is experimental evidence that $\wcross[{\WeylA[2]}]{r}$ 
is  equal to $\hcross{r}+\hcross{r}+\hcross{2}$. 
The terms that appear in the sets
$S(\alpha,\beta,0)$, $S(\alpha,\beta,\omega_1)$, 
\ldots, $S(\alpha,\beta,\omega_n)$ (see the definition of $\wcross{r}$ given by Equation~(\ref{wcross}))
and hence in $\wcross{r}$ strongly depend on the group $\gva$. 
Specific analysis for each group would provide a refined 
bound on the cardinal of $\wcross{r}$. 

Nonetheless, taking the group structure and action of $\WeylG$ into account, one can make the following estimate. 
\prpr{raspberry} implies that  $S(\alpha,\beta,0)$ is of cardinality at most $|\gva|$
while $S(\alpha,\beta,\gamma)$ is bounded by $|\gva|^2$ in general.
Yet, the isotropy group $\gva_{\omega_i}$ of $\omega_i$ is rather large: 
among the $n$ generators of the group, $n-1$ leave $\omega_i$ unchanged.
Since $\gva_{\omega_i}$ contains the identity as well we have $|\gva_{\omega_i}|\geq n$.
Therefore 
$|S(\alpha,\beta,\omega_i)|    \leq |S(\alpha,\beta,0)| |\gva/\gva_{\omega_i}| \leq \frac{1}{n} |\gva|^2$.
Hence 
\begin{eqnarray*}|\wcross{r}|\, 
&= &
\left| \bigcup_{\substack{\alpha,\beta\in \hcross{r}\\|\gamma|\leq  1}} S(\alpha,\beta,\gamma) \right| 
\quad \leq \quad  
\left|\bigcup_{\substack{\alpha,\beta\in \hcross{r}\\i=1,\ldots, n }} S(\alpha,\beta,\omega_i) \right|
\, + \, 
\left|\bigcup _{\substack{\alpha,\beta\in \hcross{r}}} S(\alpha,\beta,0)\right| 
\\ & \leq &
\left(n\left(\frac{1}{n} |\WeylG|^2\right) +|\WeylG| \right)\,r^2\log^{2n-2}\left( r\right) 
\quad \leq \quad
 2\left(|\WeylG|\,r\right)^2\log^{2n-2}\left( r\right).
\end{eqnarray*}




This is to be compared to interpolating 
a $|\WeylG|r$-sparse Laurent polynomial that would use at most 
$$
|\,\hcross{|\gva| r}+\hcross{|\gva| r}+\hcross{2}\,| \leq 
\mathbf{(n+1)}\left(|\WeylG|\,r\right)^2\log^{2n-2}\left(|\boldsymbol{\mathcal{W}}| r\right)$$ 
evaluations. Therefore, even with this crude estimate,
the number of evaluations to be performed to apply  
\algr{FKInterp}  is  less than with
the approach using Algorithm~\ref{LaurInt} 
considering the given polynomial as  being a $|\WeylG|r$-sparse Laurent polynomial.



\begin{figure}
\hspace*{0.05\textwidth}
\includegraphics[width=0.3\textwidth]{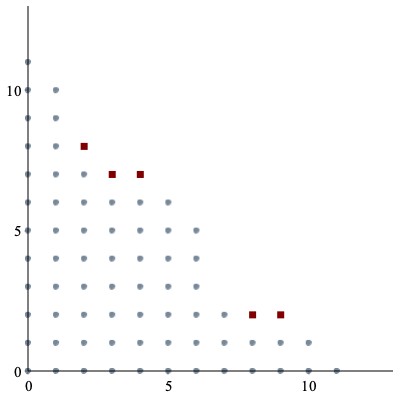}
\includegraphics[width=0.3\textwidth]{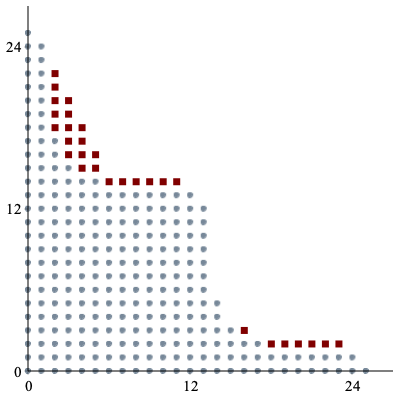}
\includegraphics[width=0.3\textwidth]{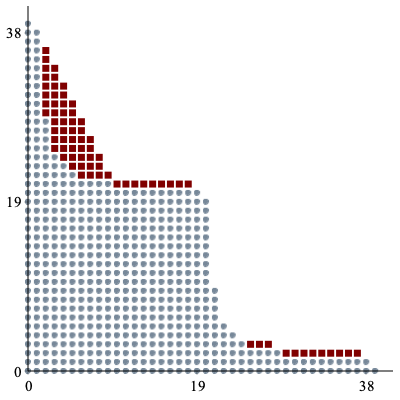}
\caption{$\wcross[{\WeylB[2]}]{r}$, for $r\in \{6,13,20\}$ : 
  the elements that do not belong to 
  $\hcross[2]{r}+\hcross[2]{r}+\hcross[2]{2}$ are represented by carmin squares.}
  \label{B2outliers}
\end{figure}

\begin{figure}
\hspace*{0.05\textwidth}
 \includegraphics[width=0.3\textwidth]{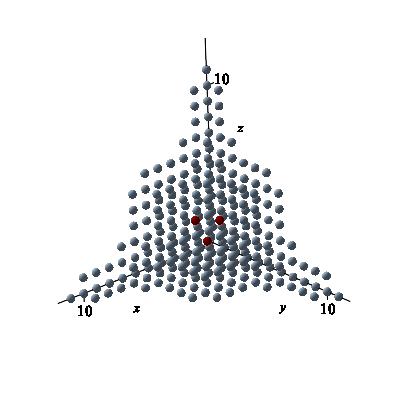}
 \includegraphics[width=0.3\textwidth]{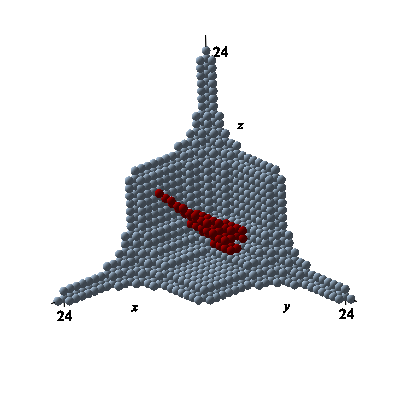}
 \includegraphics[width=0.3\textwidth]{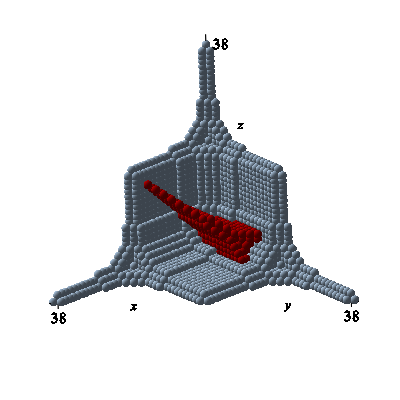}
\caption{$\wcross[{\WeylA[3]}]{r}$, for $r\in \{6,13,20\}$ : 
  the elements in purple do not belong to 
  $\hcross[3]{r}+\hcross[3]{r}+\hcross[3]{2}$.}
  \label{A3outliers}
\end{figure}

\newpage
\section{Support of a linear form on the Laurent polynomial ring}\label{sec:support}
   \label{hankel}

In \secr{interpolation} we converted the recovery of the support 
of a polynomial in the monomial or Chebyshev bases to the recovery 
of the support of a linear form. 
For 
\[f(x) = \sum_{i=1}^r a_i x^{\beta_i}, \qquad 
F(X)= \sum_{i=1}^r a_i \, T_{\beta_i}(X), \quad \hbox{ or }\quad
F(X) = \sum_{i=1}^r a_i \, U_{\beta_i}(X)\]
we respectively introduced the linear forms on $\Kx$
\[ \lifo= \sum_{i=1}^r a_i\eval[{\zeta_i}], \qquad 
    \sum_{i=1}^r a_i \sum_{A\in\gva} \eval[A \star\zeta_i],\quad \hbox{ or } \quad
    \sum_{i=1}^r a_i \sum_{A\in\gva} \det(A)\, \eval[A \star \zeta_i]\]
    where, for some chosen $\xi \in \N_{>0}$, 
    $\zeta_i =\xi^{\tr[i]{\beta}}$, $\zeta_i =\xi^{\tr[i]{\beta}S}$, 
    or $\zeta_i =\xi^{\tr{(\delta+\beta_i)}S}$.
The linear forms are known from their values at some polynomials, respectively: 
 \[ \lifo(x^\alpha)=f(\xi^\alpha), 
\quad
\lifo(\orb{\alpha})= 
F\left(\orb{\omega_1}(\xi^{\tr{\alpha}S}),\ldots,\orb{\omega_n}(\xi^{\tr{\alpha}S})\right),\]
or  
\[\lifo(\aorb{\alpha})= 
\aorb{\delta}\left(\xi^{\tr{\alpha}S}\right)\,F\left(\orb{\omega_1}\left(\xi^{\tr{\alpha}S}\right),\ldots,\orb{\omega_n}(\xi^{\tr{\alpha}S})\right).\]
This section provides the technology to recover the 
support of these linear forms. We shall either retrieve
\[ \left\{\zeta_1,\,\ldots,\zeta_r\right\} \subset \Q^n, \quad \hbox{ or } \quad 
\left\{ \left(\orb{\omega_1}(\zeta_i), \ldots,\, \orb{\omega_n}(\zeta_i) \right)
            \;|\; 1\leq i\leq r\right\}  \subset \N^n \]

Identifying the support of a linear form on a polynomial ring $\K[x]$
already has  applications in optimization, tensor decomposition and cubature 
\cite{Abril-Bucero16c,Abril-Bucero16,Bernardi18,Brachat10,Collowald18,Lasserre10,Laurent09}.
How to take advantage of symmetry in some of these applications appears in
\cite{Collowald15,Gatermann04,Riener13}. 
To a linear form $\lifo:\K[x] \rightarrow \K$ one associates 
\cite{Brachat10,Collowald18,Mourrain18,Power82} a Hankel
operator $\Hanc:\K[x] \rightarrow \K[x]^*$ 
whose kernel $\kerk$ is the  ideal of the support $\left\{\zeta_1,\,\ldots,\,\zeta_n\right\}$ of $\lifo$. 
We can compute directly these points  as eigenvalues of  
the multiplication maps on the quotient algebra $\K[x]/\kerk$.

The present application to sparse interpolation is related to a 
multivariate version of Prony's method, as tackled 
in \cite{Kunis16,Mourrain18,Sauer16a}. 
Contrary to the previously mentioned applications, where
the symmetry is given by the linear action of a finite group on the 
ambient space of the support, here  the Weyl groups  act 
linearly on $\Kx$ but nonlinearly on the ambient space of the support.
Thanks to \thmr{kumquat}, we can satisfy ourselves with recovering 
\emph{only} the values of the freely generating invariant polynomials on the support,
 \textit{i.e.} 
$\left\{ \left(\orb{\omega_1}(\zeta_i), \ldots,\, \orb{\omega_n}(\zeta_i) \right)
            \;|\; 1\leq i\leq r\right\}$.

In Section~\ref{HankelMultiplication} we review the definitions of Hankel operators associated to  a linear form, multiplication maps and their relationship to each other in the context of  $\Kx$ rather than $\K[x]$. In Section~\ref{HankelSupport} we present an algorithm to calculate the matrix representation of certain multiplication maps and determine the support of the original linear form as eigenvalues of the multiplication maps. In Section~\ref{hankelinv}, to retrieve the orbits forming the support of an invariant or semi-invariant form, 
we introduce the Hankel operators
$\Hanc:\IKx \rightarrow \left(\IKx[\chi]\right)^*$, where
$\IKx$ is the ring of invariants for the action of the Weyl group $\gva$ on $\Kx$;
$\IKx[\chi]$ is the $\IKx$-module  of $\chi$-invariant polynomials where
$\chi :\gva \rightarrow \left\{1,-1\right\}$ is 
either given by $\chi(A)=1$ or $\chi(A)=\det A$, depending whether 
we consider Chebyshev polynomials of the first or second kind. In this latter section, in analogy to the previous sections,  we also introduce the appropriate multiplication maps and give an  algorithm to determine the support of the associated $\chi$-invariant linear from in terms of the eignevalues of these multiplication maps.


The construction could be extended to other group actions on the ring of
 (Laurent) polynomials. Yet
 we shall make use of the fact that, for a Weyl group $\gva$, 
 $\IKx$ is isomorphic to a polynomial ring and   
$\IKx[\chi]$ is a free $\IKx$-module of rank one.

\subsection{Hankel operators and multiplication maps} 
\label{HankelMultiplication}

We consider a commutative $\K$-algebra $\SKx$ and a  $\SKx$-module $\Skep$.
$\SKx$ will later be either $\Kx$ or the invariant subring $\IKx$
while $\Skep$ will be either $\Kx$, $\IKx$ or $\aorb{\delta} \IKx$, 
the module of skew-invariant polynomials (\lemr{iraty}).
Hence $\Skep$ shall be a free $\SKx$-modules of rank one: 
there is an element that we shall denote $\ba $ 
in $\Skep$ s.t. $\Skep = \ba   \SKx$.
In the cases of interest $\ba   $ is either $1$ or $\aorb{\delta}$.
We shall  keep the explicit mention of $\ba   $ 
though the $\SKx$-module isomorphism between $\SKx$ and $\Skep$ 
would allow us to forego the use of $\Skep$.

$\SKx$ and $\Skep$ are also $\K$-vector spaces.
To a $\K$-linear form $\lifo$ on $\Skep$ we associate a Hankel operator
$\Hanc:\SKx \rightarrow \Skel$, where $\Skel$ is the dual of $\Skep$, 
\textit{i.e.} the $\K$-vector  space of $\K$-linear forms on $\Skep$. 
The kernel of this operator is an ideal $\kerk$ in $\SKx$, considered as a ring.
The matrices of the multiplication maps in $\SKx/\kerk$ are given in
terms of the matrix of $\Hanc$.

\subsubsection*{Hankel operator}

For a linear form $\lifo \in \Skel$, the associated \emph{Hankel  operator} 
$\SHanc$ is the $\K$-linear map
\[ \SHanc : \begin{array}[t]{ccc}\SKx &\to &\Skel \\ 
                         p &\mapsto & \lifo_p ,\end{array}
 \quad \hbox{ where } \quad
\lifo_p : \begin{array}[t]{ccc}\Skep &\to &\K \\
                     q& \mapsto &\lifo( q\, p).  \end{array}
\]


If $U$ and $V$ are $\K$-linear subspaces of 
$\SKx$ and $\Skep$ respectively we define $\SHanc_{|U,V}$ 
 using the restrictions $\lifo_{p|V}$ of $\lifo_p$  to $V$:
\[ \SHanc_{|U,V} : \begin{array}[t]{ccc} U &\to & V^* \\ 
                         p &\mapsto & \lifo_{p|V} ,\end{array}\]
Assume $U$  is the $\K$-linear span $\lspan{B}=\lspan{b_1,\ldots,b_r}$ of 
a linearly independent set   $B=\left\{b_1,\ldots,b_r\right\}$ in  $\SKx$
and $V$  is the $\K$-linear span 
$\lspan{\ba   c_1,\ldots,\ba   c_s}$ in $\Skep$, 
denoted 
$\lspan{\ba   C}$, 
where $C=\left\{c_1,\ldots,c_s\right\}$ is a linearly independent subset  of  $\SKx$.
Then the matrix of $\SHanc_{|U,V}$ 
in $B$ and the dual basis of $\ba C$ is 
\[H_1^{C,B} = \left( \lifo(\ba   c_i\, b_j )
             \right)_{\substack{1 \leq i\leq s \\ 1\leq j \leq r}} .\]

\newcommand{\Quo}{\ensuremath{\Skep/\ba   \Skerk}}
\newcommand{\Quod}{\ensuremath{\left(\Skep/\ba   \Skerk \right)^*}}

The kernel of $\SHanc$
$$\Skerk = \left\{ p \in \SKx \;|\; \lifo_p = 0 \right\}= \left\{ p
  \in \SKx \;|\; \lifo(q\,p) = 0, \;\forall q\in\Skep \right\}$$ 
is an ideal of $\SKx$.
We shall consider both the quotient spaces 
$\SKx/\Skerk$ and $\Skep/\ba   \Skerk$ where 
$\ba \Skerk$ is the submodule $\Skerk\Skep$ of $\Skep$.


\begin{lemma} \label{algebrablabla} The image of $\SHanc$ lies in $ \left(\ba   \Skerk\right)^\perp$
and $\SHanc$  induces an injective morphism 
$\SHank : \SKx/\Skerk\rightarrow \Quod$ that has the following diagram commute.

\hspace*{0.3\textwidth}
\begin{tikzpicture}
\node (A) at (0,0) {$\SKx$};
\node (C) at (4,0) {$\left(\ba   \Skerk\right)^\perp$};
\node (D) at (4,-2) {$\Quod$};
\node (E) at (0,-2) {$\SKx/ \Skerk$};

\draw[->] (A) -- (C) node[midway,above] {\SHanc};
\draw[<->] (C) -- (D) node[midway,right] {$\cong$};
\draw[->] (A) -- (E) node[near start,left] {$\pi$};
\draw[->] (E) -- (D) node[midway,below] {\SHank};
\end{tikzpicture}
\end{lemma}

\begin{proof} A basis of  $\SKx/\Skerk$ is the image by the natural projection
$\pi:\SKx\rightarrow\SKx/\Skerk$  of a linearly independent set
$C\subset \SKx$ s.t. $\SKx=  \lspan{ C }\oplus \Skerk$.
Hence $\Skep = \lspan{ \ba   C } \oplus \ba    \Skerk $.
Recall from linear algebra, 
see for instance \cite[Proposition V, Section 2.30]{Greub67}, 
that this latter equality  implies:
\[\Skel = \left(\ba   \Skerk\right)^\perp \oplus \left(\ba   C \right)^{\perp}
\quad \hbox{ and } \quad
\begin{array}[t]{ccc} \left(\ba   \Skerk\right)^\perp & \rightarrow
    &{\lspan{\ba   C}}^* \\ \Phi & \mapsto & \Phi_{|\lspan{\ba   C}}\end{array}
\hbox{is an isomorphism,}
\]
where, for  any set $V\subset\Skep$, 
$\ds V^\perp= \left\{ \Phi \in \Skel \;|\;
  \Phi(v)=0, \forall v\in V\right\}$
is a $\K$-linear subspace  of $\Skel$.

Note that the image of $\SHanc$ lies in $ \left(\ba   \Skerk\right)^\perp$.
With the natural identification of 
$\lspan{\ba    C }^*$ with $\Quod$,
the factorisation of  $\SHanc$  by  the natural projection
$\pi:\SKx\rightarrow\SKx/\Skerk$  defines  the injective morphism
$\SHank: \SKx/\Skerk\rightarrow \Quod$
 through the announced commuting diagram.
 \end{proof}

If the rank of $\SHanc$ is finite and equal to $r$, 
then the dimension of $\SKx/\Skerk$, $\Quo$ and $\Quod$, 
as  $\K$-vector spaces, is $r$
and the injective  linear operator $\SHank$ is then an isomorphism.
The  point here is the following criterion for 
detecting bases of $\SKx/\Skerk$.

\begin{theorem} \label{hankbasis}
Assume that $\rank \SHanc = r < \infty$ and consider 
$B = \{ b_1,\ldots,b_r \}$ and $C=\{c_1,\ldots,c_r\}$ subsets of $\SKx$.
Then the image of $B$  and $C$ by $\pi:\SKx \rightarrow\SKx/\Skerk$  
are both bases of $\SKx / \Skerk$ if and only if the matrix $H_1^{C,B}$ is non-singular.
\end{theorem}

\begin{proof} 
Assume that $B$ and $C$  are both bases for $\SKx/\Skerk$,
we can identify $\SKx/\Skerk$ with $\lspan{B}$ and $\Quo$ with
$\lspan{\ba   C}$. 
Hence 
$H_1^{C,B}$ is the matrix of $\SHank$ in the basis $B$ and the dual basis of $\ba   C$.
Since $\SHank$ is an isomorphism, 
 $H_1^{C,B}$ is nonsingular. 

Assume $H_1^{C,B}$ is nonsingular.
We need to show that $B$ and $C$ are linearly independent modulo
$\Skerk$, i.e. that any linear combination of 
their elements that belongs to the ideal is trivial.
Take  $a = (a_1,\dots,a_r) \in \K^r$ 
 such that $a_1 b_1 + \dots + a_r b_r \in\Skerk$.
Using the definition of $\Skerk$, we get
$ a_1 \lifo(\ba    c_i \,b_1) + \dots + a_r \lifo(\ba    c_i\, b_r) = 0 \quad \forall i = 1,\dots,r $.
These equalities amount to  $H_1^{C,B} a = 0$ and thus $a=0$.
Similarly  $a_1 c_1 + \dots + a_r c_r \equiv 0 \mod \Skerk$ leads  to 
$\tr{a} H_1^{C,B}  = 0$ and hence $a=0$.
\end{proof}

\subsubsection*{Multiplication maps}

We now assume that the Hankel operator $\SHanc$ associated
to $\lifo$ has finite rank $r$.
Then  $\SKx/\Skerk$ is of dimension $r$ when considered as a linear space
over $\K$.
For $p\in \SKx$, consider the multiplication map
\begin{equation}\label{multiplication}
 \widehat{\Multi}_p :\begin{array}[t]{ccc} 
               \SKx & \rightarrow & \SKx \\
               q & \mapsto & q\,p  \end{array}
   \quad \hbox{ and } \quad 
  \Multi_p :\begin{array}[t]{ccl} 
               \SKx/\Skerk & \rightarrow & \SKx/\Skerk\\
               q' & \mapsto & \pi(q\,p) 
                \hbox{ where } q\in \SKx \hbox{ satisfies } \pi(q)=q'. \end{array} \;
\end{equation}
 ${\Multi}_p$ is a well defined linear  map respecting  the
following  commuting diagram \cite[Proposition 4.1]{Cox05} \\[0pt]
\hspace*{0.3\textwidth}
\begin{tikzpicture}
\node (A) at (0,0) {$\SKx$};
\node (C) at (4,0) {$\SKx$};
\node (D) at (4,-2) {$\SKx/ \Skerk$};
\node (E) at (0,-2) {$\SKx/ \Skerk$};
\draw[->] (A) -- (C) node[midway,above] {$\widehat{\Multi}_p$};
\draw[->] (C) -- (D) node[midway,right] {$\pi$};
\draw[->] (A) -- (E) node[midway,left] {$\pi$};
\draw[->] (E) -- (D) node[midway,below] {$\Multi_p$};
\end{tikzpicture}

Let us temporarily introduce the Hankel operator 
$\SHanc_p$ associated to $\lifo_p$.  
This is the map defined by   $\SHanc_p=\SHanc \circ \widehat{\Multi}_p$.
Therefore the image of 
$\SHanc_p$ is included in the image of  $\SHanc$ and 
$\ker \SHanc \subset \ker \SHanc_p$. 
We can thus construct the maps $\SHank_p$ 
that satisfy $\SHanc_p=\SHank_p \circ \pi$.
Then $\SHank_p= \SHank \circ \Multi_p$ and we have the following
commuting diagram.

\hspace*{0.3\textwidth}
\begin{tikzpicture}
\node (A) at (-3,2) {$\SKx$};
\node (B) at (3,2) {$\left(\ba   \Skerk\right)^\perp$};
\node (E) at (0,0.7){$\SKx$};
\node (C) at (3,-2) {$\left(\Skep/ \ba\Skerk\right)^*$};
\node (D) at (-3,-2) {$\SKx/ \Skerk$};
\node (F) at (0,-0.7) {$\SKx/ \Skerk$};

\draw[->] (A) -- (B) node[midway,above] {$\SHanc_p$};
\draw[<->] (B) -- (C) node[midway,right] {$\cong$};
\draw[->] (D) -- (C) node[midway,below] {$\SHank_p$};
\draw[->] (A) -- (D) node[midway,left] {$\pi$};
\draw[->] (A) -- (E) node[midway,below left] {$\widehat{\Multi}_p$};
\draw[->] (D) -- (F) node[midway, above left] {${\Multi}_p$};
\draw[->] (E) -- (B) node[midway,below right] {$\SHanc$};
\draw[->] (F) -- (C) node[midway,above right] {$\SHank$};
\draw[->] (E) -- (F) node[midway,left] {$\pi$};
\end{tikzpicture}

\begin{theorem} \label{multmat}
Assume the Hankel operator $\SHanc$ associated to the linear form
$\lifo$ has finite rank $r$.
Let $B=\left\{b_1,\,\ldots,\,  b_r\right\}$ 
and $C=\left\{c_1,\,\ldots,\,  c_r\right\}$ be bases of $\SKx/\Skerk$. 
Then  the matrix  $M_p^B$ of the multiplication  by an element $p$ of $\SKx$
 in $\SKx/\Skerk$ is given by 
\[   M_p^B = \left( H_1^{C,B} \right)^{-1} \,H_p^{C,B}   
          \quad \hbox{ where } \quad 
H_1^{C,B} = \left( \lifo(\ba c_i\, b_j) \right)_{1 \leq i,j \leq r}  
\;\hbox{ and }\;
 H_p^{C,B} = \left( \lifo(\ba c_i\, b_j\, p) \right)_{1 \leq i,j \leq r} \]
\end{theorem}

\begin{proof}
The matrix of $\SHank_p$ in $B$ and the dual basis of $\ba C$ is 
$ H_p^{C,B} $. Then $H_p^{C,B} = H_1^{C,B} \, M_p^B $
since $\SHank_p=\SHank \circ \Multi_p$. 
From \thmr{hankbasis}, $H_1^{C,B} $ is invertible.
\end{proof}



\subsection{Support of a linear form on $\Kx$} \label{HankelSupport}

We now consider $\SKx$ and $\Skep$ to be the ring of Laurent polynomials  
$\Kx$.
As before,  the evaluations 
$\eval[\zeta]:\Kx \rightarrow \K$  at a point $\zeta\in\left(\Ks\right)^n$ are defined as follow:
For $p\in\Kx$, $\eval[\zeta] (p)=p(\zeta)$.
For $a_1,\,\ldots, \, a_r\in\Ks$ and
distinct $\zeta_1,\,\ldots,\,\zeta_r\in (\Ks)^n$  
we write $\lifo = \sum_{i=1}^r a_i \, \eval[\zeta_i ]$
for the  linear form 
$$\lifo : \begin{array}[t]{ccc} \Kx & \rightarrow & \K \\
      p & \mapsto & \ds \sum_{i=1}^r a_i \, p(\zeta_i). \end{array} $$
In this section we characterize such a linear form in terms of its associated
Hankel operator. We show how to compute $\zeta_1,\,\ldots,\,\zeta_r$ from
the knowledge of the values of $\lifo$ on a finite dimensional subspace  of $\Kx$.

\subsubsection{Determining a basis of the quotient algebra}

\begin{theorem} \label{mungo}
If $\ds \lifo = \sum_{i=1}^r a_i \, \eval[\zeta_i]$, 
where $a_i\in \Ks$, and $\zeta_1,\,\ldots,\,\zeta_r $ are distinct points
in $(\Ks)^n$ then the associated Hankel
operator $\Hanc$ has finite rank $r$ and its kernel $\kerk$ is the
annihilating ideal of $\left\{ \zeta_1,\,\ldots, \zeta_r\right\}$.
\end{theorem}

\begin{proof}
It is easy to see that $p(\zeta_1)=\ldots=p(\zeta_r) =0$ implies that
$p\in \kerk$. For the converse inclusion,
consider some  interpolation polynomials $p_1,\,\ldots,\,p_r$ at
$\zeta_1,\,\ldots,\, \zeta_r$, i.e. $p_i(\zeta_i)=1$ and  $p_j(\zeta_i)=0$
when $i\neq j$ \cite[Lemma 2.9]{Cox05}.
For $q \in \kerk$ we have $\lifo(q \,p_i)=0$ and thus $a_i
\,q(\zeta_i) =0$.
Hence $\kerk$ is the annihilating ideal of  
$\left\{ \zeta_1,\,\ldots,  \zeta_r\right\}$. 
It is thus a radical ideal with $\dim_{\K} \Kx/\kerk=r$.
\end{proof}

\thmr{hankbasis} gives necessary and sufficient
condition for a set $B=\left\{b_1,\ldots,b_r\right\}$ in $\Kx$ to be a basis of
$\Kx/\kerk$ when the dimension of this latter, as a $\K$-vector space, is $r$.
This condition is that the matrix
\(H_1^B = \left( \lifo(b_i b_j ) \right)_{1 \leq i,j \leq r} \)
is nonsingular.
The problem of where to look for this basis was settled in
\cite{Sauer16a} where the author introduces \emph{lower sets} 
and the \emph{positive octant of the hypercross of order $r$}.

A subset $\Gamma$ of $\N^n$ is a \emph{lower set} if
whenever $\alpha+\beta\in\Gamma$, $\alpha,\beta\in\N^n$, 
then $ \alpha\in \Gamma$.
The positive octant of the hypercross of order $r$ is 
$$\hcross[n]{r}= \left\{ \alpha \in \N^n \;\left|\; \prod_{i=1}^r(\alpha_i+1) \leq r\right.\right\}.$$ 
It is the union all the lower sets of cardinality $r$ or less 
\cite[Lemma 10]{Sauer16a}. 
%
We  extend \cite[Corollary 11]{Sauer16a}  for further use in \secr{hankelinv}.

\begin{proposition}   \label{bounty} 
Let $\leqslant$ be an order on $\N^n$ 
such that $0 \leqslant \gamma$ and   
$\alpha \leqslant \beta \Rightarrow \alpha+\gamma \leqslant \beta+\gamma$ 
for all $\alpha,\beta,\gamma\in \N^n$.
Consider two families  of polynomials  
$\left\{P_\alpha \,|\, \alpha\in \N^n\right\}$ 
and $\left\{Q_\alpha \,|\, \alpha\in \N^n\right\}$
in $\KX$ such that 
$P_\alpha = \sum_{\beta\leqslant\alpha} p_\beta X^\beta $
and 
$Q_\alpha = \sum_{\beta\leqslant\alpha} q_\beta X^\beta $
with $p_\alpha,\,q_\alpha  \neq 0$.

If  $J$ is an ideal in $\K[X]=\K[X_1,\ldots,X_n]$ such that 
$\dim_{\K} \K[X]/J =r$
then  there exists a lower set $\Gamma$ of cardinal $r$ such that
both 
$\left\{P_\alpha \, |\, \alpha\in \Gamma \right\}$
and 
$\left\{Q_\alpha \, |\, \alpha\in \Gamma \right\}$
are bases of $\K[X]/J$.
\end{proposition}

\begin{proof}
For the chosen term order $\leqslant$, a Gr\"obner basis 
of $J$ defines a lower set $\Gamma$ that has $r$
monomials and is a basis of $\K[x]/J$ \cite[Chapter 2]{Cox05}.

Consider a polynomial 
$P=\sum_{\beta\in \Gamma}a_\beta P_\beta$, for some $a_\beta\in \K$ not all zero.
Take $\alpha$ to be highest element of $\Gamma$ for which $a_\alpha\neq 0$.
Then  $X^\alpha$ is the leading term of $P$. 
As $X^\alpha$ does not belong to the initial ideal, $P\notin J$ \cite[Chapter 2]{Cox15}.
It follows that $\left\{P_\alpha\, |\, \alpha \in \Gamma\right\}$ is 
linearly independent modulo $J$ and hence is a basis of $\K[X]/J$. 
The same is applies to $\left\{Q_\alpha\, |\, \alpha \in \Gamma\right\}$.
\end{proof}

\begin{corollary}  \label{crossmyheart}   
If $I$ is an ideal in $\Kx$ such that 
$\dim_{\K} \Kx/I =r$ then 
$\Kx/I$ admits  a basis in
\( \left\{ x^\alpha \,|\, \alpha \in \hcross[n]{r}\right\}.\)
\end{corollary}

\begin{proof} 
A monomial basis of $\K[x]/J$, where $J=I\cap \K[x]$, 
is a basis  for $\Kx/I$. 
\end{proof}

\subsubsection{Eigenvalues and eigenvectors of the multiplication matrices}
\label{zeros}

The eigenvalues of the multiplication map $\Multi_p$, 
introduced in Equation~\Ref{multiplication},
are the values of $p$ on the variety of
$\kerk$; as $\kerk$ is a radical ideal, this is part of the following
result, which is a simple extension of 
\cite[Chapter 2, Proposition 4.7 ]{Cox05} 
to the Laurent polynomial ring. 
The proof appears as a special case of the later \prpr{kohlrabi}.

\begin{theorem} \label{persimon}
Let $I$ be a radical ideal in $\Kx$ whose variety consists of $r$ distinct points
  $\zeta_1,\, \ldots, \zeta_r$ in $(\bKs)^n$ then:
\begin{itemize}
\item A set  $B=\left\{b_1,\ldots,b_r\right\}$ is a basis of $\Kx/I$ if and only if the matrix 
 $W_\zeta^B = \left( b_j(\zeta_i) \right)_{1 \leq i,j \leq r}$ is non singular;
\item The matrix $M_p^B$ of the multiplication $\Multi_p$ by $p$ in
  a basis $B$ of $\Kx/I$ satisfies
$W_\zeta^B \, M_p^B =D_\zeta^p\,W_\zeta^B$ where $D_\zeta^p$ is the diagonal matrix 
$\diag ( p(\zeta_1), \,\ldots,\, p(\zeta_r) )$.
\end{itemize}
\end{theorem}
This theorem gives us a basis of left eigenvectors for $M_p^B$:
The $i$-th row of $W_\zeta^B $,
$\begin{bmatrix} b_1(\zeta_i) & \ldots & b_r(\zeta_i)\end{bmatrix}$, 
is a left eigenvector associated to the eigenvalue $p(\zeta_i)$.
One can furthermore observe that
\begin{equation} \label{hankelfacto}
 H_1^{C,B} = \tr{\left(W_\zeta^C\right)} \,A \,W_\zeta^B \hbox{ where } 
  A=\diag(a_1,\ldots, a_r).
\end{equation}

\subsubsection{Algorithm}


Assuming that a linear form $\lifo$ on $\Kx$ is a weighted sum 
of evaluations at some points of $(\Ks)^n$, we wish to determine its support and its coefficients.
We assume  we know  the cardinal $r$ of this support and that we can evaluate 
$\Omega$ at the monomials $\left\{x^\alpha\right\}_{\alpha\in\N^n}$.
In other words, we assume that $ \lifo = \sum_{i=1}^r a_i\eval[\zeta_i]$
where $\zeta_1,\,\ldots,\,\zeta_r\in(\Ks)^n$ 
and then $a_1,\ldots, a_r\in\Ks$ are the unknowns.
For that we have access as input to 
$\left\{\Omega\left(x^{\alpha+\beta+\gamma}\right) | \alpha,\beta
  \in\hcross[n]{r};  |\gamma| \leq 1 \right\}$.

The ideal $\kerk$ of these points is the kernel of the Hankel operator associated to $\lifo$. 
One strategy would consist in determining a set of generators, or even a Gr\"obner basis, 
of this ideal and then find its roots with a method to be chosen. 
In the present case there is nonetheless the possibility to directly form the matrices of 
the multiplication maps in $\Kx/\kerk$ (applying \thmr{multmat}) 
once a basis for  $\Kx/\kerk $ is determined 
(applying \thmr{hankbasis} and \colr{crossmyheart}).
The key fact 
that is used is that the set of $j^{th}$ coordinates of the $\zeta_i$,  
$\{\zeta_{1,j}, \ldots ,\zeta_{r,j}\}$, are the left eigenvalues of the multiplication map 
$\Multi_{x_j}: \Kx/\kerk \rightarrow \Kx/\kerk$, 
where $\Multi_{x_j}(\overline{p}) = \overline{x_j} \overline{p}$.
The matrices of these maps commute and are simultaneously diagonalizable 
(\thmr{persimon} or \cite[Chapter 2, \S 4, Exercise 12]{Cox15}).  
One could calculate the eigenspaces of the first matrix  and proceed by induction 
to give such a common diagonalization since these eigenspaces are left invariant 
by the other matrices. 
A more efficient approach given in the algorithm is to take a generic linear combination 
of these matrices that ensures that this new matrix has distinct eigenvalues and calculate 
a basis of eigenvectors for it.  
In this basis each of the original matrices is diagonal.


\begin{algoy} \label{support} {\Algf Support \& Coefficients}
{}

\In{ $r \in \N_{>0}$ and 
   $\left\{\Omega\left(x^{\gamma+\alpha+\beta}\right) \;|\; 
            \alpha, \beta \in \hcross{r},\, |\gamma|\leq 1\right\}${}
    }

\Out{\begin{itemize}
\item The points
    $ \zeta_i=\left[\zeta_{i,1},\,\ldots,\,\zeta_{i,n} \right]\in\K^n$, 
    for $1\leq i\leq r$,
    \item The vector $\left[a_1,\ldots, a_r\right]\in (\Ks)^n$ of coefficients,
  \end{itemize}
   such that    $\ds \lifo = \sum_{i=1}^r a_i\eval[\zeta_i]$.
    }

\begin{italg} 
  \item Form the matrix 
         $H_0^{\hcross{r}}=\left[\Omega\left(x^{\alpha+\beta}\right) \right]_{ 
            \alpha, \beta \in \hcross{r} } $

  \item Determine a lower set $\Gamma$ within $\hcross{r}$ of cardinal $r$ 
        such that the principal submatrix $H_0^\Gamma$ 
          indexed by $\Gamma$ is nonsingular.

        {\small \emph{\hspace*{\stretch{3}}  \% 
         $\Gamma =\left\{0,\,\gamma_2,\,\ldots,\, \gamma_r\right\}$ and 
          $\left\{x^\gamma \,|\, \gamma\in \Gamma\right\}$
           is a basis of $\Kx/\kerk$ (\thmr{hankbasis}).}}

  \item Form the matrices 
        $H_{j}^\Gamma =
            \left[ \Omega\left( x_j\, x^{\alpha+\beta}\right)\right]_{\alpha, \beta \in \Gamma }$
        and the matrices
         $M_{j}^\Gamma = \left(H_1^\Gamma\right)^{-1} H_{j}^\Gamma$,   
        for $1\leq j \leq n$.
        
         {\small \emph{
         \hspace*{\stretch{3}} 
         \% $M_j^\Gamma$ is the matrix of multiplication by $x_j$ 
            in $\Kx/\kerk$ (\thmr{multmat})
         \\[0pt] \hspace*{\stretch{3}} 
         \%  The matrices $M_{1}^\Gamma,\,\ldots,\,M_{n}^\Gamma$ 
          are simultaneously diagonalisable (\thmr{persimon}).
        }}

  \item Consider $L = \ell_1 M_{1}+\ldots+\ell_n M_{n} $ 
         a \emph{generic} linear combination of $M_{1},\,\ldots,\,M_{n}$

         {\small\emph{\hspace*{\stretch{3}}  \% 
          The eigenvalues of $L$ are 
          $ \lambda_i=\sum_{j=1}^n \ell_j \zeta_{i,j}$, 
          for $1\leq i\leq r$.
          For most 
            $\left[{\ell}_1,\,\ldots,\,{\ell}_n\right]\in\K^n$ 
           they are distinct\footnote{We note that in forming $L$ we desire that the eigenvalues of $L$ are distinct. This is violated only when the characteristic polynomial $P_L(x)$ of $L$ has repeated roots. This latter condition is given by the vanishing of the resultant $Res_x(P_L, \frac{dP_L}{dx})$ which yields a polynomial condition on the $\ell_i$ that fail to meet the required condition.}.
           }}

  \item Compute $W$ a matrix whose rows are $r$ linearly independent 
        left eigenvectors of $L$ appropriately nomalized

       {\small \emph{
       \hspace*{\stretch{3}} \% A left eigenvector associated to  $\lambda_i$ is a nonzero multiple of the row vector
       $\left[1, (\zeta_i)^{\gamma_2},\ldots,(\zeta_i)^{\gamma_r}  \right]$ (\thmr{persimon})
       \\[0pt]
       \hspace*{\stretch{3}}\% The normalization of the first component to $1$ allows us to assume  the rows of $W$ are exactly these vectors. 
        }}

  \item For  $1\leq j \leq n$,
        determine the matrix 
              $D_{j}=\diag(\zeta_{1,j},\,\ldots,\,\zeta_{r,j})$ 
              such that
              $W M_{x_j}^\Gamma = D_{j} W$. 
 
  \item For  $1\leq i \leq r$, form the points  
     $ \zeta_i=\left[\zeta_{i,1},\,\ldots,\,\zeta_{i,n} \right]$           
     from the diagonal entries of the matrices $D_j$,  $1\leq j \leq n$.



  \item Determine the matrix $\diag(a_1,\ldots,a_r)$ such that 
      $H_0^\Gamma=\tr{W} A W$.

       {\small \emph{
       \hspace*{\stretch{3}} 
       \% Only the first row of the left and right handside matrices need to be considered, \\[0pt] \hspace*{\stretch{3}}
       \% resulting in the linear system
        $\ds \begin{bmatrix} a_1 & \ldots & a_r\end{bmatrix} W  = 
       \begin{bmatrix} \lifo(1) & \lifo(x^{\gamma_2})\ldots & \lifo(x^{\gamma_r})\end{bmatrix}$}}

\end{italg}
\end{algoy}


Depending on the elements of $\Gamma$  it might be possible to  retrieve 
the coordinates of the points $\zeta_j$ directly 
from $W$. 
The easiest case is when 
$\tr{[1,\,0\,\ldots,\,0]},\,\ldots,\,\tr{[0\,\ldots,\,0,\,1]} \in\Gamma$~:
the coordinates of $\zeta_j$ can be read directly from the normalized left eigenvectors of $L$, \textit{i.e.} the rows of $W$.

The determination of a lower set $\Gamma$ of cardinality $r$
 whose associated principal submatrix is not singular is actually  
 an algorithmic subject on its own. 
It is strongly tied to determining the Gr\"obner bases of $\kerk$ 
and is the focus of, for instance, \cite{Berthomieu17,Sakata09}.
In a complexity meticulous approach to the problem, 
one would not form the matrix $H_0^{\hcross{r}}$ 
at once, but construct $\Gamma$ and the associated submatrix 
degree by degree or following some term order. 
The number of  evaluations of the function to interpolate is then reduced.
This actual number of evaluation heavily depends on the shape of  $\Gamma$.

\begin{example} \label{hankex} In \exmr{spinterp} we called on \algr{support} with $r=2$ and 
$$\lifo\left(x^{\gamma_1}y^{\gamma_2}\right)=f\left(\xi^{\gamma_1},\xi^{\gamma_2}\right)
= a\, \xi^{\alpha_1\gamma_1+\alpha_2\gamma_2} +b\,\xi^{\beta_1\gamma_1+\beta_2\gamma_2}.$$
Hence
  $$H_0^{\hcross[2]{2}}= \begin{bmatrix} 
       f(\xi^0,\xi^0) & f(\xi^1,\xi^0) &f(\xi^0,\xi^1) \\
       f(\xi^1,\xi^0) & f(\xi^2,\xi^0) & f(\xi^1,\xi^1) \\
       f(\xi^0,\xi^1) & f(\xi^1,\xi^1) & f(\xi^0,\xi^2) 
       \end{bmatrix}
       =
       \begin{bmatrix} 
       a +b & a \xi^{\alpha_1}+b \xi^{\beta_1} & a \xi^{\alpha_2} +b \xi^{\beta_2} \\ 
       a \xi^{\alpha_1} +b \xi^{\beta_1} & a \xi^{2\alpha_1}+b \xi^{2\beta_1} & a \xi^{\alpha_1+\alpha_2}+b \xi^{\beta_1+\beta_2} \\
       a \xi^{\alpha_2} +b \xi^{\beta_2} & a \xi^{\alpha_1+\alpha_2}+b \xi^{\beta_1+\beta_2}  & a \xi^{2\alpha_2}+b \xi^{2\beta_2}  
       \end{bmatrix}$$
       The lower sets of cardinality 2 are 
        $\Gamma_1=\left\{\tr{\begin{bmatrix}0 &0\end{bmatrix}},\,\tr{\begin{bmatrix}1 &0\end{bmatrix}}\right\}$
       and $\Gamma_2=\left\{\tr{\begin{bmatrix}0 &0\end{bmatrix}},\,\tr{\begin{bmatrix}0 &1\end{bmatrix}}\right\}$.
 One can check that the determinant of $H_0^{\hcross[2]{2}}$ is zero while the determinants of the principal
submatrices indexed by $\Gamma_1$ and $\Gamma_2$ are respectively 
 $ab(\xi^{\alpha_1}-\xi^{\beta_1})^2$ and $ab(\xi^{\alpha_2}-\xi^{\beta_2})^2$.    
 Hence, whenever $\alpha_1\neq \beta_1$, $\Gamma_1$ is a valid choice, i.e. $H_0^{\Gamma_1}$ is non singular. 
 Similarly for $\Gamma_2$ when $\alpha_2\neq \beta_2$. 

 Let us assume we can take $\Gamma=\Gamma_1$. We form:
 $$H_0^{\Gamma}= \begin{bmatrix} 
       f(\xi^0,\xi^0) & f(\xi^1,\xi^0)  \\
       f(\xi^1,\xi^0) & f(\xi^2,\xi^0) 
       \end{bmatrix}
       =
       \begin{bmatrix} 
       a +b & a \xi^{\alpha_1}+b \xi^{\beta_1} \\ 
       a \xi^{\alpha_1} +b \xi^{\beta_1} & a \xi^{2\alpha_1}+b \xi^{2\beta_1} 
       \end{bmatrix},$$
 $$H_1^{\Gamma}= \begin{bmatrix} 
       f(\xi^1,\xi^0) & f(\xi^2,\xi^0) \\
       f(\xi^2,\xi^0) & f(\xi^3,\xi^0)   
       \end{bmatrix}
       =
       \begin{bmatrix} 
       a \xi^{\alpha_1}+b \xi^{\beta_1}  &  a \xi^{2\alpha_1}+b \xi^{2\beta_1}\\ 
        a \xi^{2\alpha_1}+b \xi^{2\beta_1}  &  a \xi^{3\alpha_1}+b \xi^{3\beta_1}
       \end{bmatrix}, $$
 \hbox{ and } 
$$   H_2^{\Gamma}= \begin{bmatrix} 
       f(\xi^0,\xi^1) & f(\xi^1,\xi^1) \\
       f(\xi^1,\xi^1) & f(\xi^2,\xi^1)   
       \end{bmatrix}
       =
       \begin{bmatrix} 
       a \xi^{\alpha_2}+b \xi^{\beta_2}  &  a \xi^{\alpha_1+\alpha_2}+b \xi^{\beta_1+\beta_2}\\ 
        a \xi^{\alpha_1+\alpha_2}+b \xi^{\beta_1+\beta_2}  &  a \xi^{2\alpha_1+\beta_2}+b \xi^{2\beta_1+\beta_2}
       \end{bmatrix}.$$
It follows that the multiplication matrices are:
$$ M_1 = \left(H_0^{\Gamma}\right)^{-1}H_1^{\Gamma}=
\begin{bmatrix} 
      0  &  - \xi^{\alpha_1+\beta_1}\\ 
       1 &  \xi^{\alpha_1}+\xi^{\beta_1}
       \end{bmatrix}
      \hbox{ and }   
      M_2 = \left(H_0^{\Gamma}\right)^{-1}H_2^{\Gamma}=
\begin{bmatrix} 
      \ds \frac{\xi^{\alpha_1+\beta_2}-\xi^{\alpha_2+\beta_1}}{\xi^{\alpha_1}-\xi^{\beta_1}} 
                & \ds  - \xi^{\alpha_1+\beta_1} \frac{\xi^{\alpha_2}-\xi^{\beta_2}}{\xi^{\alpha_1}-\xi^{\beta_1}} \\ 
      \ds  \frac{\xi^{\alpha_2}-\xi^{\beta_2}}{\xi^{\alpha_1}-\xi^{\beta_1}}   & \ds
        \frac{\xi^{\alpha_1+\alpha_2}-\xi^{\beta_1+\beta_2}}{\xi^{\alpha_1}-\xi^{\beta_1}}  
       \end{bmatrix}.$$
 
 The matrix of common left eigenvectors of $M_1$ and $M_2$,  with only $1$ in the first column, is 
 \(\ds W = \begin{bmatrix}1 & \xi^{\alpha_1} \\ 1 & \xi^{\beta_1}\end{bmatrix}.\)
 The diagonal matrices of eigenvalues are 
 \(\ds D_1 =\diag(\xi^{\alpha_1},\xi^{\beta_1} ) \) and \(\ds D_2 =\diag(\xi^{\alpha_2},\xi^{\beta_2} ). \) 
We shall thus output the points $\tr{\left[\xi^{\alpha_1},\xi^{\alpha_2}\right]}$ 
     and $\tr{\left[\xi^{\beta_1},\xi^{\beta_2}\right]}$ of $\K^2$.

The first row of $H_0^{\Gamma}$ is $\begin{bmatrix} a & b \end{bmatrix} W $ so that the vector of coefficients 
$\begin{bmatrix} a & b \end{bmatrix}$ can be retrieved by solving the related linear system.
\end{example}


\subsection{The case of $\chi$-invariant linear forms} \label{hankelinv}

We now consider $\SKx=\IKx$ and $\Skep=\IKx[\chi]$ where $\gva$ 
is a Weyl group  acting on $\Kx$ according to~\Ref{WeylLaurentAction}.
The group morphism $\chi:\gva \rightarrow \{1,-1\}$ is  given by either
  $\chi(A)=1$ or $\chi(A)=\det(A)$. 
$\IKx[\chi]$ is a free $\IKx$-module of rank one. 
When $\chi(A)=\det(A)$ 
a basis for it is $\aorb{\delta}$ (\prpr{iraty}). 
We may  write $\IKx[\chi]=\ba \IKx$ 
where $\ba$ can be $1$ or $\aorb{\delta}$.

\subsubsection{Restriction to the invariant ring}

The starting point is a linear form $\lifo$ on $\Kx$ that is 
$\chi$-invariant, \textit{i.e.} 
$\lifo(A \cdot p)=\chi(A) \, \lifo(p)$ for all $A\in\gva$ and $p\in\Kx$.
We show how the restricted Hankel operator 
\[ \Hanc : \IKx \rightarrow \IKx[\chi]\]
allows one  to recover the orbits in the support of the $\chi$-invariant form 
$$\lifo=\sum_{i=1}^r a_i \sum_{A\in\gva} \chi(A)\, \eval[A\star\zeta_i]$$
where the  $\zeta_i\in(\Ks)^n$ have distinct orbits.
By that we mean that we shall retrieve the values of the invariant 
polynomials $\orb{\omega_1},\ldots,\orb{\omega_n}$ 
on $\left\{\zeta_1,\ldots,\, \zeta_n\right\}$.



The  linear map
\[ \reynold: \begin{array}[t]{ccc} \Kx& \rightarrow & \SIKx \\ 
    q & \mapsto & \ds\frac{1}{|\gva|}\sum_{A\in \gva} \chi(A)^{-1}\,A\cdot q \end{array}\]
is a  projection that satisfies 
\begin{itemize}
\item $\reynold(p\,q)=p\,\reynold(q)$ for
all $p\in \IKx$, $q\in\Kx$, and 
\item  $\reynold( A\cdot  q)= \chi(A)\, \reynold(q)$ for all $q\in\Kx$.
\end{itemize}
Then, for any $\chi$-invariant form $\lifo$ we have 
$\ds \lifo(p) = \lifo\left(\reynold(p)\right)$. Hence
$\lifo$ is fully determined by its restriction to $\SIKx$.
We shall write $\lifo^{\gva}$ when we consider the restriction of
$\lifo$ to $\SIKx$.
Similarly,  we denote $\IHanc$ and $\Ikerk$ the Hankel
operator associated to $\lifo^\gva$ and its kernel. 
Hence $\IHanc:\IKx\rightarrow\SIKxd$ and $\Ikerk$ is an ideal
of $\IKx$.

\begin{lemma} \label{bokchoi} 
If 
$\ds \lifo = \sum_{i=1}^r a_i \sum_{A\in\gva} \chi(A)\,  \eval[A\star\zeta_i]$
then  
$\ds \Ikerk = \kerk \cap \IKx$ and the dimension of $\IKx/\Ikerk$ is $r$.
\end{lemma}
\begin{proof}
Take $p\in \Ikerk \subset \IKx$. 
One wishes to show that for any $q\in\Kx$ we have $\lifo(p\,q)= 0$. 
This is true because
$\lifo(p\,q)=\lifo(\reynold(p\,q))=\lifo(p\,\reynold(q))$ and $\lifo(p\,q')=0$ for any
$q'\in\IKx$. Hence $\Ikerk \subset \kerk\cap\IKx$. The other inclusion is obvious.

The proof that $\dim \IKx/\Ikerk=r$ follows the structure of \cite[Ch.2,Proposition 2.10]{Cox05}. Let $Z = \{ A\star \zeta_i \ | \ A \in \gva, \; 1\leq i\leq r\}$ be the union of the orbits of the $\zeta_i$. According to \cite[Lemma 2.9]{Cox05}, for each $i$ there exists a polynomial $\tilde{p}_i$ such that for $z \in Z$
$$\tilde{p}_i(z)=\left\{\begin{array}{ll}
    1 & \hbox{ if }  z = \zeta_i \\
    0 & \hbox{ otherwise.  } \end{array}
   \right.$$
Let $\gva_i$ be the stabilizer of $\zeta_i$. 
 Note that for $A \in \gva$, $\tilde{p}_i (A\star \zeta_j) = 0$ 
 if $j\neq i$ and  $\tilde{p}_i (A\star \zeta_i) = 1$ 
 if and only if $A \in \gva_i$ .  
 Define $p_i = \frac{|\gva|}{|\gva_i |}\,\reynold[1](\tilde{p}_i)$.     
We have $p_i\in \IKx$ and $p_i(\zeta_j)=\delta_{i,j}$.
Hence the linear map
\[ \begin{array}{cccl} 
         \phi~: & \IKx & \rightarrow &  \K^r \\
            & q & \mapsto & \begin{bmatrix} q(\zeta_1) & \ldots & q(\zeta_r) \end{bmatrix}
\end{array}\]
is onto. We proceed to determine its kernel.

One easily sees that  $\kerk\cap\IKx  \subset \ker \phi$.
Consider $q\in \ker \phi$.
Since $q$ is invariant $q(A\star \zeta_i)=q(\zeta_i)=0$ for all $A\in\gva$.
By \thmr{mungo}, $\kerk$ is the annihilating ideal of 
$\left\{ A\textasteriskcentered \zeta_i \,|\, 1\leq i\leq r,\; A\in \gva/\gva_{\zeta_i}\right\}$.
Hence $q\in \kerk\cap\IKx$. 
Since $\Ikerk= \kerk\cap\IKx$, we have proved that $\ker \phi=\Ikerk$.
Hence $\IKx/\Ikerk$ is isomorphic to $\K^r$.
\end{proof}

\begin{theorem} 
If $\ds \lifo = \sum_{i=1}^r a_i \sum_{A\in\gva} \chi(A)\, \eval[A\star\zeta_i]$, 
where $a_i\in \Ks$ and $\zeta_1,\,\ldots,\,\zeta_r $ have distinct orbits 
in $(\Ks)^n$, 
then the Hankel operator $\IHanc$ associated to $\lifo^{\gva}$ is of
rank $r$.
The variety of the extension of $\Ikerk$  to $\Kx$ is 
 $\left\{ A\star\zeta_i \;|\;  A\in\gva,\; 1\leq  i\leq r\right\}$.
\end{theorem}

\begin{proof} Follows from \lemr{bokchoi} and \thmr{mungo}. 
\end{proof}

\subsubsection{Determining a basis of the quotient algebra}

$\IKx$ is isomorphic to a polynomial ring $\KX=\K[X_1,\ldots,X_n]$ 
(\prpr{prop:invariants}).
Then \prpr{bounty} implies the following.

\begin{corollary} \label{chufa}
Let $I^\gva$ be an ideal of $\IKx$ such that the dimension of $\IKx/I^\gva$ 
is of dimension $r$ as a $\K$-linear space.
There exists a lower set $\Gamma$ of cardinal $r$ such that 
$\left\{\orb{\alpha} \;|\; \alpha \in \Gamma\right\}$ 
and $\left\{\cha{\alpha} \;|\; \alpha \in \Gamma\right\}$ 
are both bases of $\IKx/I^\gva$.
\end{corollary}

\begin{proof}
The Chebyshev polynomials of the first and second kind,
$\left\{T_\alpha\right\}_\alpha$ and $\left\{U_\alpha\right\}_\alpha$, 
were defined in \dfnr{cheby1} and \ref{cheby2} as the (only) polynomials in 
$\K[X]=\K[X_1,\ldots,X_n]$ such that
$T_{\alpha}(\orb{\omega_1},\ldots, \orb{\omega_2})=\orb{\alpha}$
and $U_{\alpha}(\orb{\omega_1},\ldots, \orb{\omega_2})=\cha{\alpha}$.

Consider $J$ the ideal in $\K[X]$ that corresponds to $I^\gva$ through the isomorphism 
between $\IKx$ and $\K[X]$. Then $\dim_\K \K[X]/J =r$.
With the order $\leq$ on $\N^n$ defined in \prpr{termorder}, 
and by \prpr{gooseberry},
$\left\{T_\alpha\right\}_\alpha$ and $\left\{U_\alpha\right\}_\alpha$ 
satisfy the hypothesis of \prpr{bounty}.
Hence there is a lower set $\Gamma$ of cardinality $r$ s.t. 
$\left\{T_{\alpha} \;|\; \alpha \in \Gamma\right\}$ 
and $\left\{U_{\alpha} \;|\; \alpha \in \Gamma\right\}$ 
are both bases of $\K [ X ]/J$.
This particular $\Gamma$ provides the announced  conclusion 
through the isomorphism between $\IKx$ and $\K[X]$.
\end{proof}

\begin{proposition} \label{sprouts}
Assume $\lifo$ is a $\chi$-invariant linear form on $\Kx$ whose restricted 
Hankel operator $\IHanc:\IKx \rightarrow \IKx[\chi]$ is of rank $r$.
Then there is a non singular principal submatrix of  size $r$
in 
\[H_1 = 
\left\{\begin{array}{ll} \ds 
\left[ \lifo\left(\orb{\alpha}\orb{\beta}\right)\right]_{\alpha,\beta\,\in\,\hcross{r}}
 & \hbox{ if } \chi=1, \\ \\ 
\ds \left[ \lifo\left(\aorb{\delta+\alpha}\orb{\beta}\right)\right]_{\alpha,\beta\,\in\,\hcross{r}}
 & \hbox{ if } \chi=\det . \\
\end{array}\right. \]
Let  $\Gamma$ be the index set of such a submatrix.
Then 
$\left\{\orb{\alpha} \,|\, \alpha\in \Gamma\right\}$ 
is a basis of $\IKx/\Ikerk$ considered as a $\K$-linear space.
Furthermore one can always find such a $\Gamma$ that is a lower set.
\end{proposition}

\begin{proof}
When  $\IHanc$ is of rank $r$, $r$ is the dimension of $\IKx/\Ikerk$.
By \colr{chufa}, applied to $\Ikerk$, there is a lower set $\Gamma$ of size $r$
such that  $\left\{\orb{\alpha} \;|\; \alpha \in \Gamma\right\}$ 
and $\left\{\cha{\alpha} \;|\; \alpha \in \Gamma\right\}$ 
are both bases of $\IKx/\Ikerk$.

As $\Gamma \subset \hcross{r}$ \cite[Lemma 10]{Sauer16a}, by \thmr{hankbasis}, both
\[\left[ \lifo\left(\ba\orb{\alpha}\,\orb{\beta} \right)\right]_{\alpha,\beta\,\in\,\Gamma} 
\hbox{ and }
\left[ \lifo\left(\ba\cha{\alpha}\,\orb{\beta} \right)\right]_{\alpha,\beta\,\in\,\Gamma}.
\]
are non-singular.
When $\chi=1$ then $\ba =1$ and the left hanside matrix above is a submatrix of $H_1$. 
When $\chi=\det$ then 
$\ba =\aorb{\delta}$ and $\aorb{\delta}\cha{\alpha} =\aorb{\delta+\alpha}$ by \thmr{weyl}.
Hence the right handside matrix above is a submatrix of $H_1$.
\end{proof}

Introducing the matrices $A=\diag(a_1,\ldots, a_r)$,
\[ W_\zeta^{\orb{}} = \left[\orb{\alpha}(\zeta_i)\right]_{\substack{1\leq i\leq r,\, \alpha\in\Gamma }}
\quad \hbox{ and }\quad
W_\zeta^{\aorb{}} = \left[\aorb{\delta+\alpha}(\zeta_i)\right]_{\substack{1\leq i\leq r,\,\alpha\in\Gamma }},\]
one observes that
\begin{equation} \label{invfacto}
\left[ \lifo\left(\orb{\alpha}\orb{\beta}\right)\right]_{\alpha,\beta\,\in\,\Gamma} = 
\tr{\left( W_\zeta^{\orb{}}\right)} A \, W_\zeta^{\orb{}}
\quad \hbox{ and }\quad
\left[ \lifo\left(\aorb{\delta+\alpha}\orb{\beta}\right)\right]_{\alpha,\beta\,\in\,\Gamma}
=
\tr{\left(W_\zeta^{\aorb{}}\right)} A\, W_\zeta^{\orb{}}
 \end{equation}
according to whether $\chi=1$ or $\det$.

\subsubsection{Multiplication maps}

\begin{proposition} \label{kohlrabi} 
Assume that the ideal $I^{\gva}$ of $\IKx$ is
radical with $\IKx/I^\gva$ of dimension $r$.
Consider $\zeta_1,\, \ldots, \zeta_r$  in $(\bKs)^n$  
whose distinct orbits  form the  variety of $I^{\gva}\!\cdot \!\Kx$. Then
\begin{itemize}
\item A set  $B=\left\{b_1,\ldots,b_r\right\}$ is a basis of $\IKx/I^{\gva}$ if and only if the matrix 
 $W_\zeta^B = \left( b_j(\zeta_i) \right)_{1 \leq i,j \leq r}$ is non singular;
\item The matrix $M_p^B$ of the multiplication $\Multi_p$ by $p\in\IKx$ in
  a basis $B$ of $\IKx/I^{\gva}$ satisfies
$W_\zeta^B \, M_p^B =D_\zeta^p\,W_\zeta^B$ where $D_\zeta^p$ is the diagonal matrix 
$\diag ( p(\zeta_1), \,\ldots,\, p(\zeta_r) )$.
\end{itemize}
\end{proposition}

\begin{proof} 
Clearly if $B$ is linearly dependent modulo $\I^{\gva}$ then $\det W_\zeta^B=0$.

Assume  $B=\left\{b_1,\,\ldots,\,b_r\right\}$ is a basis of $\IKx/I^{\gva}$. 
For any $q\in\IKx$ there thus exist  unique $(q_1,\,\ldots,\, q_r)\in\K^r$
such that $q\equiv  q_1\,b_1+\ldots+q_r\,b_r \mod I^{\gva}$.
Observe that
\[ W_\zeta^B \, \begin{bmatrix} q_1 \\ \vdots \\ q_r \end{bmatrix}
=  \begin{bmatrix} q(\zeta_1) \\ \vdots \\ q(\zeta_r)\end{bmatrix} \]
Thus 
\[ W_\zeta^B \, M_p^B  \begin{bmatrix} q_1 \\ \vdots \\ q_r \end{bmatrix}
   = \begin{bmatrix} p(\zeta_1)\,  q(\zeta_1) \\ \vdots \\p(\zeta_r)\,
     q(\zeta_r)\end{bmatrix}
= D_\zeta^p \, \begin{bmatrix} q(\zeta_1) \\ \vdots \\ q(\zeta_r)\end{bmatrix} 
= D_\zeta^p \, W_\zeta^B \, \begin{bmatrix} q_1 \\ \vdots \\ q_r \end{bmatrix}.
\]
This thus shows the equality $W_\zeta^B \, M_p^B =D_\zeta^p\,W_\zeta^B$, for
all $p\in\IKx$. 
This latter equality means that the $i^{th}$ 
row of $W_\zeta^B$ is a left
eigenvector of $M_p^B$ associated to the eigenvalue $p(\zeta_i)$.
If we choose $q\in \IKx$ so that it 
separates the orbits of zeros of $I^{\gva}$, 
then the  left eigenvectors associated to the $r$ distinct eigenvalues $q(\zeta_i)$
are linearly independent. Those are nonzero multiples of the rows of
$W_\zeta^B$.  Therefore $\det W_\zeta^B\neq 0$.
\end{proof}

\subsubsection{Algorithm and examples} 

Let  $\lifo = \sum_{i=1}^r a_i \sum_{A\in\gva} \chi(A)\,\eval[A\star\zeta_i]$
for $\zeta_1,\ldots,\zeta_r$ with distinct orbits. 
{The underlying ideas of the following algorithm are similar to those of \algr{support}.}

\begin{algoy} \label{invsupport} 
{\Algf Invariant Support \& Coefficients}

\In{ $r \in \N_{>0}$ and  
\begin{itemize}
\item    $\left\{\,\lifo\left(\orb{\alpha}\,\orb{\beta}\,\orb{\gamma}\right) \;|\; 
            \alpha, \beta \in \hcross{r},\, |\gamma|\leq 1\right\}$ if  $\chi=1$
\item $\left\{\,\lifo\left(\aorb{\delta+\alpha}\,\orb{\beta}\,\orb{\gamma}\right) \;|\; 
            \alpha, \beta \in \hcross{r},\, |\gamma|\leq 1\right\}$
           if $\chi=\det$.
\end{itemize}
    }

\Out{
\begin{itemize}
\item The vectors
    $ \left[\orb{\omega_1}(\zeta_i),\,\ldots,\orb{\omega_n}(\zeta_i) \right]$ 
    for $1\leq i\leq r$, 
    where $\omega_i = (0, \ldots , 0,1, 0,\ldots , 0)$ is the $i^{th}$ fundamental weight. 
\item The row vector $\tilde{\mathrm{a}}=\begin{bmatrix}\tilde{a}_1 & \ldots & \tilde{a}_r\end{bmatrix}$ such that
    \begin{itemize}
      \item $\tilde{\mathrm{a}}=\begin{bmatrix}|\gva|\,a_1 & \ldots & |\gva|\,a_r\end{bmatrix}$ when $\chi=1$
      \item $\tilde{\mathrm{a}}=\begin{bmatrix}\aorb{\delta}(\zeta_1)\,a_1 & \ldots & \aorb{\delta}(\zeta_r)\,a_r\end{bmatrix}$ when $\chi=\det$
    \end{itemize}
\end{itemize}
such that 
$\ds \lifo = \sum_{i=1}^r a_i \sum_{A\in\gva} \chi(A)\,\eval[A\star\zeta_i]$}

\begin{italg} 
  \item Form the matrix 
         $H_0^{\hcross{r}}=\left[\Omega\left(\orb{\alpha}\,\orb{\beta}\right) \right]_{ 
            \alpha, \beta \,\in\, \hcross{r} } $
            or $\left[\Omega\left(\aorb{\delta+\alpha}\,\orb{\beta}\right) \right]_{ 
            \alpha, \beta \,\in\, \hcross{r} } $
            according to whether $\chi$ is $1$ or $\det$.

  \item Determine a {lower} set $\Gamma$  of cardinal $r$ 
        such that the principal submatrix $H_0^\Gamma$ 
          indexed by $\Gamma$ is nonsingular.

      {\small \emph{
       \hspace*{\stretch{3}}  \% 
         $\Gamma =\left\{0,\,\gamma_2,\,\ldots,\, \gamma_r\right\}$ and
       the subset $\left\{\orb{\gamma} \,|\, \gamma\in \Gamma\right\}$ 
       is a basis of $\IKx/\Ikerk$. 
       }}

  \item For $1\leq j \leq n$, form the matrices 
        \begin{itemize} 
        \item[-] $H_{j}^\Gamma =
            \left[\Omega\left(\orb{\alpha}\,\orb{\beta}\, \orb{\omega_j}\right)
             \right]_{\alpha, \beta \,\in\, \Gamma }$
          or $\left[\Omega\left(\aorb{\delta+\alpha}\,\orb{\beta}\, \orb{\omega_j}\right)
             \right]_{\alpha, \beta \,\in\, \Gamma }$ 
              according to whether $\chi$ is $1$ or $\det$; 
        \item[-] the matrices
         $M_{j}^\Gamma = \left(H_0^\Gamma\right)^{-1} H_{j}^\Gamma$;
         \end{itemize}

       {\small  \emph{\hspace*{\stretch{3}} \%
        $M_j^\Gamma$ is the matrix of multiplication by 
          $\orb{\omega_i}$ 
            in $\IKx/\Ikerk$ (\thmr{multmat})
         \\[0pt] \hspace*{\stretch{3}} 
         \%  
        The matrices $M_{1}^\Gamma,\,\ldots,\,M_{n}^\Gamma$ 
        are simultaneously diagonalisable (\thmr{kohlrabi}).
        } }

  \item Consider $L = \ell_1 M_{1}+\ldots+\ell_n M_{n} $ 
         a \emph{generic} linear combination of $M_{1},\,\ldots,\,M_{n}$ 
        
         {\small{\emph{\hspace*{\stretch{3}} 
          \% The eigenvalues of $L$ are 
          $ \lambda_i=\sum_{j=1}^n \ell_j \orb{\omega_j}(\zeta_i)$, 
          for $1\leq i\leq r$.  }
          \\[0pt]
         \emph{\hspace*{\stretch{3}} \% For most 
            $\left[{\ell}_1,\,\ldots,\,{\ell}_n\right]\in\K^n$ 
           these eigenvalues are distinct.
           }}}

  \item Compute $W$ a matrix whose rows are the $r$ linearly independent 
        normalized left eigenvectors of $L$.

       { \small{\emph{ \hspace*{\stretch{3}} 
        \% A left eigenvector associated to  $\lambda_i$ is a scalar multiple of
                     $\left[ \orb{\gamma}(\zeta_i) \;|\; \gamma\in \Gamma \right]$.
       \\[0pt] \hspace*{\stretch{3}} 
      \% Since $\orb{0}(\zeta_i)=|\gva|$ the eigenvectors can be rescaled so that they are exactly
         $\left[ \orb{\gamma}(\zeta_i) \;|\; \gamma\in \Gamma \right]$.}}}

\item For $1\leq j\leq n$, determine the matrix 
              $D^{(j)}=\diag(\orb{\omega_j}(\zeta_1),\,\ldots,\,\orb{\omega_j}(\zeta_r))$ 
              s.t.
              $W M_{j}^\Gamma = D^{(j)} W$. 
  \item From the diagonal entries of the matrices $D^{(j)}$ 
  form the vectors   
     $ \begin{bmatrix}
      \orb{\omega_1}(\zeta_i) & \ldots & 
       \orb{\omega_r}(\zeta_i)\end{bmatrix}$ 
       for $1\leq i \leq r$           

  { \small{ \emph{\hspace*{\stretch{3}} 
   \% If $\left\{\omega_1,\,\ldots,\omega_n\right\}\subset\Gamma$,
      we can form the vectors 
           $\begin{bmatrix} 
            \orb{\omega_1}(\zeta_i) & \ldots & 
                  \orb{\omega_r}(\zeta_i)\end{bmatrix}$
        directly from the entries of $W$. }}}



\item Take $\mathrm{h}$ to be the  first row of $H_0^\Gamma$ and solve the linear system 
      $\;\tilde{\mathrm{a}} \, W=  \mathrm{h} \;$ for the row vector 
      $\tilde{\mathrm{a}}=\begin{bmatrix}\tilde{a}_1 & \ldots & \tilde{a}_r\end{bmatrix}$.

  \small{ \emph{
   \hspace*{2mm} 
   \% From Equation \Ref{invfacto} 
      $H_0^\Gamma = \tr{W}\, \diag(a_1,\ldots,a_r)\, W$ if $\chi =1$ and  
      $H_0^\Gamma = \tr{\left(W_\zeta^{\aorb{}}\right)}\diag(a_1,\ldots,a_r)\, W$ if $\chi =\det$.
    \\[0pt] \hspace*{5mm}
    \% The first row of this equality is $ \tilde{\mathrm{a}} \, W=  \mathrm{h} $ where
     \\[0pt] \hspace*{5mm}
     \% \hspace*{7mm}
      $\tilde{\mathrm{a}}= \begin{bmatrix}\orb{0}(\zeta_1)a_1 & \ldots & \orb{0}(\zeta_r)a_r\end{bmatrix} $, when $\chi=1$,
     and
     $\tilde{\mathrm{a}}=\begin{bmatrix}\aorb{\delta}(\zeta_1)a_1 & \ldots & \aorb{\delta}(\zeta_r)a_r\end{bmatrix}$, when $\chi=\det$.
         }}






     
\end{italg}
\end{algoy}

\algr{invsupport} is called within \algr{FKInterp} and \ref{SKInterp}. 
At the next step of these algorithms one computes 
$T_{\mu}\left(\orb{\omega_1}(\zeta_i), \ldots,       \orb{\omega_r}(\zeta_i)\right)$ 
for $\mu$ runing through a set of $n$ linearly independent strongly dominant weights. 
We have that 
$$T_{\mu}\left(\orb{\omega_1}(\zeta_i), \ldots,       \orb{\omega_r}(\zeta_i)\right)
= \orb{\mu}(\zeta_i).$$
Hence if $\Gamma$ includes some strongly dominant weights 
the entries of the related row of $W$ could be output to save on these evaluations.

\begin{example} In \exmr{spinterp1} we called on \algr{invsupport} with $r=2$ and 
$\lifo\left(\orb{{\gamma_1},{\gamma_2}}\right)=
f\left(\xi^{\frac{2}{3}\,\gamma_{{1}}+\frac{1}{3}\,\gamma_{{2}}},
       \xi^{\frac{1}{3}\,\gamma_{{1}}+\frac{2}{3}\,\gamma_{{2}}}\right)$
where
$f(x,y)= F\left(\orb{\omega_1}(x,y),\orb{\omega_2}(x,y)\right) 
   = a\,\orb{\alpha}(x,y)+b\,\orb{\beta}(x,y).$

The underlying ideas of \algr{invsupport} follow these of \algr{support} which was fully illustrated in \exmr{hankex}. The same level of details would be very cumbersome in the present case and probably not enlightening. 
We shall limit ourselves to illustrate the formation of the matrices 
$H_0^{\hcross[2]{2}}$, $H_0^{\Gamma}$, $H_1^{\Gamma}$ and $H_2^{\Gamma}$ in terms of evaluation of the function to interpolate and make explicit the matrix $W$ to be computed.

We first need to consider the matrix $H_0$ indexed by 
$\hcross[2]{2}=\left\{ \tr{\begin{bmatrix} 0 & 0 \end{bmatrix}},
\tr{\begin{bmatrix} 1 & 0 \end{bmatrix}},\tr{\begin{bmatrix} 0 & 1 \end{bmatrix}}\right\}$
 \begin{eqnarray*}
 H_0^{\hcross[2]{2}} & = &  \begin{bmatrix}
\lifo \left( \orb{0,0}^2\right) & \lifo \left( \orb{0,0} \orb{1,0}\right) & \lifo \left( \orb{0,0} \orb{0,1}\right) \\
\lifo \left( \orb{1,0} \orb{0,0}\right) & \lifo \left( \orb{1,0}^2 \right) & \lifo \left( \orb{1,0} \orb{0,1}\right) \\
\lifo \left( \orb{0,1} \orb{0,0}\right) & \lifo \left( \orb{0,1} \orb{1,0}\right) & \lifo \left(  \orb{0,1}^2\right) \\
 \end{bmatrix}
 =\begin{bmatrix} 
 6\, \lifo\left( \orb{0,0} \right) 
 & 6\,\lifo\left( \orb{1,0} \right) 
 &6\,\lifo\left( \orb{ 0,1} \right) \\ 
6\,\lifo\left( \orb{1,0}\right) 
&2\,\lifo\left( \orb{2,0} \right) +4\,\lifo\left( \orb{0,1 }\right) 
&4\,\lifo\left(  \orb{1,1} \right) +2\,\lifo\left( \orb{0,0}\right) 
 \\ 
 6\,\lifo\left(  \orb{0,1} \right) 
 & 4\,\lifo\left(  \orb{1,1} \right) +2\,\lifo\left( \orb{0,0} \right) 
& 2\,\lifo\left( \orb{0,2} \right) +4\,\lifo\left( \orb{1,0} \right) 
\end{bmatrix}
 \\
& = & 
\begin{bmatrix}  
6\,f \left( 1,1 \right) 
&6\,f \left( {\xi}^{2/3},{\xi}^{1/3}\right) 
&6\,f \left( {\xi}^{1/3},{\xi}^{2/3}\right) 
\\ 
6\,f \left( {\xi}^{2/3},{\xi}^{1/3} \right) 
&2\,f \left( {\xi}^{4/3},{\xi}^{2/3} \right) +4\,f\left( {\xi}^{1/3},{\xi}^{2/3}\right) 
&4\,f \left( \xi,\xi \right) +2\,f\left( 1,1 \right) 
\\ 
 6\,f \left({\xi}^{1/3},{\xi}^{2/3} \right) 
&4\,f \left( \xi,\xi \right) +2\,f \left( 1,1\right) 
&2\,f \left( {\xi}^{2/3},{\xi}^{4/3} \right) +4\,f \left( {\xi}^{2/3},{\xi}^{1/3} 
\right) 
\end{bmatrix}
 \end{eqnarray*}
 One can check that this matrix has determinant zero whatever $\alpha$ and $\beta$.
 The possible lower sets $\Gamma$ of cardinality $2$ are
 $\left\{ \tr{\begin{bmatrix} 0 & 0 \end{bmatrix}},
\tr{\begin{bmatrix} 1 & 0 \end{bmatrix}}\right\}$
or $\left\{ \tr{\begin{bmatrix} 0 & 0 \end{bmatrix}},
\tr{\begin{bmatrix} 0 & 1 \end{bmatrix}}\right\}$.
One can actually check that the respective determinants of 
the associated principal submatrices are 
$$
12\,f \left( 1,1 \right) \left( f \left( {\xi}^{4/3},{\xi}^{2/3} \right) +2
\, f \left( \xi^{1/3},{\xi}^{2/3} \right) \right) 
-36\, \left( f \left( {\xi}^{2/3},\xi^{1/3} \right)  \right) ^{2}
=36\,ab\,(\orb{\alpha}(\xi^{2/3},\xi^{1/3})-\orb{\beta}(\xi^{2/3},\xi^{1/3}))^2 
 $$
and 
$$
12\,f \left( 1,1 \right) \left(f \left( {\xi}^{2/3},{\xi}^{4/3} \right) +2
\,f \left( {\xi}^{2/3},\xi^{1/3} \right)\right) -36
\, \left( f \left( \xi^{1/3},{\xi}^{2/3} \right)  \right) ^{2}=
36\,ab\,(\orb{\alpha}(\xi^{1/3},\xi^{2/3})-\orb{\beta}(\xi^{1/3},\xi^{2/3}))^2
.$$
At least one of these is non zero if $\alpha$ and $\beta$ have distinct orbits. 
Assume it is the former, so that we choose $\Gamma=\left\{ \tr{\begin{bmatrix} 0 & 0 \end{bmatrix}},
\tr{\begin{bmatrix} 1 & 0 \end{bmatrix}}\right\}$.
Then 
\begin{eqnarray*}
H_0^\Gamma = \begin{bmatrix}  
6\,f \left( 1,1 \right) 
&6\,f \left( {\xi}^{2/3},{\xi}^{1/3}\right) 
\\ 
6\,f \left( {\xi}^{2/3},{\xi}^{1/3} \right) 
&2\,f \left( {\xi}^{4/3},{\xi}^{2/3} \right)+4\,f\left( {\xi}^{1/3},{\xi}^{2/3}\right) 
\end{bmatrix},
\end{eqnarray*}

\begin{eqnarray*}
 H_1^{\Gamma} & = & 
  \begin{bmatrix}
\lifo \left( \orb{0,0}^2\orb{1,0}\right) & \lifo \left( \orb{0,0} \orb{1,0}^2\right)  \\
\lifo \left( \orb{1,0}^2 \orb{0,0}\right) & \lifo \left( \orb{1,0}^3 \right) 
 \end{bmatrix}
 =
\begin{bmatrix}
 36\,\lifo \left( \orb{1,0} \right) 
 &12\,\lifo \left( \orb{2,0}\right) +24\,\lifo \left( \orb{0,1} \right) 
 \\ 
 12\,\lifo \left( \orb{ 2,0 }\right) +24\,\lifo \left( \orb{0,1} \right) 
 &24\,\lifo \left( \orb{1,1} \right) +8\,\lifo \left( \orb{0,0} \right) 
   +4\,\lifo \left( \orb{3,0} \right) 
  \end{bmatrix}
\\
 & = &
 \begin{bmatrix}
 36\,f \left( {\xi}^{2/3},{\xi}^{1/3}\right) 
 &24\,f \left( {\xi}^{1/3},{\xi}^{2/3} \right) +12\,f
 \left( {\xi}^{4/3},{\xi}^{2/3} \right) 
 \\ 
 24\,f \left( {\xi}^{1/3},{\xi}^{2/3} \right) +12\,f \left( {\xi}^{4/3},{\xi}^{2/3} \right) 
&8\,f \left( 1,1 \right) +24\,f \left( \xi,\xi
 \right) +4\,f \left( {\xi}^{2},\xi \right) 
 \end{bmatrix},
 \end{eqnarray*}

\begin{eqnarray*}
 H_2^{\Gamma} & = &  
 \begin{bmatrix}
\lifo \left( \orb{{0,0}}^2\orb{{0,1}}\right) 
& \lifo \left( \orb{{0,0}} \orb{{1,0}}\orb{{0,1}}\right)  
\\
\lifo \left( \orb{{1,0}} \orb{{0,0}}\orb{{0,1}}\right) 
& \lifo \left( \orb{{1,0}}^2 \orb{{0,1}}\right) 
\end{bmatrix}
=
\begin{bmatrix}
36\,\lifo \left( \orb{0,1} \right) 
&24\,\lifo \left( \orb{1,1}\right) +12\,\lifo \left( \orb{0,0} \right) 
 \\ 
 24\,\lifo \left( \orb{1,1} \right) +12\,\lifo \left( \orb{0,0} \right) 
&8\,\lifo \left( \orb{0,2} \right) +20\,\lifo \left( \orb{1,0} \right) +8\,\lifo \left( \orb{2,1} \right)
 \end{bmatrix}
\\
& = &  
\begin{bmatrix}
36\,f \left( {\xi}^{1/3},{\xi}^{2/3} \right) 
&12\,f \left( 1,1 \right) +24\,f \left( \xi,\xi \right) 
\\ 
12\,f \left( 1,1 \right) +24\,f \left( \xi,\xi \right) 
 &8\,f \left( {\xi}^{5/3},{\xi}^{4/3} \right) 
   +8\,f \left( {\xi}^{2/3},{\xi}^{4/3} \right)+20\,f \left( {\xi}^{2/3},{\xi}^{1/3} \right) 
\end{bmatrix}.
\end{eqnarray*}

The matrix of left eigenvectors common to 
$M_1=\left(H_0^{\Gamma}\right)^{-1} H_1^\Gamma$ 
and $M_2=\left(H_0^{\Gamma}\right)^{-1} H_2^\Gamma$  
to be computed is 
$$ W = \begin{bmatrix} 
\orb{0,0}(\xi^{\tr{\alpha}S}) & \orb{1,0}(\xi^{\tr{\alpha}S})
\\
\orb{0,0}(\xi^{\tr{\beta}S}) & \orb{1,0}(\xi^{\tr{\beta}S})
\end{bmatrix}
=\begin{bmatrix} 
6 & \orb{1,0}(\xi^{\tr{\alpha}S})
\\
6 & \orb{1,0}(\xi^{\tr{\beta}S})
\end{bmatrix}
=
\begin{bmatrix} 
6 & \orb{\alpha}(\xi^{2/3},\xi^{1/3})
\\
6 & \orb{\beta}(\xi^{2/3},\xi^{1/3})
\end{bmatrix}.
$$
We have 
$W M_1 = \diag\left(\orb{1,0}(\xi^{\tr{\alpha}S}),\; \orb{1,0}(\xi^{\tr{\beta}S}) \right)\, W$
and 
$W M_2 = \diag\left(\orb{0,1}(\xi^{\tr{\alpha}S}),\; \orb{0,1}(\xi^{\tr{\beta}S})\right) \, W$
so that the points
\[ \vartheta_\alpha = \tr{\begin{bmatrix}\orb{1,0}(\xi^{\tr{\alpha}S}) & \orb{0,1}(\xi^{\tr{\alpha}S}) \end{bmatrix}} \quad \hbox{ and }
\vartheta_\beta = \tr{\begin{bmatrix}\orb{1,0}(\xi^{\tr{\beta}S}) & 
\orb{0,1}(\xi^{\tr{\beta}S}) \end{bmatrix}}
\]
can be output. 
We know that $H_0^{\Gamma}= \tr{W}\, \diag( a, \; b)\, W$. Extracting the first rows of this equality
provides the linear system 
\[\begin{bmatrix} 6\, a & 6\, b \end{bmatrix} \,W = \begin{bmatrix}  
6\,f \left( 1,1 \right) 
&6\,f \left( {\xi}^{2/3},{\xi}^{1/3}\right) 
\end{bmatrix} \]
to be solved in order to provide the second component of the output.
\end{example}

\begin{example} In \exmr{spinterp2} we called on \algr{invsupport} with $r=2$ and 
$\lifo\left(\aorb{{\gamma_1},{\gamma_2}}\right)=
f\left(\xi^{\frac{2}{3}\,\gamma_{{1}}+\frac{1}{3}\,\gamma_{{2}}},
       \xi^{\frac{1}{3}\,\gamma_{{1}}+\frac{2}{3}\,\gamma_{{2}}}\right)$
where
$f(x,y)= \aorb{\delta}(x,y)\, F\left(\orb{\omega_1}(x,y),\orb{\omega_2}(x,y)\right) 
   = a\,\aorb{\delta+\alpha}(x,y)+b\,\aorb{\delta+\beta}(x,y).$

As in previous example, we illustrate the formation of the matrices 
$H_0^{\hcross[2]{2}}$, $H_0^{\Gamma}$, $H_1^{\Gamma}$ and $H_2^{\Gamma}$ in terms of evaluation of the function to interpolate and make explicit the matrix $W$ to be computed.

We first need to consider the matrix $H_0$ indexed by 
$\hcross[2]{2}=\left\{ \tr{\begin{bmatrix} 0 & 0 \end{bmatrix}},
\tr{\begin{bmatrix} 1 & 0 \end{bmatrix}},\tr{\begin{bmatrix} 0 & 1 \end{bmatrix}}\right\}$
 \begin{eqnarray*}
 H_0^{\hcross[2]{2}} & = &  \begin{bmatrix}
\lifo \left( \aorb{1,1} \orb{0,0}\right) & \lifo \left( \aorb{1,1} \orb{1,0}\right) & \lifo \left( \aorb{1,1} \orb{0,1}\right) \\
\lifo \left( \aorb{2,1} \orb{0,0}\right) & \lifo \left( \aorb{2,1}\orb{1,0} \right) & \lifo \left( \aorb{2,1} \orb{0,1}\right) \\
\lifo \left( \aorb{1,2} \orb{0,0}\right) & \lifo \left( \aorb{1,2} \orb{1,0}\right) & \lifo \left( \aorb{1,2} \orb{0,1}\right) \\
 \end{bmatrix}
 =\begin{bmatrix} 
6\,\aorb{{1,1}}
& 2\,\aorb{{2,1}}
& 2\,\aorb{{1,2}}
\\ 
6\,\aorb{{2,1}}
&2\,\aorb{{3,1}}+2\,\aorb{{1,2}}
&2\,\aorb{{2,2}}+2\,\aorb{{1,1}}
\\ \noalign{\medskip}
6\,\aorb{{1,2}}
&2\,\aorb{{2,2}}+2\,\aorb{{1,1}}
&2\,\aorb{{1,3}}+2\,\aorb{{2,1}}
\end{bmatrix}
 \\
& = & 
\begin{bmatrix}  
6\,f \left( \xi,\xi \right) 
& 2\,f \left( {\xi}^{5/3},{\xi}^{4/3} \right) 
& 2\,f \left( {\xi}^{4/3},{\xi}^{5/3}\right) 
\\ 
6\,f \left( {\xi}^{5/3},{\xi}^{4/3} \right) 
& 2\,f \left( {\xi}^{7/3},{\xi}^{5/3} \right) +2\,f \left( {\xi}^{4/3},{\xi}^{5/3} \right) 
& 2\,f \left( {\xi}^{2},{\xi}^{2}\right) +2\,f \left( \xi,\xi \right) 
\\ 
6\,f \left( {\xi}^{4/3},{\xi}^{5/3} \right) 
& 2\,f \left( {\xi}^{2},{\xi}^{2} \right) +2\,f \left( \xi,\xi \right) 
& 2\,f \left( {\xi}^{5/3},{\xi}^{7/3} \right) +2\,f \left( {\xi}^{5/3},{\xi}^{4/3} \right) 
\end{bmatrix}
 \end{eqnarray*}
 One can check that this matrix has determinant zero whatever $\alpha$ and $\beta$.
 The possible lower sets $\Gamma$ of cardinality $2$ are
 $\left\{ \tr{\begin{bmatrix} 0 & 0 \end{bmatrix}},
\tr{\begin{bmatrix} 1 & 0 \end{bmatrix}}\right\}$
or $\left\{ \tr{\begin{bmatrix} 0 & 0 \end{bmatrix}},
\tr{\begin{bmatrix} 0 & 1 \end{bmatrix}}\right\}$.
One can actually check that the respective determinants of 
the associated principal submatrices are 
$$\begin{array}{l}
12\,f \left( \xi,\xi \right) \left(
f \left( {\xi}^{7/3},{\xi}^{5/3} \right) + f \left( {\xi}^{4/3},{\xi}^{5/3}\right) \right) 
-12\,  f \left( {\xi}^{5/3},{\xi}^{4/3} \right)  ^{2}
\\ =
6\,ab\,
\left(
\orb{\delta+\beta}(\xi^{5/3},\xi^{4/3})-\orb{\delta+\alpha}(\xi^{5/3},\xi^{4/3})
\right)
\left(
\aorb{\delta+\alpha}(\xi,\xi)\aorb{\delta+\beta}(\xi^{5/3},\xi^{4/3})
-\aorb{\delta+\beta}(\xi,\xi)\aorb{\delta+\alpha}(\xi^{5/3},\xi^{4/3})
\right)
 \end{array}$$
and 
$$\begin{array}{l}
12\,f \left( \xi,\xi \right) \left( 
 f \left( {\xi}^{5/3},{\xi}^{7/3} \right) + f \left( {\xi}^{5/3},{\xi}^{4/3}\right)\right)
  -12\, f \left( {\xi}^{4/3},{\xi}^{5/3} \right)^{2}
\\ =
6\,ab\,
\left(
\orb{\delta+\beta}(\xi^{4/3},\xi^{5/3})-\orb{\delta+\alpha}(\xi^{4/3},\xi^{5/3})
\right)
\left(
\aorb{\delta+\alpha}(\xi,\xi)\aorb{\delta+\beta}(\xi^{4/3},\xi^{5/3})
-\aorb{\delta+\beta}(\xi,\xi)\aorb{\delta+\alpha}(\xi^{4/3},\xi^{5/3})
\right).
 \end{array}$$
At least one of these is non zero. 
Assume  the former is and choose $\Gamma=\left\{ \tr{\begin{bmatrix} 0 & 0 \end{bmatrix}},
\tr{\begin{bmatrix} 1 & 0 \end{bmatrix}}\right\}$.
Then 
\begin{eqnarray*}
H_0^\Gamma = \begin{bmatrix}  
6\,f \left( \xi,\xi \right) 
& 2\,f \left( {\xi}^{5/3},{\xi}^{4/3} \right) 
\\ 
6\,f \left( {\xi}^{5/3},{\xi}^{4/3} \right) 
& 2\,f \left( {\xi}^{7/3},{\xi}^{5/3} \right) +2\,f \left( {\xi}^{4/3},{\xi}^{5/3} \right) 
\end{bmatrix},
\end{eqnarray*}

\begin{eqnarray*}
 H_1^{\Gamma} & = & 
  \begin{bmatrix}
\lifo \left( \aorb{1,1} \orb{0,0}\orb{1,0}\right) 
    & \lifo \left( \aorb{1,1} \orb{1,0}^2\right)  \\
\lifo \left( \aorb{2,1} \orb{0,0}\orb{1,0}\right) 
    & \lifo \left( \aorb{2,1}\orb{1,0}^2 \right) 
 \end{bmatrix}
 =
\begin{bmatrix}
 12\,\aorb{{2,1}}
 & 4\,\aorb{{1,2}}+4\,\aorb{{3,1}}
 \\ 
 12\,\aorb{{1,2}}+12\,\aorb{{3,1}}
 & 4\,\aorb{{4,1}}+4\,\aorb{{1,1}}+8\,\aorb{{2,2}}
  \end{bmatrix}
\\
 & = &
 \begin{bmatrix}
 12\,f \left( {\xi}^{5/3},{\xi}^{4/3}
 \right) &4\,f \left( {\xi}^{4/3},{\xi}^{5/3} \right) +4\,f \left( {
\xi}^{7/3},{\xi}^{5/3} \right) \\ \noalign{\medskip}12\,f \left( {\xi}
^{4/3},{\xi}^{5/3} \right) +12\,f \left( {\xi}^{7/3},{\xi}^{5/3}
 \right) &8\,f \left( {\xi}^{2},{\xi}^{2} \right) +4\,f \left( \xi,\xi
 \right) +4\,f \left( {\xi}^{3},{\xi}^{2} \right) 
 \end{bmatrix}, 
 \end{eqnarray*}

\begin{eqnarray*}
 H_2^{\Gamma} & = &  
 \begin{bmatrix}
\lifo \left( \aorb{1,1} \orb{0,0}\orb{0,1}\right) 
    & \lifo \left( \aorb{1,1} \orb{1,0}\orb{0,1}\right)  \\
\lifo \left( \aorb{2,1} \orb{0,0}\orb{0,1}\right) 
    & \lifo \left( \aorb{2,1}\orb{1,0}\orb{0,1} \right) 
\end{bmatrix}
=
\begin{bmatrix}
12\,\aorb{{1,2}}
& 4\,\aorb{{1,1}}+4\,\aorb{{2,2}}
\\ 
12\,\aorb{{1,1}}+12\,\aorb{{2,2}}
&4\,\aorb{{3,2}}+4\,\aorb{{1,3}}+8\,\aorb{{2,1}}
 \end{bmatrix}
\\
& = &  
\begin{bmatrix}
12\,f \left( {\xi}^{4/3},{\xi}^{5/3}
 \right) &4\,f \left( {\xi}^{2},{\xi}^{2} \right) +4\,f \left( \xi,\xi
 \right) \\ \noalign{\medskip}12\,f \left( {\xi}^{2},{\xi}^{2}
 \right) +12\,f \left( \xi,\xi \right) &4\,f \left( {\xi}^{8/3},{\xi}^
{7/3} \right) +4\,f \left( {\xi}^{5/3},{\xi}^{7/3} \right) +8\,f
 \left( {\xi}^{5/3},{\xi}^{4/3} \right)
\end{bmatrix}.
\end{eqnarray*}

The matrix of left eigenvectors common to 
$M_1=\left(H_0^{\Gamma}\right)^{-1} H_1^\Gamma$ 
and $M_2=\left(H_0^{\Gamma}\right)^{-1} H_2^\Gamma$  
to be computed is 
$$ W = \begin{bmatrix} 
\orb{0,0}(\xi^{\tr{(\delta+\alpha)}S}) & \orb{1,0}(\xi^{\tr{(\delta+\alpha)}S})
\\
\orb{0,0}(\xi^{\tr{(\delta+\beta)}S}) & \orb{1,0}(\xi^{\tr{(\delta+\beta)}S})
\end{bmatrix}
=\begin{bmatrix} 
6 & \orb{1,0}(\xi^{\tr{(\delta+\alpha)}S})
\\
6 & \orb{1,0}(\xi^{\tr{(\delta+\beta)}S})
\end{bmatrix}
=
\begin{bmatrix} 
6 & \orb{\delta+\alpha}(\xi^{2/3},\xi^{1/3})
\\
6 & \orb{\delta+\beta}(\xi^{2/3},\xi^{1/3})
\end{bmatrix}.
$$
We have 
$W M_1 = \diag\left(\orb{1,0}(\xi^{\tr{(\delta+\alpha)}S}),\; \orb{1,0}(\xi^{\tr{(\delta+\beta)}S}) \right)\, W$
and 
$W M_2 = \diag\left(\orb{0,1}(\xi^{\tr{(\delta+\alpha)}S}),\; \orb{0,1}(\xi^{\tr{(\delta+\beta)}S})\right) \, W$
so that the points
\[ \vartheta_\alpha = \tr{\begin{bmatrix}\orb{1,0}(\xi^{\tr{(\delta+\alpha)}S}) & \orb{0,1}(\xi^{\tr{(\delta+\alpha)}S}) \end{bmatrix}} \quad \hbox{ and }
\vartheta_\beta = \tr{\begin{bmatrix}\orb{1,0}(\xi^{\tr{(\delta+\beta)}S}) & 
\orb{0,1}(\xi^{\tr{(\delta+\beta)}S}) \end{bmatrix}}
\]
can be output. 
We know that $H_0^{\Gamma}= \tr{\widehat{W}}\, \diag( a, \; b)\, W$
where 
$$ \widehat{W} = \begin{bmatrix} 
\aorb{1,1}(\xi^{\tr{(\delta+\alpha)}S}) & \aorb{2,1}(\xi^{\tr{(\delta+\alpha)}S})
\\
\aorb{1,1}(\xi^{\tr{(\delta+\beta)}S}) & \aorb{2,1}(\xi^{\tr{(\delta+\beta)}S})
\end{bmatrix}
=\begin{bmatrix} 
\aorb{\delta+\alpha}(\xi,\xi) & \aorb{\delta+\alpha}(\xi^{5/3},\xi^{4/3})
\\
\aorb{\delta+\beta}(\xi,\xi) & \aorb{\delta+\beta}(\xi^{5/3},\xi^{4/3})
\end{bmatrix}. $$
Extracting the first rows of this equality
provides the linear system 
\[\begin{bmatrix} 
  \aorb{\delta+\alpha}(\xi,\xi)\, a 
  & \aorb{\delta+\beta}(\xi,\xi)\, b 
\end{bmatrix} \,W 
= \begin{bmatrix}  
6\,f \left( \xi,\xi \right) 
&2\,f \left( {\xi}^{5/3},{\xi}^{4/3}\right) 
\end{bmatrix} \]
to be solved in order to provide the second component of the output, namely 
$\begin{bmatrix} 
  \aorb{\delta+\alpha}(\xi,\xi)\, a 
  & \aorb{\delta+\beta}(\xi,\xi)\, b 
\end{bmatrix}$.
\end{example}
\color{black}

\newpage
\section{Final Comments}
\label{final}
For the benefit of clarity
we have decribed the algorithms for sparse interpolation, be it in terms of Laurent monomials or generalized Chebyshev polynomials, in two separate phases~: in \secr{interpolation} we basically massaged the sparse interpolation problem into the recovery of the support of the linear form and offered to perform there all the evaluations of the functions that may be needed to cover all the possible cases. 
Once we examine the algorithms
to recover the support of the linear forms, in \secr{hankel}, it becomes apparent 
that not all these evaluations are used.
First, as commented upon after \algr{support} determining the lower set $\Gamma$ of the appropriate cardinality $r$ can be approached 
iteratively and should not require forming the whole matrix $H_0^{\hcross{r}}$.
Then only the evaluations indexed by $\Gamma+\Gamma+\hcross{2}$
(rather than $\hcross{r}+\hcross{r}+\hcross{2}$)
are required to form the subsequent matrices.
It is thus clear that going further with our intrinsically mutivariate approach 
to sparse interpolation  needs a holistic approach. 

All along the article we have mostly worked under the assumption that we know the number $r$ of summands exactly. 
Much of the litterature  on sparse interpolation 
 considers an upper bound $R$ to the number of summands.
It is not a theoretical difficulty. The algorithms work similarly with $R$ instead of $r$ as input.
The exact number of summands can then be retrieved as the rank of the matrix $H_0^{\hcross{R}}$.
This would indicate that, in this case where we only know an upper bound, 
we actually need to form the whole matrix $H_0^{\hcross{R}}$ first. 
But the practical approach to sparse interpolation is to 
design early termination strategies that provide probabilistic 
certificate on the actual number of summands   \cite{Kaltofen03,Kaltofen00,Huang18}.
Such strategies would deserve an extension to the generalized Chebyshev polynomials considered here.




As noted in \secr{sec:sparse}, one can consider an $r$-sparse sum of generalized Chebyshev polynomials as a $\tilde{r}$-sparse sum of monomials where $\tilde{r}$ is bounded by $r|\gva|$. Yet the approach we presented for 
$r$-sparse sum of generalized Chebyshev polynomials allows to restrict the size of matrices 
to $|\hcross{r}|$ instead of $|\hcross{|\gva|r}|$. 
Our initial hope was to have an analogous benefit, by a factor $|\gva|$, on the number of evaluations. 
The number of evaluations needed for the sparse interpolation of 
a sum of $r|\gva|$-monomials, is bounded by the cardinality of 
$\hcross{|\gva|r}+\hcross{|\gva|r}+\hcross{2}$. 
We nonetheless bounded the number of evaluations to be made by the cardinality of $\wcross{r}$,
which is only a superset of $\hcross{r}+\hcross{r}+\hcross{2}$.
Our initial estimate of the cardinality of  $\wcross{r}$
still shows a benefit of our approach also in terms of the number of evaluations.
Yet we feel that a more refined analysis, taking into account 
the specific properties of the different  Weyl groups, 
would testify to a stronger benefit.

In our generalized approach to sparse interpolation the emphasis is on 
the associated Hankel operator rather than the 
matrices that arose when laying down the problem as a set of linear equations.
In \cite{Ben-Or88,Lakshman95}  the structure  of these matrices is exploited 
to work out the best complexity of the linear algebra used in the algorithm for the univariate cases. 
One has to recognize that 
it is the multiplication rules on the polynomial basis 
(monomial or Chebyshev respectively)
that gives the specific structure to the matrix of the Hankel operator. 
A deeper understanding of how the action of the Weyl group can be used to express these multiplication rules in the most economical form should lead to a better control of the complexity of our approach.

\newpage
\bibliographystyle{plain}
\bibliography{aaaaa}



\end{document}